\documentclass[11pt]{article}
\usepackage{url}
\usepackage{mathtools}
\usepackage{amssymb}
\usepackage{amsthm}
\usepackage{empheq}
\usepackage{latexsym}
\usepackage{enumitem}
\usepackage[unicode]{hyperref}
\usepackage{eurosym}
\usepackage{dsfont}
\usepackage{appendix}
\usepackage{color} 
\usepackage[utf8]{inputenc}
\usepackage[T1]{fontenc}
\usepackage{multirow}
\usepackage{todonotes}
\usepackage{comment} 
\DeclareUnicodeCharacter{014D}{\=o}
\definecolor{red}{rgb}{0.7,0.15,0.15}
\definecolor{green}{rgb}{0,0.5,0}
\definecolor{blue}{rgb}{0,0,0.7}
\hypersetup{colorlinks, linkcolor={red},citecolor={green}, urlcolor={blue}}

\usepackage{geometry} 
\numberwithin{equation}{section}

\renewenvironment{thebibliography}[1]{%
\begin{oldthebibliography}{#1}%
\setlength{\baselineskip}{.8em}
\linespread{.8}
\small 
\setlength{\parskip}{0ex}%
\setlength{\itemsep}{.1em}%
}%
{%
\end{oldthebibliography}%
}

\theoremstyle{plain}
\newtheorem{theorem}{Theorem}[section]
\newtheorem{proposition}[theorem]{Proposition}
\newtheorem{lemma}[theorem]{Lemma}
\newtheorem{corollary}[theorem]{Corollary}
\newtheorem{definition}[theorem]{Definition}
\newtheorem{assumption}[theorem]{Assumption}

\theoremstyle{definition}
\newtheorem{remark}[theorem]{Remark}

\begin{document}

\title{Equilibrium Asset Pricing with Transaction Costs\footnote{The authors are grateful to Lukas Gonon, Dmitry Kramkov, and Xiaofei Shi for for fruitful discussions and pertinent remarks. The helpful suggestions of two anonymous referees are also gratefully acknowledged.}} 

\author{Martin Herdegen\thanks{University of Warwick, Department of Statistics, Coventry, CV4 7AL, UK, email \texttt{m.herdegen@warwick.ac.uk}.}
\and
Johannes Muhle-Karbe\thanks{Imperial College London, Department of Mathematics, London, SW7 1NE, UK, email \texttt{j.muhle-karbe@imperial.ac.uk}. Research supported by the CFM--Imperial institute of Quantitative Finance. Parts of this paper were written while this author was visiting ETH Z\"urich; he thanks the Forschungsinstitut f\"ur Mathematik and H.~Mete Soner for their hospitality.}
\and
Dylan Possama\"i\thanks{ETH Z\"urich, Department of Mathematics, R\"amistrasse 101, Z\"urich 8092, Switzerland, \texttt{dylan.possamai@math.ethz.ch}.}
}

\date{September 30, 2020}

\maketitle

\begin{abstract}
We study risk-sharing economies where heterogeneous agents trade subject to quadratic transaction costs. The corresponding equilibrium asset prices and trading strategies are characterised by a system of nonlinear, fully coupled forward-backward stochastic differential equations. We show that a unique solution exists provided that the agents' preferences are sufficiently similar. In a benchmark specification with linear state dynamics, the illiquidity discounts and liquidity premia observed empirically correspond to a positive relationship between transaction costs and volatility.
\end{abstract}




\section{Introduction}

How does the introduction of a transaction tax affect the volatility of a financial market? Such questions about the interplay of liquidity and asset prices need to be tackled with equilibrium models, where prices are not exogenous inputs but determined endogenously by matching supply and demand. However, equilibrium analyses lead to notoriously intractable feedback loops. Indeed, if the optimal strategies for a given candidate price do not clear the market, then the price needs to adjust until this iteration converges. Trading costs compound these difficulties, because they severely complicate the corresponding optimisation problems. Accordingly, the literature on equilibrium asset prices with transaction costs has focused either on numerical methods (Heaton and Lucas \cite{heaton.lucas.96}; Buss, Dumas, Uppal, and Vilkov~\cite{buss.al.16}; Buss and Dumas \cite{buss.dumas.17}) or on models where the market volatility is either zero (Vayanos and Vila \cite{vayanos.vila.99}; Lo, Mamaysky, and Wang~\cite{lo.al.04}, Weston \cite{weston.17}) or given exogenously (Vayanos \cite{vayanos.98}; Garleanu and Pedersen~\cite{garleanu.pedersen.16}; Sannikov and Skrzypacz \cite{ sannikov.skrzypacz.16}; Bouchard, Herdegen, Fukasawa, and Muhle-Karbe~\cite{bouchard.al.17}). 

\medskip

In the present study, we analyse a risk-sharing equilibrium where price levels, expected returns \emph{and} volatilities are determined endogenously, by both balancing supply and demand and matching an exogenous terminal condition for the risky asset. We consider two agents with mean--variance preferences who trade a safe and a risky asset to hedge the fluctuations of their random endowment streams. We show that a unique equilibrium with transaction costs generally exists provided the agents' risk aversions are sufficiently similar. To explore the impact of transaction costs on equilibrium asset prices and trading volume, we then consider a benchmark example with linear state dynamics. Without transaction costs, the corresponding equilibrium equilibrium price has Bachelier dynamics, i.e., constant expected returns and volatilities. With transaction costs, equilibrium prices show a much richer behaviour already in this simple setting.

First, expected returns endogenously become mean-reverting due to the sluggishness of the trading process, as is assumed in many reduced-form models for active portfolio management (cf., e.g., Kim and Omberg~\cite{kim.omberg.96}; De Lataillade, Deremble, Potters, and Bouchaud~\cite{bouchaud.al.12}; Martin \cite{martin.12}; Garleanu and Pedersen \cite{garleanu.pedersen.13}).  These random and time varying ``liquidity premia'' are zero on average if the asset volatility is specified endogenously (cf.~Sannikov and Skrzypacz \cite{sannikov.skrzypacz.16}, Bouchard et al. \cite{bouchard.al.17}). In contrast, the present model with endogenous volatility can produce the systematically positive liquidity premia that have been documented empirically (Amihud and Mendelson~\cite{amihud.mendelson.86a}, Brennan and Subrahmanyam \cite{brennan.subrahmanyam.96}, Pastor and Stambaugh \cite{pastor.stambaugh.03}). Our general equilibrium results in turn complement a large partial-equilibrium literature on liquidity premia going back to Constantinides \cite{constantinides.86}; see Lynch and Tan \cite{lynch.tan.11} and the references therein for an overview. 

Second, trading volume is finite in the frictional equilibrium and approximately follows Ornstein-Uhlenbeck dynamics as in the partial-equilibrium model of~Guasoni and Weber \cite{guasoni.weber.15}. Thus, our model combines rather realistic price dynamics with the main stylized features of trading volume observed empirically, such as mean-reversion and autocorrelation. In contrast, Lo et al. \cite{lo.al.04} obtain a similar dynamic model for trading volume but the corresponding asset prices are constant; conversely, the equilibrium prices in Yayanos \cite{vayanos.98} are diffusive but accompanied by deterministic trading patterns. 

Third, our model with endogenous price volatility allows to study how the latter is affected by the trading costs. For agents with similar risk aversions, we obtain explicit asymptotic formulas that reveal close connections between the effects of transaction costs on expected returns and volatilities. To wit, the liquidity premia that distinguish frictional expected returns from their frictionless counterparts, and the adjustment of the corresponding volatilities always have the same sign in our model, determined by the difference of the agents' risk aversion parameters. In the empirically relevant case of positive liquidity premia, our model predicts a positive relation between transaction costs and volatility, corroborating empirical evidence of~Umlauf \cite{umlauf.93}, Jones and Seguin \cite{jones.seguin.97} and Hau \cite{hau.06}, numerical results of Adam, Beutel, Marcet, and Merkel~\cite{adam.al.15} as well as~Buss et al. \cite{buss.al.16}, and findings in a risk-neutral model with asymmetric information by Danilova and Julliard~\cite{danilova.julliard.19}. In our model, this empirically relevant regime obtains when agents whose frictionless trading targets increase with positive price shocks (``trend followers'') have a larger risk aversion (and, in turn, stronger motive to trade) than the ``contrarians'' whose trading targets decrease with price shocks.

\medskip

On a technical level, our general existence and uniqueness results are based on new well-posedness results for fully coupled systems of nonlinear forward-backward stochastic differential equations (FBSDEs). Without transaction costs, the equilibrium dynamics of the risky asset are determined by a scalar purely quadratic BSDE in our model,  which leads to explicit formulas in concrete examples. With quadratic transaction costs on the agents' trading rates, we show that the corresponding equilibria are characterised by fully coupled systems of FBSDEs. To wit, the optimal risky positions evolve forward from the given initial allocations. In contrast, the corresponding trading rates controlling these positions need to be determined from their zero terminal values -- near the terminal time, trading stops since additional trades can no longer earn back the costs that would need to be paid to implement them. If a constant volatility is given exogenously as in Bouchard et al.~\cite{bouchard.al.17}, then these forward-backward dynamics suffice to pin down the equilibrium returns. In this case, the FBSDEs are linear, and therefore can be solved explicitly in terms of Riccati equations and conditional expectations of the endowment processes (cf.~Garleanu and Pedersen \cite{garleanu.pedersen.16}; Bank, Soner, and Vo\ss~ \cite{bank.al.17}). In the present context, where the volatility is determined endogenously from the terminal condition for the risky asset, the corresponding FBSDEs are coupled to an additional backward equation arising from this extra constraint. Due to the quadratic preferences and trading costs, the resulting forward-backward system is still linear in the trading rates and positions. However, it also depends quadratically on the volatility of the risky asset, which is now no longer an exogenous constant but needs to be determined as part of the solution. 

Accordingly, explicit solutions are no longer possible and existence and uniqueness are beyond the scope of the extant literature. Indeed, there is no general well-posedness theory for fully coupled systems of FBSDEs. In fact, even for linear equations, one can obtain either well-posedness, or infinitely many solutions, or no solutions at all, see the example in the introduction of Ma, Wu, Zhang, and Zhang~\cite{ma2015well}. Under a variety of additional monotonicity, non-degeneracy, Lipschitz assumptions, or for scalar forward and backward components, well-posedness results have been obtained, cf., e.g., Ma et al.~\cite{ma2015well} and the references therein for an overview. However, none of these results are applicable to our fully-coupled system, which is not Lipschitz and has a bivariate backward component.

  To overcome these difficulties, we focus on the case where both agents' risk aversions are sufficiently similar. If these parameters coincide, then the BSDE for the equilibrium price decouples from the FBSDEs for the optimal position and trading rate, and in fact reduces to its frictionless counterpart. For distinct but similar risk aversions, we in turn establish the existence of a unique solution. Our proof is based on a Picard iteration under smallness conditions inspired by~Tevzadze \cite{tevzadze.08}. However, due to the coupling between forward and backward components, this standard argument only yields existence here if the time horizon is sufficiently short -- a degenerate result in the present context since the cost on the trading rate then essentially imposes a no-trade equilibrium. Proving existence on arbitrary time horizons requires more subtle arguments tailored to the structure of the equations. Here, the key insight is that, for a given volatility process, the FBSDE for the corresponding optimal positions and trading rates can be solved in terms of \emph{stochastic} Riccati equations as in Kohlmann and Tang \cite{kohlmann.tang.02}, Annkirchner and Kruse \cite{annkirchner.kruse.15}, Bank and Voß \cite{bank.voss.18}. We develop a number of novel stability estimates for such equations. These in turn allow us to devise a \emph{one-dimensional} Picard iteration for the equilibrium price process only -- the corresponding positions and trading rates are constructed using the stochastic Riccati equations of Kohlmann and Tang~\cite{kohlmann.tang.02} in each step. If the agents' risk aversions are sufficiently similar, we can in turn establish the existence of a solution, which is unique in a neighbourhood of its frictionless counterpart.

This well-posedness result applies in general settings without requiring a Markovian structure. However, it crucially exploits that all primitives of the model belong to suitable BMO spaces. This assumption ensures that the optimal positions remain uniformly bounded and our BSDEs are of quadratic growth, but rules out concrete specifications based on Brownian motions, for example. However, our approach can be adapted to such settings, e.g., the concrete model with linear state dynamics. To wit, the FBSDEs characterising the equilibrium can then be reduced to a system of four coupled scalar Riccati ODEs. For sufficiently similar risk aversions, existence for this system can in turn be established by adapting our Picard iteration scheme. Again, the key idea is not to work with the full multidimensional system, but instead focus on only one component (the others are in turn constructed from this source term in each step of the iteration). 

The remainder of this article is organised as follows. Section~\ref{sec:model} introduces our model, both in the frictionless baseline version and with quadratic transaction costs on the trading rate. The agents' individual optimisation problems for given price dynamics are then discussed in Section~\ref{s:indopt}. Our main results on equilibrium asset prices without and with transaction costs are subsequently presented in Section~\ref{s:equilibrium}. This is followed by the discussion of the benchmark model with linear state dynamics in Section~\ref{s:linear}. For better readability, all proofs are delegated to Sections \ref{s:proofs}--\ref{sec:proofseq} as well as Appendices \ref{app:A} and \ref{app:B}.

\paragraph*{Notations}

Throughout, we fix a filtered probability space $(\Omega,\mathcal{F},\mathbb F:=(\mathcal{F}_t)_{t \in [0,T]},\mathbb P)$ with finite time horizon $T>0$; the filtration is generated by a standard Brownian motion $(W_t)_{t \in [0,T]}$. For $0\leq s\leq t\leq T$,  the set of $[s,t]$-valued stopping times is denoted by $\mathcal T_{s,t}$; for $\tau\in\mathcal T_{0,T}$, we write $\mathbb E_\tau[\cdot]$ for the $\mathcal{F}_\tau$-conditional expectation. The $\mathbb R$-valued, progressively measurable processes $(X_t)_{t \in [0,T]}$ satisfying $\|X\|_{\mathbb H^p}^p:=\mathbb E[(\int_0^T X_t^2 \mathrm{d}t)^{p/2}]<\infty$ for some $p\in[1,\infty)$ are denoted by $\mathbb{H}^p$. We also write $\mathbb{H}^2_{\mathrm{BMO}}$ for the $\mathbb R$-valued, progressively measurable processes $(X_t)_{t \in [0,T]}$ satisfying 
\[
\|X\|_{\mathbb H^2_{\mathrm{BMO}}}^2:=\bigg\|\sup_{\tau\in\mathcal T_{0,T}}\mathbb E_\tau\bigg[\int_\tau^T X_t^2 \mathrm{d}t\bigg]\bigg\|_{\mathbb L^\infty}<\infty.
\]
Finally, for any $p\in[1,\infty]$, $\mathcal S^p$ denotes the $\mathbb R$-valued, $\mathbb F$-progressively measurable processes $X$ with continuous paths for which $\sup_{0\leq t\leq T}|X_t|$ belongs to $\mathbb L^p$. The associated norm is denoted by $\|\cdot\|_{\mathcal S^p}$. For any other probability measure $\mathbb Q$ on $(\Omega,\mathcal F)$, we define similarly $\mathbb L^p(\mathbb Q)$, $\mathbb H^p(\mathbb Q)$, $\mathbb H^2_{\mathrm{BMO}}(\mathbb Q)$, and $\mathcal S^p(\mathbb Q)$.

\section{Model}\label{sec:model}

\subsection{Financial Market}

We consider a financial market with two assets. The first one is safe, with exogenous price normalised to one. The second one is risky, with price dynamics 
\begin{equation}\label{eq:price}
\mathrm{d}S_t=\mu_t \,\mathrm{d}t+\sigma_t \,\mathrm{d}W_t.
\end{equation}
Here, the initial asset price $S_0 \in \mathbb{R}$ as well as the (progressively measurable) expected returns process $(\mu_t)_{t \in [0,T]}$ and volatility process $(\sigma_t)_{t \in [0,T]}$ are to be determined in equilibrium by matching demand to the (exogenous) \emph{supply} $s \in \mathbb{R}$ of the risky asset. To pin down the equilibrium volatility -- unlike in \cite{vayanos.98,bouchard.al.17}, where this process is an exogenous constant, or in \cite{sannikov.skrzypacz.16}, where a particular value is singled out by focusing on linear equilibria -- the terminal value of the risky asset is as in \cite{grossman.stiglitz.80} also required to match an exogenous $\mathcal{F}_T$-measurable random variable: 
\[S_T=\mathfrak{S}.\]
This can be interpreted as a fundamental liquidation value~\cite{kyle.85}, a terminal dividend~\cite{kramkov.15}, or the payoff of a derivative depending on an exogenous underlying~\cite{cheridito.al.15}.

\subsection{Agents}
The assets are traded by two agents $n=1,2$ with mean--variance preferences over wealth changes as in \cite{kallsen.02,martin.schoeneborn.11,bouchaud.al.12,garleanu.pedersen.13,garleanu.pedersen.16}. The agents have risk aversions $\gamma^n>0$, $n=1,2$, and trade to hedge the fluctuations of their (cumulative) random endowments,\footnote{The impact of trading costs on equilibria where agents trade due to heterogenous beliefs is studied in~\cite{muhlekarbe.al.20}, for example.}
\begin{equation}\label{eq:endow}
\mathrm{d}Y^n_t=\beta^n_t \mathrm{d}W_t, \quad \beta^n \in \mathbb{H}^2.
\end{equation}
Agent $n$'s initial position in the risky asset is fixed throughout and denoted by $x^n$. To clear the market initially, we naturally assume that $x^1+x^2=s$.

\begin{remark}
For notational simplicity, we only model the diffusion part of the (spanned) random endowments. This is without loss of generality since the optimizers of the linear-quadratic goal functionals~\eqref{eq:pref1}, \eqref{eq:pref2} would not depend on an additional finite-variation part (or unspanned endowment shocks) as in \cite{bouchard.al.17}.
\end{remark}

\subsection{Frictionless Trading}

Suppose that $\mu=\sigma\kappa$ where the \emph{market price of risk} $\kappa$ belongs to $\mathbb{H}^2$. Without transaction costs, agents' trading strategies are described by the number $\varphi_t$ of risky shares held at each time $t \in [0,T]$. Taking into account each agent's random endowment, their frictionless wealth dynamics are $\varphi_t\, \mathrm{d}S_t +\mathrm{d}Y^n_t$. For \emph{admissible} strategies $\varphi$ which satisfy $\varphi_0 = x^n$ and $\varphi \sigma \in \mathbb{H}^2$,\footnote{By the Cauchy--Schwarz inequality, we then also have $\varphi\mu \in \mathbb{H}^1$ since $\kappa\in \mathbb{H}^2$.}  the corresponding mean--variance goal functional is 
\begin{align}
J^n(\varphi) &:=\mathbb E\bigg[\int_0^T \bigg(\varphi_t \mathrm{d}S_t +\mathrm{d}Y^n_t -\frac{\gamma^n}{2}\mathrm{d}\Big\langle {\int_0^\cdot}\varphi_s \mathrm{d}S_s +Y^n \Big\rangle_t\bigg)\bigg]\notag\\ 
&\;=\mathbb E\bigg[\int_0^T \bigg(\mu_t\varphi_t-\frac{\gamma^n}{2} \left(\sigma_t\varphi_t+\beta^n_t\right)^2\bigg) \mathrm{d}t \bigg]
\to \max!\label{eq:pref1}
\end{align}
Accordingly, for $\sigma>0$, the process $-\beta^n/\sigma$ can also be interpreted as agent $n$'s target position in the risky asset. Related models where deviations from an exogenous target are directly penalised by an exogenous deterministic weight rather than the infinitesimal variance of the corresponding asset are studied by~\cite{choi.al.18,sannikov.skrzypacz.16}. 

\subsection{Trading with Transaction Costs}

Now suppose as in~\cite{almgren.chriss.01} that an exogenous quadratic transaction cost $\lambda/2>0$ is levied on the turnover rate $\dot{\varphi}_t:=\mathrm{d}\varphi_t /\mathrm{d}t$ of each agent's portfolio. Then, the corresponding position $\varphi$ becomes a state variable that can only be influenced gradually by adjusting the control $\dot{\varphi}$. We focus on \emph{admissible} trading rates $\dot{\varphi} \in \mathbb{H}^2$ for which the corresponding position $\varphi=x^n +\int_0^\cdot \dot{\varphi}_t \mathrm{d}t$ satisfies $\varphi\sigma\in \mathbb{H}^2$, in analogy to the frictionless case. The frictional version of the mean--variance goal functional~\eqref{eq:pref1} is
\begin{equation}\label{eq:pref2}
J^n_\lambda(\dot{\varphi}):=\mathbb E\bigg[\int_0^T \bigg(\mu_t\varphi_t-\frac{\gamma^n}{2} \left(\sigma_t\varphi_t+\beta^n_t\right)^2 -\frac{\lambda}{2}\dot{\varphi}_t^2\bigg) \mathrm{d}t\bigg] \to \max!
\end{equation}
The quadratic transaction costs on the turnover rate in~\eqref{eq:pref2} correspond to execution prices that are shifted linearly both by trade size and speed. Note, however, that each agent's payoff is only affected by their own trading rate. Accordingly, the trading cost should be interpreted as a tax or the fees charged by an exchange rather than as a temporary price impact cost here. 

\begin{remark}
The linear-quadratic goal functional \eqref{eq:pref2} is chosen for tractability. Indeed, the local mean--variance trade-off in~\eqref{eq:pref1}, \eqref{eq:pref2} is a more tractable proxy for preferences described by concave utility functions. The corresponding equilibrium with frictions already leads to a novel coupled FBSDE system that is beyond the solution methods in the extant literature. For more general preferences, the corresponding FBSDE system would involve additional coupled backward and forward components describing the agents' value and wealth processes, respectively. 

The assumption of quadratic rather than proportional costs also simplifies the analysis by ensuring the FBSDE system describing the equilibrium remains linear in the agents positions. Subquadratic trading costs lead to FBSDEs that are nonlinear in the agents' positions~\cite{gonon.al.19}; the limiting case of proportional costs corresponds to an FBSDE system with reflection, as is typical for such singular control problems (compare, e.g., \cite{elie.al.18}).

While these assumptions are made for tractability,  numerical results reported in \cite{gonon.al.19} suggest that the qualitative and quantitative properties of the equilibrium asset prices are surprisingly robust across different specifications of the trading costs (given that their absolute magnitudes are matched appropriately). This is in line with partial-equilibrium results for models with small trading costs~\cite{moreau.al.15}, where the fluctuations of frictional positions around their frictionless counterparts and the welfare effects of small trading costs are governed by the same drivers for different specifications of both preferences and trading costs. An extension of these robustness results to a general-equilibrium context is an important but challenging direction for future research.
\end{remark}

\section{Individual Optimisation}\label{s:indopt}

The first step towards solving for the equilibrium is to determine each agent's individually optimal trading strategy for given asset prices. To this end, fix an initial risky asset price $S_0 \in \mathbb{R}$, an expected return process $(\mu_t)_{t \in [0,T]}$, and a volatility process $(\sigma_t)_{t \in [0,T]}$ for which $\mu=\sigma\kappa$ with a market price of risk $\kappa \in \mathbb{H}^2$. For better readability, all proofs are relegated to Section~\ref{s:proofs}.

\subsection{Frictionless Optimisation}

Agent $n$'s optimiser for the frictionless model~\eqref{eq:pref1} can be computed directly by pointwise optimisation
\begin{equation}\label{eq:indopt}
\hat\varphi^n_t:=\begin{cases}\displaystyle \frac{\mu_t}{\gamma^n\sigma_t^2}-\frac{\beta^n_t}{\sigma_t}, & \sigma_t \neq 0,\\ \displaystyle x^n, & \sigma_t=0, \end{cases} \quad t \in (0,T].
\end{equation}

\begin{remark}
Note that the optimal strategy \eqref{eq:indopt} is not determined uniquely on the set $\{\sigma=0\}$, since these values do not contribute to the payoff~\eqref{eq:pref1}. We therefore choose arbitrary values that ensure market clearing. All subsequent result are independent of this choice.
\end{remark}

\subsection{Optimisation with Transaction Costs}

Unlike its frictionless counterpart, the frictional optimisation problem~\eqref{eq:pref2} is no longer myopic and therefore cannot be solved directly using pointwise optimisation. However, \eqref{eq:pref2} can be rewritten as 
\begin{equation*}
J^n_\lambda(\dot{\varphi})=-\mathbb E\bigg[\int_0^T \bigg(\frac{\gamma^n\sigma_t^2}{2} (\varphi_t-\hat\varphi^n_t)^2 +\frac{\lambda}{2}\dot{\varphi}_t^2\bigg) \mathrm{d}t\bigg]+\mathbb E\bigg[\int_0^T \frac{\gamma^n}{2}\bigg(\big(\sigma_t\hat\varphi_t^n\big)^2-(\beta^n_t)^2\bigg)\mathrm{d}t\bigg].
\end{equation*}
Note that the second expectation on the right-hand side of this decomposition is finite for $\kappa, \beta \in \mathbb{H}^2$. Therefore, maximising the frictional mean--variance functional $J^n_\lambda$ is equivalent to solving a quadratic tracking problem, where the target is the frictionless optimiser~\eqref{eq:indopt}:
\begin{equation}\label{eq:pref3}
\mathbb E\bigg[\int_0^T \left(\frac{\gamma^n\sigma_t^2}{2} (\varphi_t-\hat{\varphi}^n_t)^2 +\frac{\lambda}{2}\dot{\varphi}_t^2\right) \mathrm{d}t\bigg] \to\min!
\end{equation}
Problems of this type have been studied by \cite{kohlmann.tang.02,annkirchner.kruse.15,bank.voss.18}, for example. By strict convexity, each agent's optimal trading rate is characterised by the first-order condition that its G\^ateaux derivative vanishes in all directions \cite[Proposition~II.2.1]{ekeland.temam.99}.  A calculus of variations argument (compare \cite{bank.al.17,bouchard.al.17}) in turn shows that the optimal trading rate $\dot{\varphi}^n_t$ of agent $n$, and the corresponding position $\varphi_t^n$, are characterised by a forward-backward stochastic differential equation (FBSDE)\footnote{Here, the terminal condition for the trading rate is zero, because trades close to the terminal time $T$ can no longer earn back the trading costs that would need to be paid to implement them. More general terminal conditions are studied in~\cite{annkirchner.kruse.15,bank.voss.18}, for example.}:
\begin{alignat}{2}
\mathrm{d}\varphi^n_t &=\dot{\varphi}^n_t \mathrm{d}t, \qquad && \varphi^{n}_0= x^n,  \label{eq:indopt21}\\
\mathrm{d}\dot{\varphi}^n_t &=\frac{\gamma^n\sigma^2_t}{\lambda}(\varphi^n_t-\hat\varphi^n_t)\mathrm{d}t+\dot{Z}^n_t\mathrm{d}W_t, \qquad && \dot{\varphi}^n_T=0. \label{eq:indopt22}
\end{alignat}
Observe that the process $\dot{Z}^n$ needs to be determined as part of the solution here. Unlike for the constant volatilities $\sigma$ considered in \cite{bank.al.17,bouchard.al.17}, this equation cannot be solved by reducing to standard Riccati equations. Instead, a backward \emph{stochastic} Riccati equation (BSRDE) plays a crucial role in the analysis of~\cite{kohlmann.tang.02,annkirchner.kruse.15,bank.voss.18}. It is shown in \cite{kohlmann.tang.02} that for bounded $\sigma$, this equation has a unique solution. A localisation argument shows that  this remains true for $\sigma \in \mathbb{H}^2_{\mathrm{BMO}}$, which will be the natural space for our equilibrium analysis in Section~\ref{s:equilibrium}.

\begin{lemma}\label{lem:BRSDE}
For $\gamma,\lambda>0$ and $\sigma \in \mathbb{H}^2_{\mathrm{BMO}}$, the {\rm BSRDE}
\begin{equation}\label{eq:BRSDE}
c_t =\int_t^T \left(\frac{\gamma}{\lambda}\sigma^2_s-c_s^2\right)\mathrm{d}s-\int_t^TZ^c_s\mathrm{d}W_s, \quad  t\in[0,T],
\end{equation}
has a unique solution $(c,Z) \in \mathcal{S}^{\infty} \times \mathbb{H}^2_{\mathrm{BMO}}$. It satisfies 
\begin{equation}\label{eq:Linftybis}
0 \leq c_t \leq \frac{\gamma}{\lambda} \|\sigma\|_{\mathbb{H}^2_{\mathrm{BMO}}}^2, \quad  t \in [0,T].
\end{equation}
\end{lemma}

With the auxiliary process $c$ at hand, the solution of the FBSDE~(\ref{eq:indopt21}--\ref{eq:indopt22}) characterising the optimal trading rate for the tracking problem~\eqref{eq:pref3}, or equivalently the original mean--variance optimisation~\eqref{eq:pref2}, can in turn be constructed as follows:

\begin{lemma}\label{lem:kt02}
For $\gamma,\lambda>0$ and $\sigma \in \mathbb{H}^2_{\mathrm{BMO}}$, let $c$ be the solution of the corresponding {\rm BSRDE}~\eqref{eq:BRSDE}. For a progressively measurable process $\xi$ satisfying $\sigma \xi \in \mathbb{H}^2$, define
\begin{equation}\label{eq:signal}
\bar{\xi}_t := \frac{\gamma}{\lambda}\mathbb E_t\bigg[\int_t^T \mathrm{e}^{-\int_t^s c_u \mathrm{d}u}\sigma^2_s \xi_s \mathrm{d}s\bigg], \quad t \in [0,T],
\end{equation}
and the linear $($random$)$ {\rm ODE}
\begin{equation}\label{eq:ode}
\dot{\varphi}_t= \bar{\xi}_t -c_t \varphi_t, \quad  t \in [0,T], \quad \varphi_0 = x,
\end{equation}
which has the explicit solution 
\begin{equation}\label{eq:varphi}
\varphi_t=\mathrm{e}^{-\int_0^t c_u \mathrm{d}u}x+\int_0^t \mathrm{e}^{-\int_s^t c_u \mathrm{d}u}\bar{\xi}_s \mathrm{d}s, \quad t \in [0,T].
\end{equation}
Then, for $\gamma=\gamma^n$, $x = x^n$, and $\xi=\hat{\varphi}^n$ from~\eqref{eq:indopt}, the corresponding solution $(\varphi^n,\dot{\varphi}^n)$ is optimal for~\eqref{eq:pref3} or equivalently~\eqref{eq:pref2}. Moreover, if $\sigma |\xi|^{\frac{1}{2}} \in \mathbb{H}^2_{\mathrm{BMO}}$, then $\dot \varphi$ and $\varphi$ are uniformly bounded.
\end{lemma}
For uniformly bounded $\sigma$, this result is proved in \cite{kohlmann.tang.02}. For $\sigma \in \mathbb{H}^2_{\mathrm{BMO}}$, we provide a short self-contained proof in Section~\ref{s:proofs}. As a side product, we obtain that the solution coincides with its counterpart for the time-truncated ``auxiliary problem'' considered by~\cite{bank.voss.18}.

Lemma~\ref{lem:kt02} shows that for $t \in [0,T)$, the optimal strategy with transaction costs trades towards the ``signal process'' $\bar{\xi}_t/c_t$ at a (time-dependent and random) speed $c_t$ determined by the BSRDE~\eqref{eq:BRSDE}.\footnote{In particular, since $\bar{\xi}$ only depends on $\sigma^2\xi$, the optimiser for \eqref{eq:pref2} in independent of the (arbitrary) values chosen for the frictionless optimiser on $\{\sigma=0\}$.} For each agent's individual optimisation problem~\eqref{eq:pref3}, the signal is obtained from the corresponding frictionless optimiser~\eqref{eq:indopt}, by appropriate discounting of its expected future values at a rate also derived from the BSRDE. For our equilibrium analysis in Section~\ref{s:equilibrium}, the same construction will be applied to a different target strategy, see~\eqref{eq:eqbwd}.

\section{Equilibrium}\label{s:equilibrium}

With the characterisation of each agent's individually optimal strategy at hand, we now turn to the determination of the equilibrium asset prices for which the agents' aggregate demand for the risky asset equals its supply~$s$. For better readability, all proofs are deferred to Section~\ref{sec:proofseq}. 

\subsection{Frictionless Equilibrium}\label{sec:eqfrictionless}

We first consider the frictionless case.

\begin{definition}\label{def:equi}
A price process $S$ for the risky asset with initial value $S_0 \in \mathbb{R}$, expected returns $(\mu_t)_{t \in [0,T]}$ and volatility $(\sigma_t)_{t \in [0,T]}$ is called a \emph{(Radner) equilibrium}, if:
\begin{itemize}
\item[$(i)$] $\mu=\sigma\kappa$ for  $\kappa \in \mathbb{H}^2;$
\item[$(ii)$] the terminal condition $S_T=\mathfrak{S}$ is satisfied$;$
\item[$(iii)$] the agents' individual optimisation problems \eqref{eq:pref1} for the given price process $S$ have solutions $\varphi^1$ and $\varphi^2$ that clear the market for the risky asset at all times, $\varphi^1_t+\varphi^2_t =s$, $t \in [0,T]$.
\end{itemize}
\end{definition}

For any equilibrium $(S_0,\mu,\sigma)$, market clearing and the representation~\eqref{eq:indopt} for the agents' individually optimal strategies give
\[
\mu_t= \bar{\gamma}\big(s \sigma^2_t+\sigma_t(\beta^1_t+\beta^2_t)\big), \quad t \in [0,T], \quad \mbox{where } \bar{\gamma}:=\frac{\gamma^1\gamma^2}{\gamma^1+\gamma^2}.
\]
Accordingly, $(S,\sigma)$ solves the following quadratic BSDE:
\begin{equation}\label{eq:bsdebis1}
\mathrm{d}S_t=  \bar{\gamma}\big(s\sigma_t^2+\sigma_t (\beta^1_t+\beta^2_t)\big)\mathrm{d}t+\sigma_t\mathrm{d}W_t, \quad S_T=\mathfrak{S}.\end{equation}
Conversely, the individually optimal strategies~\eqref{eq:indopt} corresponding to the dynamics~\eqref{eq:bsdebis1} are admissible if $\sigma \in \mathbb{H}^2$ and evidently clear the market. Whence, existence and uniqueness of Radner equilibria are equivalent to existence and uniqueness of solutions of the quadratic BSDE~\eqref{eq:bsdebis1}.  Provided that the measure 
\begin{equation}\label{eq:Pbeta}
\mathbb{P}^\beta \sim \mathbb{P}, \mbox{ with density process } Z^\beta :=\mathcal{E}\bigg(-\int_0^\cdot \bar{\gamma}(\beta^1_t+\beta^2_t)\mathrm{d}W_t\bigg)
\end{equation}
is well defined, the BSDE~\eqref{eq:bsdebis1} can be rewritten in terms of the $\mathbb{P}^\beta$-Brownian motion $W^\beta=W-\int_0^\cdot  \bar{\gamma}(\beta^1_t+\beta^2_t) \mathrm{d}t$ as a purely quadratic BSDE,
\begin{equation}\label{eq:bsdebis}
\mathrm{d}S_t=  \bar{\gamma}s\sigma_t^2\mathrm{d}t+\sigma_t\mathrm{d}W^\beta_t, \quad S_T=\mathfrak{S}.
\end{equation}
If in addition the terminal condition $\mathfrak{S}$ is sufficiently integrable, it is well known that \eqref{eq:bsdebis} has an explicit solution in terms of the Laplace transform of $\mathfrak{S}$:
\begin{equation}\label{eq:laplace}
S_t= -\frac{1}{2\bar{\gamma}s} \log \mathbb{E}_t^{\beta}\Big[\mathrm{e}^{-2\bar{\gamma}s\mathfrak{S}}\Big], \quad t \in [0,T].
\end{equation}
To make sure the measure $\mathbb{P}^\beta$ is well defined and verify that~\eqref{eq:laplace} is indeed the unique solution of~\eqref{eq:bsdebis} in a suitable class, we make the following integrability assumption on the aggregate trading target $\beta^1+\beta^2$ and the terminal condition $\mathfrak{S}$:

\begin{assumption}\label{ass:int1}
$\beta^1+\beta^2 \in \mathbb{H}^2_{\mathrm{BMO}}$ and $|\mathfrak{S}|$ has finite exponential moments of all orders.
\end{assumption}

With this integrability assumption (which is for instance satisfied if $\beta^1+\beta^2$ and $\mathfrak{S}$ are uniformly bounded), we obtain the following existence and uniqueness result for the BSDE \eqref{eq:bsdebis1}:
\begin{proposition}
\label{prop:BSDE}
Suppose Assumption~\ref{ass:int1} is satisfied. Then, \eqref{eq:laplace} is the unique solution of~\eqref{eq:bsdebis1} among continuous, progressively measurable processes $S$ for which $(\mathrm{e}^{-2\bar\gamma s S_\tau})_{\tau\in\mathcal T_{0,T}}$ is uniformly $\mathbb{P}^\beta$-integrable. In particular, the price process~\eqref{eq:laplace} is the unique Radner equilibrium in this class.
\end{proposition}

\begin{remark}\label{rem:nonunique}
As already observed in {\rm\cite{delbaen.al.15}}, the class of price processes for which $(\mathrm{e}^{-2\bar\gamma s S_\tau})_{\tau\in\mathcal T_{0,T}}$ is uniformly $\mathbb{P}^\beta$-integrable is the largest possible class for uniqueness. Indeed, if this family is \emph{not} uniformly $\mathbb{P}^\beta$-integrable, then $\mathrm{e}^{-2\bar{\gamma}s S}$ is a strict local $\mathbb{P}^\beta$-martingale by It\^o's formula and the dynamics~\eqref{eq:bsdebis}, and hence a strict $\mathbb{P}^\beta$-supermartingale since it is also positive. As a result, the corresponding price process $S$ is strictly larger than~\eqref{eq:laplace}.  
\end{remark}

The non-uniqueness described in Remark~\ref{rem:nonunique} can only arise for price processes that are unbounded from below. In fact, uniqueness always holds among price processes $S$ which admit an equivalent martingale measure with square-integrable density process $Z$ with respect to $\mathbb{P}^\beta$. Indeed, in view of the dynamics~\eqref{eq:bsdebis}, we necessarily have $Z=\mathcal{E}\big(-\bar{\gamma}s\int_0^\cdot \sigma_t \mathrm{d}W^\beta_t\big)$ and in turn
\begin{align*}
0 \leq \mathrm{e}^{-2\bar{\gamma}s S_\tau}= \mathrm{e}^{-2\bar{\gamma}^2s^2\int_0^\tau \sigma^2_t \mathrm{d}t-2\bar{\gamma}s\int_0^\tau \sigma_t \mathrm{d}W^\beta_t} \leq \mathrm{e}^{-\bar{\gamma}^2s^2\int_0^\tau \sigma^2_t \mathrm{d}t-2\bar{\gamma}s\int_0^\tau \sigma_t \mathrm{d}W^\beta_t}=Z_\tau^2, \quad \mbox{for any $\tau \in\mathcal{T}_{0,T}$.}
\end{align*}
Whence uniform $\mathbb{P}^\beta$-integrability of $(\mathrm{e}^{-2\bar\gamma s S_\tau})_{\tau\in\mathcal T_{0,T}}$ follows from Doob's maximal inequality in this case. If the terminal condition is bounded, uniqueness even holds among all price processes $S$ admitting an equivalent martingale measure,\footnote{Such a notion of uniqueness is used in \cite{kramkov.pulido.16}, for example.} since $S$ is then automatically bounded.

\begin{corollary}\label{cor:BSDE}
Suppose Assumption~\ref{ass:int1} is satisfied and, moreover, $\mathfrak S\in\mathbb L^\infty$. Then, \eqref{eq:laplace} is the unique solution of \eqref{eq:bsdebis} in $\mathcal{S}^\infty \times \mathbb{H}^2_{\mathrm{BMO}}$, and therefore the unique Radner equilibrium among bounded price processes.
\end{corollary}

\subsection{Equilibrium with Transaction Costs}\label{s:eq}

We now turn to the main subject of the present study, equilibria with transaction costs. The notion of equilibrium is the same as in Definition \ref{def:equi}, with the exception that the individual optimisation problems are given by \eqref{eq:pref2} rather than \eqref{eq:pref1}.

\medskip
To clear the market, purchases must equal sales at all times, i.e., all individual trading rates must sum to zero. After summing the backward equations \eqref{eq:indopt22} for both agents' optimal trading rates and using the market clearing condition $\varphi_t^2=s-\varphi_t^1$, this leads to
\[
0 =   \bigg(\frac{\sigma_t}{\lambda}\big(\gamma^1 \beta^1_t + \gamma^2 \beta^2_t\big) +\frac{\sigma_t^2}{\lambda}\bigg(\gamma^2s+(\gamma^1-\gamma^2)\varphi^1_t\bigg)-\frac{2\mu_t}{\lambda} \bigg)\mathrm{d}t+\big(\dot{Z}^1_t+\dot{Z}^2_t\big)\mathrm{d}W_t.
\]
Since any local martingale of finite variation is constant, it follows that
\begin{equation}\label{eq:lipr2}
	\mu_t =\sigma_t\left(\frac{\gamma^1\beta_t^1+\gamma^2\beta_t^2}{2}+\sigma_t\bigg(\frac{\gamma^2s}{2} +\frac{\gamma^1-\gamma^2}{2}\varphi^1_t\bigg)\right), \quad t \in [0,T].
\end{equation} 
Plugging this back into agent 1's individual optimality condition \eqref{eq:indopt22} and recalling the terminal condition $\dot{\varphi}^1_T=0$ as well as the forward equation \eqref{eq:indopt21}, we obtain the following FBSDE:
\begin{alignat}{2}
	\mathrm{d}\varphi^1_t &= \dot{\varphi}^1_t, \qquad && \varphi^1_0=x^1, \label{eq:eqfwd}\\
	\mathrm{d}\dot{\varphi}^1_t &=  \frac{(\gamma^1+\gamma^2)}{2\lambda}\bigg( \frac{\gamma^1\beta^1_t-\gamma^2\beta^2_t}{\gamma^1+\gamma^2}\sigma_t-\frac{\gamma^2 s}{\gamma^1+\gamma^2}\sigma_t^2 +\varphi^1_t \sigma_t^2  \bigg)\mathrm{d}t+\dot{Z}^1_t\mathrm{d}W_t, \qquad && \dot{\varphi}^1_T=0.\label{eq:eqbwd}
\end{alignat}
The corresponding optimal strategy for agent $2$ is determined by market clearing. As in the frictionless case discussed in Section~\ref{sec:eqfrictionless}, the corresponding equilibrium volatility is pinned down by the terminal condition $S_T=\mathfrak{S}$. More specifically, inserting \eqref{eq:lipr2} into \eqref{eq:price}, we obtain the following BSDE, which is coupled to the forward-backward system $(\ref{eq:eqfwd}$--$\ref{eq:eqbwd})$:
\begin{equation}\label{eq:BSDES}
	\mathrm{d}S_t = \left(\frac{\gamma^1-\gamma^2}{2}\varphi^{1}_t\sigma_t^2+\frac{\gamma^2 s}{2}\sigma_t^2+\frac{\gamma^1\beta^1_t+\gamma^2\beta^2_t}{2}\sigma_t\right) \mathrm{d}t +\sigma_t \mathrm{d}W_t, \quad S_T=\mathfrak{S}.
\end{equation}
By reversing these arguments, it is straightforward to verify that sufficiently integrable solutions of the FBSDE~$(\ref{eq:eqfwd}$--$\ref{eq:BSDES})$ indeed identify Radner equilibria with transaction costs (sufficient conditions for the existence of a solution to the FBSDE~$(\ref{eq:eqfwd}$--$\ref{eq:BSDES})$ are provided in Theorem \ref{thm:ex3} below):

\begin{proposition}\label{prop:eq}
	Suppose that there exists a solution of the {\rm FBSDE}~$(\ref{eq:eqfwd}$--$\ref{eq:BSDES})$ with $(\dot{\varphi}^1,\sigma)\in  \mathbb{H}^2\times\mathbb H^2_{\mathrm{BMO}}$. Then, $(S_0,\mu,\sigma)$ with $\mu$ as in \eqref{eq:lipr2} is a Radner equilibrium with transaction costs.
\end{proposition}

Due to the coupling between the forward-backward equations~(\ref{eq:eqfwd}--\ref{eq:BSDES}) a direct existence proof by fixed-point iteration is elusive, unless the time horizon is sufficiently short so that very little trading is possible with costs on the trading rate. Establishing existence for sufficiently small transaction costs is also delicate, since the corresponding trading rates explode, which needs to be handled by a suitable renormalisation. Inspired by~\cite{sannikov.skrzypacz.16}, we therefore focus on a different smallness condition, namely the case where both agents risk aversions are similar, $\gamma^1 \approx \gamma^2$.

For $\gamma^1=\gamma^2$, the BSDE~\eqref{eq:BSDES} for the frictional equilibrium price decouples from $(\ref{eq:eqfwd}$--$\ref{eq:eqbwd})$ and reduces to its frictionless counterpart~\eqref{eq:bsdebis}. Accordingly, for $\gamma^1 \approx \gamma^2$, we expect the frictional equilibrium price $S$ and its volatility $\sigma$ to be close to their frictionless versions $\bar{S}$ and $\bar{\sigma}$, respectively. To make this precise, the frictionless equilibrium volatility $\bar{\sigma}$ and the volatilities $\beta^1, \beta^2$ of the agents' random endowments need to be sufficiently integrable:

\begin{assumption}\label{assump:friction}
$(i)$ the frictionless equilibrium volatility $\bar{\sigma}$ from Proposition~\ref{prop:BSDE} belongs to $\mathbb H^2_{\mathrm{BMO}}$;

\medskip
$(ii)$ $\beta^1, \beta^2 \in \mathbb{H}^2_{\mathrm{BMO}}$, so that we can define the measure 
\[\mathbb Q^\beta \sim \mathbb P \quad \mbox{with density process}\quad \frac{\mathrm{d} \mathbb Q^\beta}{\mathrm{d} \mathbb P} := \mathcal{E}\bigg(-\int_0^\cdot \bigg(\gamma^2 s \bar \sigma_t+ \frac{\gamma^1\beta^1_t+\gamma^2 \beta^2_t}{2}\bigg) \mathrm{d}W_t\bigg)_T;
\]

$(iii)$ for some $p>2$, we have $\mathbb{E}^\mathbb{Q^\beta}\Big[\exp\Big(p \int_0^T\big(\gamma^2 s \bar\sigma_t+ \frac{\gamma^1\beta^1_t+\gamma^2 \beta^2_t}{2}\big)^2 \mathrm{d} t\Big)\Big] < \infty$.
\end{assumption}

We can now formulate our main result. It shows that an equilibrium with transaction costs exists, provided that the agents' risk aversions $\gamma^1$, $\gamma^2$ are sufficiently similar. This equilibrium is also unique in a neighbourhood of the frictionless equilibrium price $\bar S$ and volatility $\bar{\sigma}$. To make these statements precise we define, for any $R>0$, the following set of progressively measurable processes:
\[
\mathcal B_\infty(R):=\big\{(S,\sigma): \|S-\bar S\|^2_{\mathcal S^\infty}+\|\sigma-\bar\sigma\|^2_{\mathbb H^2_{\rm BMO}(\mathbb Q^\beta)}\leq R^2\big\}.
\]
Our main result then can be formulated as follows:

\begin{theorem}\label{thm:ex3}
	Suppose Assumptions~\ref{ass:int1} and \ref{assump:friction} are satisfied. Then, there exists $R_{\rm max}>0$ such that for any $R<R_{\rm max}$ the system of coupled {\rm FBSDEs} $(\ref{eq:eqfwd}$--$\ref{eq:BSDES})$ has a unique solution $(S,\sigma)\in\mathcal B_\infty(R)$ provided that $|\gamma^1-\gamma^2|$ is small enough to satisfy the conditions of Theorem~\ref{thm:ex}.\footnote{In Theorem~\ref{thm:ex}, an exact upper bound for $\gamma^1-\gamma^2$ depending on $R$ and an explicit expression for $R_{\rm max}$ are provided.}
	\end{theorem}
	
Theorem~\ref{thm:ex3} is a special case of our more general well-posedness result Theorem \ref{thm:ex} and applies, for example, if the endowment volatilities $\beta^1, \beta^2$ and the terminal condition $\mathfrak{S}$ are uniformly bounded. More generally, the BMO assumptions from Assumption~\ref{assump:friction} guarantee that the equilibrium positions $\varphi^1$ and trading rates $\dot{\varphi}^1$ are uniformly bounded, which is crucial for the Picard iteration we use to prove Theorem~\ref{thm:ex3}. However, Assumption~\ref{assump:friction} does not cover specifications where the primitives $\beta^1,\beta^2$ follow certain unbounded diffusion processes such as Brownian motion. As a complement to Theorem~\ref{thm:ex3}, we therefore discuss such a concrete example in Section~\ref{s:linear} and show that the FBSDE system $(\ref{eq:eqfwd}$--$\ref{eq:BSDES})$ can be reduced to a system of deterministic but coupled Riccati equations in this case. For sufficiently similar risk aversions $\gamma^1$ and $\gamma^2$, existence of these Riccati ODEs can in turn be established by adapting the Picard iteration used to prove Theorem~\ref{thm:ex3}.

\section{An Example with Linear State Dynamics}\label{s:linear}

\subsection{Primitives and Frictionless Benchmark}

To study the impact of transaction costs an equilibrium asset prices and trading volume, we now consider a concrete example with linear state dynamics. Similarly as in~\cite{lo.al.04}, we assume that the aggregate endowment is zero and both agents' endowment volatilities follow Brownian motions:
\[
\beta^1_t= -\beta^2_t= \beta W_t, \quad \beta>0.
\]
The terminal condition for the risky asset also is a linear function of the underlying Brownian motion:
\[
\mathfrak{S}=bT+aW_T, \quad a>0 , \quad b \in \mathbb{R}.
\]
Then, the frictionless equilibrium price from Proposition~\ref{prop:BSDE} is a Bachelier model with constant expected returns and volatility:
\[
\bar{S}_t = (b-\bar{\gamma}s a^2)T+\bar{\gamma}s a^2t+aW_t, \quad t \in [0,T].
\]

\subsection{Reduction to Riccati System}

In this Markovian setting, the FBSDE system $(\ref{eq:eqfwd}$--$\ref{eq:BSDES})$ can be reformulated as a PDE. Indeed, make the standard Markovian ansatz that the backward components are smooth functions of time $t$ and the forward components $W_t$ and $\varphi^1_t$ and set $S_t = \bar S_t +  f(t,  W_t, \varphi^1_t)$ and $\dot \varphi^1_t = g(t,  W_t, \varphi^1_t,)$. Applying It\^o's formula to $f$ and $g$ and comparing the drift terms in turn leads to the following semilinear PDE for $(f, g)$, where the arguments $(t,x,y)$ are omitted to ease notation:
\begin{align*}
f_t +\frac{1}{2} f_{xx} +f_y g &= \frac{\gamma^1-\gamma^2}{2}(a+f_x)^2 y +\frac{\gamma^2}{2}f_x^2 +f_x\left(\gamma^2 a+\frac{\gamma^1-\gamma^2}{2}\beta x\right)-\frac{\gamma^1-\gamma^2}{2}a^2\left(\frac{\gamma^2s}{\gamma^1+\gamma^2}-\frac{\beta}{a}x\right),\\
g_t +\frac{1}{2} g_{xx}+ g_y g &=\frac{\gamma^1+\gamma^2}{2\lambda} (a+f_x)\beta x -\frac{\gamma^2s}{2\lambda}(a+f_x)^2+\frac{\gamma^1+\gamma^2}{2\lambda}(a+f_x)^2y, 
\end{align*}
on $[0,T) \times \mathbb{R}^2$, with terminal conditions $f(T,x,y)= g(T,x,y)=0$. 

\medskip
For the linear state dynamics and terminal conditions considered here, these PDEs can be reduced to a system of Riccati ODEs. To this end, make the linear ansatz $f(t, x, y) = A(t) + B(t) x + C(t) y$ and $g(t, x, y) = D(t) + E(t) x + F(t) y$. Plugging this into the PDEs and comparing coefficients for terms proportional to $1$, $x$, and $y$ then leads to a system of coupled Riccati equations. (An analogous ansatz is also used to link equilibria to systems of nonlinear equations in \cite{sannikov.skrzypacz.16,isaenko.20}, for example.) If these have a solution (e.g., under the conditions of Theorem~\ref{thm:odeex} below), then it identifies an equilibrium with transaction costs:

\begin{proposition}\label{prop:odes}
Suppose the following system of coupled Riccati equations has a solution on $[0,T]$:
\begin{alignat*}{2}
B'(t) &= \frac{\gamma^1-\gamma^2}{2} \beta (a+B(t))-C(t)E(t),\qquad  && B(T)=0,\\
C'(t) &= \frac{\gamma^1-\gamma^2}{2}(a+B(t))^2-C(t)F(t), && C(T)=0,\\
 E'(t) &= \frac{\gamma^1+\gamma^2}{2\lambda}\beta (a+B(t))-E(t)F(t), &&  E(T)=0,\\
 F'(t) &= \frac{\gamma^1+\gamma^2}{2\lambda}(a+B(t))^2-F(t)^2, && F(T)=0,
\end{alignat*}
and define, for $t \in [0,T]$,
\begin{align*}
A(t)&=\int_t^T \bigg(C(u)D(u)+\bar{\gamma}s a^2-\frac{\gamma^2 s}{2}(a+B(u))^2\bigg)\mathrm{d}u,\\
D(t) &= \int_t^T \bigg(\mathrm{e}^{\int_t^u F(r)\mathrm{d}r}\frac{\gamma^2 s}{2\lambda}(a+B(u))^2\bigg) \mathrm{d}u.
\end{align*}
Then, an equilibrium price with transaction costs and the corresponding optimal trading rates are given by
\begin{align*}
S_t=\bar{S}_t+A(t)+B(t)W_t+C(t)\varphi^1_t,\qquad \dot{\varphi}^1_t = -\dot{\varphi}^2_t= D(t)+E(t)W_t+F(t)\varphi^1_t, \qquad t \in [0,T],
\end{align*}
where
\[
\varphi^1_t = \mathrm{e}^{\int_0^t F(r)\mathrm{d}r} x^1 + \int_0^t \mathrm{e}^{\int_u^t F(r)\mathrm{d}r}(D(u)+E(u)W_u)\mathrm{d}u, \quad t \in [0,T].
\]
\end{proposition}

\subsection{Existence and Approximations for Similar Risk Aversions}

Similarly as in Theorem~\ref{thm:ex3}, a solution of the ODE system is guaranteed to exist, provided the agents' risk aversions are sufficiently similar.

\begin{theorem}\label{thm:odeex}
Suppose that 
\begin{equation}
\label{eq:thm:odeex}
|\gamma^1-\gamma^2| < \min\left(\frac{16 \lambda}{27 a^2 T^3 (\gamma^1 + \gamma^2) + 48 T \beta \lambda }, \frac{32 \lambda^2}{81 a^4 T^5 (\gamma^1 + \gamma^2)^2 + 72 a^2 T^3\beta (\gamma^1 + \gamma^2) \lambda + 32 T \beta \lambda^2}\right).
\end{equation}
Then, the system of Riccati equations from Proposition~\ref{prop:odes} has a $($unique$)$ solution on $[0,T]$.
\end{theorem}

The Riccati equations from Proposition~\ref{prop:odes} can readily be solved numerically with standard ODE solvers. To shed some light on their comparative statics, it is nevertheless instructive to consider the asymptotics as the difference 
\[\varepsilon=\gamma^1-\gamma^2,\]
of the agents' risk aversions tends to zero. For $\varepsilon=0$, we evidently have $C(t;0)=0$, which in turn gives $B(t;0)= 0$ and $A(t;0)=0$. Next, we use the following simple facts:
\begin{itemize}
\item for $\alpha >  0$, the Riccati ODE $H'(t) = \alpha^2 - H^2(t), H(T) =0$, has the solution $H(t) = -\alpha \tanh(\alpha (T-t))$; 
\item for $H$ as above, and $\kappa\in\mathbb R$, the Riccati ODE $J'(t) = \kappa \alpha^2 - J(t) H(t), J(T) = 0$, has the solution $J(t) = \kappa H(t)$;
\item for $H$ as above, $\int_t^T \mathrm{e}^{\int_t^s H(r)\mathrm{d}r}\mathrm{d} s = -\frac{1}{\alpha^2} H(t)$, $\int_t^T H(s)^2 \mathrm{d} s  = \alpha^2(T - t) + H(t)$ as well as $\int_t^T H(s) \mathrm{d} s = \log (\cosh(\alpha (T-t)))$.
\end{itemize}
Setting 
$$\delta:=\sqrt{\frac{(\gamma^1+\gamma^2) a^2}{2\lambda}},$$
this first yields
\begin{gather*}
F(t;0) =-\delta \tanh\big(\delta (T-t)\big), \qquad E(t;0) =\frac{\beta}{a} F(t;0), \qquad D(t;0) = -\frac{\gamma^2 s}{\gamma^1 + \gamma^2} F(t;0).
\end{gather*}
With these limiting functions for $\varepsilon \to 0$ at hand, it is straightforward to also derive the corresponding first--order asymptotics of $C(t;\varepsilon)$, $B(t;\varepsilon)$, and $A(t;\varepsilon)$,
\begin{align}
C(t;\varepsilon) &= -\frac{\varepsilon a^2}{2} \int_t^T \mathrm{e}^{\int_t^s F(r;0)\mathrm{d}r}\mathrm{d}s +o(\varepsilon)= \frac{\varepsilon a^2}{2 \delta^2} F(t;0) +o(\varepsilon), \notag \\
B(t;\varepsilon) &= \int_t^T \left(-\frac{\varepsilon \beta a}{2}+C(s;\varepsilon)E(s;0)\right)\mathrm{d}s +o(\varepsilon) = -\frac{\varepsilon \beta a}{2} (T -t ) + \frac{\varepsilon \beta a}{2 \delta^2}\int_t^T  F^2(s;0) \mathrm{d}s +o(\varepsilon) \notag \\
&= \frac{\varepsilon \beta a}{2 \delta^2}F(t;0) +o(\varepsilon), \label{eq:funcB} \\
A(t;\varepsilon) &= \int_t^T\bigg( C(s;\varepsilon)D(s;0)+\frac{\varepsilon \gamma^2 s a^2}{2(\gamma^1+\gamma^2)}-\gamma^2 s a B(s,\varepsilon)\bigg)\mathrm{d}s+o(\varepsilon) \notag \\
&=- \frac{\varepsilon \gamma^2 s a^2}{2 (\gamma^1 + \gamma^2)\delta^2}  \int_t^T F^2(s;0) \mathrm{d}s + \frac{\varepsilon \gamma^2 s a^2}{2(\gamma^1+\gamma^2)} (T - t)  - \frac{\varepsilon \beta \gamma^2 s a^2}{2 \delta^2} \int_t^T F(s; 0) \mathrm{d}s +o(\varepsilon) \notag \\
&=- \frac{\varepsilon \gamma^2 s a^2}{2 (\gamma^1 + \gamma^2)\delta^2}  F(t;0) + \frac{\varepsilon \gamma^2 s \beta \lambda}{\gamma^1+\gamma^2} \log \Big(\cosh\big(\delta(T-t)\big)\Big) +o(\varepsilon). 
\label{eq:funcA}
\end{align}

\subsection{Trading Volume}

The above expansions show that, as $\varepsilon \to 0$, the equilibrium trading rate $\dot{\varphi}^1$ from Proposition~\ref{prop:odes} converges to
\begin{equation}\label{eq:speed}
\dot{\varphi}^1_t = D(t; 0) + E(t;0) W_t + F(t; 0) \varphi^1_t = F(t; 0) \times (\varphi^1_t - \bar{\varphi}^1_t),
\end{equation}
where $\bar{\varphi}^1_t:=\frac{\gamma^2s}{\gamma^1+\gamma^2}-\frac{\beta}{a}W_t$. Whence, at the leading order for small $\varepsilon$, the equilibrium position of agent $1$ tracks its frictionless counterpart $\bar{\varphi}^1_t$ with the relative trading speed $-F(t;0)$. Accordingly, for small $\varepsilon$, the corresponding deviation $\varphi^1_t-\bar{\varphi}^1_t$ approximately has Ornstein--Uhlenbeck dynamics,
\begin{equation}\label{eq:OUdyn}
\mathrm{d}\big(\varphi^1_t-\bar{\varphi}^1_t\big) \approx \big(F(t;0)(\varphi^1_t-\bar{\varphi}^1_t)\big)\mathrm{d}t +\frac{\beta}{a}\mathrm{d}W_t.
\end{equation}
In view of~\eqref{eq:speed}, the corresponding trading volume has Ornstein--Uhlenbeck dynamics as well, until trading slows down and eventually stops near the terminal time $T$. As in the partial equilibrium model of \cite{guasoni.weber.15}, trading volume in the model therefore reproduces the main stylized features observed empirically such as mean-reversion and autocorrelation~\cite{lo.wang.00}.

\subsection{Illiquidity Discounts, Liquidity Premia, and Increased Volatility}

Let us now turn to the corresponding equilibrium price of the risky asset. It's initial level $S_0$ is adjusted by $A(0;\varepsilon)+C(0;\varepsilon)x^1$ compared to the frictionless case. Here, $C(0;\varepsilon)x^1$ quickly converges to a stationary value as the time horizon $T$ grows. In contrast, $A(0;\varepsilon)$ approximately grows linearly (via the second term in \eqref{eq:funcA}) and therefore dominates for long time horizons. Hence, the “illiquidity discount” is given by
\begin{equation}\label{eq:discount}
-(A(0;\varepsilon)+C(0;\varepsilon)x^1) = -\frac{(\gamma^1-\gamma^2)\gamma^2 s}{\sqrt{2(\gamma^1+\gamma^2)}} T \beta a \sqrt{\lambda}+O(1) = \gamma^2 s a T B(0; \varepsilon) + O(1), \quad \mbox{as $T \to \infty$}.
\end{equation}
Therefore, as in the overlapping generations model of \cite{vayanos.98}, the stock price in our model can be either increased or decreased due to transaction costs. The sign of this correction term is determined by the difference $\gamma^1-\gamma^2$ of the agents' risk aversions. If we choose $\gamma^2>\gamma^1$ to match the positive illiquidity discounts observed empirically~\cite{amihud.mendelson.86a}, then the discount~\eqref{eq:discount} is concave in the transaction cost consistent with the empirical findings of~\cite{amihud.mendelson.86a}. Note also that up to a scaling factor $\gamma^2 s a T$, the latter coincides with the volatility correction $B(t)$ for small $|\gamma^1-\gamma^2|$. 

Next, let us turn to the drift rate of the risky asset.  Using integration by parts, the ODEs satisfied by $A$, $B$, and $C$, and the asymptotics~\eqref{eq:funcB} for the function $B(t)$, we obtain that  
the difference to its frictionless counterpart is given by
\begin{align*}
&A'(t)+B'(t)W_t +C'(t)\varphi^1_t+C(t)\dot{\varphi}^1_t\\
&\qquad= (A'(t)+C(t)D(t))+(B'(t)+C(t)E(t))W_t+(C'(t)+C(t)F(t))\varphi^1_t\\
&\qquad=\bigg(-\bar{\gamma}s a^2+\frac{\gamma^2 s}{2}(a+B(t))^2\bigg)+\bigg(\frac{\gamma^1-\gamma^2}{2}\beta(a+B(t))\bigg)W_t+\bigg(\frac{\gamma^1-\gamma^2}{2}(a+B(t))^2\bigg)\varphi^1_t\\
&\qquad = \frac{(\gamma^1-\gamma^2)\gamma^2 s  }{2(\gamma^1 + \gamma^2)} a^2  + \gamma^2 s a B(t)+\frac{\gamma^1-\gamma^2}{2} a \beta W_t+\frac{\gamma^1-\gamma^2}{2} a^2\varphi^1_t + o(|\gamma^1-\gamma^2|) \\
&\qquad=\frac{\gamma^1-\gamma^2}{2}a^2 (\varphi^1_t-\bar{\varphi}^1_t)+\gamma^2 s a B(t)+o(|\gamma^1-\gamma^2|).
\end{align*}
We see that the ``liquidity premium'' compared to the frictionless case consists of two parts. The first is a rescaling of the Ornstein--Uhlenbeck process~\eqref{eq:OUdyn}: like in \cite{sannikov.skrzypacz.16,bouchard.al.17}, transaction costs endogenously lead to a mean-reverting ``momentum factor'' as in the reduced form models of \cite{kim.omberg.96,bouchaud.al.12,martin.12,garleanu.pedersen.13}. However, unlike in \cite{sannikov.skrzypacz.16,bouchard.al.17} where the difference between frictional and frictionless expected returns fluctuates around zero, an additional deterministic component appears here. Up to rescaling with the factor $\gamma^2 s a$, it coincides with the volatility correction $B(t)$ for small $|\gamma^1-\gamma^2|$. 

As a consequence, the illiquidity discount of the initial price $S_0$, the average liquidity premium in the expected returns, and the volatility correction all have the same sign in our model, which is determined by the difference $\gamma^1-\gamma^2$ of the agents' risk aversion coefficients. The empirical literature consistently finds positive illiquidity discounts~\cite{amihud.mendelson.86a} and liquidity premia~\cite{amihud.mendelson.86a,brennan.subrahmanyam.96,pastor.stambaugh.03}. If we choose $\gamma^2>\gamma^1$ to reproduce this in our model, then it follows that the corresponding volatility correction due to transaction costs is also positive. This theoretical result that illiquidity should lead to higher volatilities corroborates empirical results of \cite{umlauf.93,jones.seguin.97,hau.06}, numerical findings of \cite{adam.al.15,buss.al.16}, and the predictions of a risk-neutral model with asymmetric information studied in~\cite{danilova.julliard.19}.

To understand the intuition behind this result in our model, recall that $\beta>0$, so that for a positive price shock, the cash value of the frictionless trading target from \eqref{eq:indopt} decreases for agent 1 and increases for  agent~2. Accordingly, agent 2 can be interpreted as a ``trend follower'', whereas agent 1 follows a ``contrarian'' strategy. With transaction costs, if $\gamma^2>\gamma^1$, the trend follower's buying motive after a positive price shock is stronger than the contrarian's motive to sell. To clear the market, the expected return of the risky asset therefore has to decrease compared to the frictionless benchmark to make selling more attractive for agent 1. Accordingly, positive price shocks are dampened and an analogous argument shows that the same effect persists for negative price shocks. Since price shocks are dampened, the equilibrium volatility therefore has to increase in order to match the fixed terminal condition.

\section{Proofs for Section~\ref{s:indopt}}\label{s:proofs}

This section contains the proofs of the results on Riccati BSRDEs and FBSDEs from Section~\ref{s:indopt}. First, we prove Lemma~\ref{lem:BRSDE}, which ensures existence and uniqueness of suitably integrable solutions of the BSRDE~\eqref{eq:BRSDE} for volatility processes $\sigma \in \mathbb{H}^2_{\mathrm{BMO}}$.

\begin{proof}[Proof of Lemma~\ref{lem:BRSDE}]
For each $n \in \mathbb{N}$, consider the truncated process $\sigma^n:=\sigma \wedge n$. Since this process is uniformly bounded, the truncated BSRDE
\begin{equation}\label{eq:n}
c^n_t =\int_t^T \left(\frac{\gamma}{\lambda}(\sigma^n_s)^2-(c^n_s)^2\right)\mathrm{d}s-\int_t^TZ^n_s\mathrm{d}W_s, \quad t\in[0,T],
\end{equation}
has a unique solution $(c^n,Z^n) \in \mathcal S^{\infty} \times \mathbb{H}^2$ for each $n$ by~\cite[Theorem~2.1]{kohlmann.tang.02}. Indeed, in their notation, our case corresponds to
\[A=C=D=0,\quad N=B=1,\quad Q(t)=\frac{\gamma}{\lambda} (\sigma_t^n)^2,\quad M=0.\]
Since $N$ is positive and uniformly bounded away from $0$, $M$ is bounded and nonnegative, and $Q$ is bounded and nonnegative, \cite[Theorem~2.1]{kohlmann.tang.02} indeed does apply.

Then, by taking conditional expectations, we see that all of these solutions are uniformly bounded from above, since
\[
c^n_t = \mathbb E_t\bigg[\int_t^T\left( \frac{\gamma}{\lambda} (\sigma^n_s)^2  - (c^n_s)^2 \right)\mathrm{d}s\bigg] \leq \frac{\gamma}{\lambda} \|\sigma\|_{\mathbb{H}^2_{\mathrm{BMO}}}^2.
\]
By the comparison theorem for Lipschitz BSDEs~\cite[Theorem~9.4]{touzi.13}, $c^n_t\geq 0$ for any $t\in[0,T]$, since $(0,0)$ is the unique solution of the BSDE with terminal condition $0$ and generator $-y^2$. Whence, the solutions of the truncated equations satisfy 
\begin{equation}\label{eq:Linfty}
0 \leq c_t^n \leq \frac{\gamma}{\lambda} \|\sigma\|_{\mathbb{H}^2_{\mathrm{BMO}}}^2, \quad t \in [0,T].
\end{equation}
Also record for future reference that the corresponding martingale parts are given by
\begin{equation}\label{eq:Mn}
M^n_t = \int_0^t Z^n_s \mathrm{d}W_s = -c^n_0+c^n_t +\int_0^t \left(\frac{\gamma}{\lambda} (\sigma^n_s)^2  - (c^n_s)^2\right) \mathrm{d}s.
\end{equation}
The family $(\sup_{t \in [0, T]} |M^n_t|)_{n \in \mathbb{N}}$ is bounded in $\mathbb{L}^2$. Indeed, as each $c^n$ satisfies~\eqref{eq:Linftybis}, we obtain
\[ \sup_{t \in [0, T]} |M^n_t| \leq 2 \frac{\gamma}{\lambda}\|\sigma\|_{\mathbb{H}^2_{\mathrm{BMO}}}^2 + \frac{\gamma}{\lambda} \int_0^T \sigma_s^2\mathrm{d}s + T\frac{\gamma^2}{\lambda^2}\|\sigma\|_{\mathbb{H}^2_{\mathrm{BMO}}}^4, \quad n \in \mathbb{N}.\]
Now use the elementary inequality $(a + b + c)^2 \leq 3 (a^2 + b^2 + c^2)$ and the energy inequality for BMO martingales~\cite[p.26]{kazamaki.94} to obtain the desired bound:
\begin{equation}
\label{eq:Mn:L2bound}
\mathbb E\bigg[\Big(\sup_{t \in [0, T]}|M^n_t|\Big)^2\bigg] \leq 3 \frac{\gamma^2}{\lambda^2}\|\sigma\|_{\mathbb{H}^2_{\mathrm{BMO}}}^4 + 6 \frac{\gamma^2}{\lambda^2} \|\sigma\|_{\mathbb{H}^2_{\mathrm{BMO}}}^4 + 3 T^2 \frac{\gamma^4}{\lambda^4} \|\sigma\|_{\mathbb{H}^2_{\mathrm{BMO}}}^8, \quad n \in \mathbb{N}.
\end{equation}
Next, note that since the solutions of \eqref{eq:n} are uniformly bounded by \eqref{eq:Linfty} for all $n$, the pair $(c^n, Z^n)$ also solves the BRSDE
\[c^n_t =\int_t^T \left(\frac{\gamma}{\lambda}(\sigma^n_s)^2-\left(\left(c^n_s\right)^+\wedge\frac{\gamma}{\lambda} \|\sigma\|_{\mathbb{H}^2_{\mathrm{BMO}}}^2\right)^2\right)\mathrm{d}s-\int_t^TZ^n_s\mathrm{d}W_s.\]
Since the generator of this BRSDE is uniformly Lipschitz continuous, and its value at $0$ is bounded, the standard comparison theorem for Lipschitz BSDEs (see, e.g.,  \cite[Theorem~9.4]{touzi.13}) shows that the solutions $c^n$ are nondecreasing in $n$.

Therefore, the monotone limit $c = \lim_{n \to \infty} c^n$ is well defined, and satisfies $c_T=0$ and \eqref{eq:Linftybis} by construction. Now set
\[
M_t:=-c_0+c_t +\int_0^t \left(\frac{\gamma}{\lambda}\sigma_s^2  - c_s^2\right) \mathrm{d}s,\quad t\in[0,T].
\]
Recalling that both $\sigma^n$ and $c^n$ are nonnegative and nondecreasing in $n$, the monotone convergence theorem gives
\[\lim_{n\rightarrow \infty}\int_0^t \frac{\gamma}{\lambda}\big(\sigma_s^n\big)^2 \,\mathrm{d}s=\int_0^t \frac{\gamma}{\lambda}\sigma_s^2\,\mathrm{d}s,\quad \lim_{n\rightarrow \infty} \int_0^t \big(c_s^n\big)^2\,\mathrm{d}s=\int_0^t c_s^2\,\mathrm{d}s.\]
Therefore, $M$ is the pointwise limit of $M^n$. Since the family $(\sup_{t \in [0, t]} |M^n_t|)_{n \in \mathbb{N}}$ is bounded in $\mathbb{L}^2$, $M^n_t$ therefore converges to $M_t$ in $\mathbb{L}^1$ for each $t \in [0, T]$. Hence, it follows that $M$ is a square-integrable martingale, and the martingale representation theorem shows that $M=\int_0^\cdot Z_t \mathrm{d}W_t$ for a process $Z \in \mathbb{H}^2$. 

In summary, recalling that $c_T=0$, we have
\[
\int_t^T Z_s \,\mathrm{d}W_s = M_T-M_t = -c_t + \int_t^T \left(\frac{\gamma}{\lambda} \sigma_s^2  - c_s^2\right) \mathrm{d}s,
\]
that is, $(c,Z) \in \mathcal{S}^\infty \times \mathbb{H}^2$ solves the original BSDE. Moreover, by It\^o isometry and the conditional version of the argument used in \eqref{eq:Mn:L2bound}, it follows that for any $\tau \in \mathcal{T}_{0, T}$,
\begin{align*}
\mathbb E_\tau\bigg[\int_{\tau}^T Z_s^2 \mathrm{d}s\bigg] &= \mathbb E\bigg[\left(\int_{\tau}^T Z_s \mathrm{d} W_s\right)^2\bigg] \leq \mathbb E_\tau\bigg[\sup_{t \in [\tau, T]}\left|\int_{\tau}^t Z_s \mathrm{d} W_s\right|^2\bigg] \notag \\
&\leq 3 \frac{\gamma^2}{\lambda^2}\|\sigma\|_{\mathbb{H}^4_{\mathrm{BMO}}}^2 + 6 \frac{\gamma^2}{\lambda^2} \|\sigma\|_{\mathbb{H}^2_{\mathrm{BMO}}}^4 + 3 T^2 \frac{\gamma^4}{\lambda^4} \|\sigma\|_{\mathbb{H}^2_{\mathrm{BMO}}}^8.
\end{align*}
Thus, $Z$ is also in $\mathbb{H}^2_{\mathrm{BMO}}$. 

Finally, uniqueness follows from the following standard estimates. Let $(c^\prime,Z^\prime)\in\mathcal S^\infty\times\mathbb H^2_{\rm BMO}$ be another solution. Then
\begin{equation*}
c_t -c^\prime_t=\int_t^T \big(c^\prime_s+c_s\big)\big(c^\prime_s-c_s\big)\mathrm{d}s-\int_t^T\big(Z_s-Z^\prime_s)\mathrm{d}W_s.
\end{equation*}
The product rule then gives 
\begin{equation*}
\mathrm{e}^{\int_0^t(c^\prime_s+c_s)\mathrm{d}s}\big(c_t -c^\prime_t)=-\int_t^T\mathrm{e}^{\int_0^s(c^\prime_u+c_u)\mathrm{d}u}\big(Z_s-Z^\prime_s)\mathrm{d}W_s.
\end{equation*}
Since both $c$ and $c^\prime$ are bounded and both $Z$ and $Z^\prime$ are in $\mathbb H^2_{\rm BMO}$, the right-hand side is a $\mathbb{P}$-martingale. We conclude that $c=c^\prime$ by taking conditional expectations, which in turn implies that $Z=Z^\prime$ by uniqueness in the It\^o representation theorem.
\end{proof}

Next, we prove Lemma~\ref{lem:kt02}, which solves the FBSDE~(\ref{eq:indopt21}--\ref{eq:indopt22}) describing the optimiser of the quadratic tracking problem~\eqref{eq:pref3}.

\begin{proof}[Proof of Lemma~\ref{lem:kt02}]
First note that since $\sigma \in \mathbb{H}^2_{\mathrm{BMO}}$, Lemma~\ref{lem:BRSDE} shows that there is a unique solution $c$ of the BSRDE~\eqref{eq:BRSDE} which is nonnegative and bounded. Next, as $c$ is nonnegative, the (conditional version of the) Cauchy-Schwarz inequality, $\sigma \in \mathbb{H}^2_{\mathrm{BMO}}$ and Fubini's theorem give
\begin{align*}
\mathbb{E}\bigg[\int_0^T \bar \xi^2_t \mathrm{d}t \bigg] &\leq \frac{\gamma^2}{\lambda^2} \mathbb{E}\bigg[\int_0^T \mathbb{E}_t\bigg[\int_t^T \sigma_s (\sigma_s|\xi_s|) \mathrm{d}s \bigg]^2 \mathrm{d}t \bigg] \leq  \frac{\gamma^2}{\lambda^2} \mathbb{E}\bigg[\int_0^T \| \sigma \|^2_{\mathbb H^2_{\mathrm{BMO}}} \mathbb{E}_t \bigg[\int_t^T \sigma^2_s \xi_s^2 \mathrm{d}s \bigg] \mathrm{d}t \bigg] \\
&\leq \frac{\gamma^2}{\lambda^2} \| \sigma \|^2_{\mathbb H^2_{\mathrm{BMO}}} \int_0^T \mathbb{E} \bigg[\int_t^T \sigma^2_s \xi_s^2 ds \bigg]  \mathrm{d}t \leq \frac{\gamma^2}{\lambda^2} \| \sigma \|^2_{\mathbb H^2_{\mathrm{BMO}}} T \mathbb{E} \bigg[\int_0^T \sigma^2_s \xi_s^2 \mathrm{d}s \bigg].
\end{align*}
Together with $\sigma \xi \in \mathbb{H}^2$, this shows that $\bar{\xi}$ also belongs to $\mathbb{H}^2$. Notice now that $\bar{\xi}$ can also be directly characterised as the unique solution of the linear {\rm BSDE}
\begin{equation}\label{eq:dynsig}
\bar{\xi}_t = \int_t^T\Big(\frac{\gamma}{\lambda} \sigma^2_s\xi_s-c_s \bar{\xi}_s \Big)\mathrm{d}s-\int_t^TZ^{\xi}_s \mathrm{d}W_s, \quad t\in[0,T].
\end{equation}
Similarly, for $\gamma=\gamma^n$, $x=x^n$ and $\xi=\hat{\varphi}^n$, $(\varphi,\dot{\varphi})$ solves the FBSDE~$(\ref{eq:indopt21}$--$\ref{eq:indopt22})$ characterising the optimisers for~\eqref{eq:pref3}. Indeed, the forward equation~\eqref{eq:indopt21} is evidently satisfied by definition. The terminal condition $\dot{\varphi}_T=0$ follows from $c_T=0$, the fact that $\bar{\xi}\in\mathbb H^2$, and $\sigma^2\xi\in\mathbb H^1$. It therefore remains to show that $\dot\varphi$ also has the backward dynamics~\eqref{eq:indopt22}. The ODE~\eqref{eq:ode} for $\dot{\varphi}$, integration by parts, and the dynamics~\eqref{eq:dynsig} and \eqref{eq:BRSDE} of $\bar{\xi}$ and $c$ show that
\begin{align*}
\mathrm{d}\dot{\varphi}_t &= \mathrm{d}\bar{\xi}_t -c_t \dot{\varphi}_t \mathrm{d}t -\varphi_t \mathrm{d}c_t\\
&= \bigg(-\frac{\gamma}{\lambda} \sigma^2_t \xi_t+c_t \bar{\xi}_t\bigg)\mathrm{d}t+Z^\xi_t \mathrm{d}W_t  - c_t \dot{\varphi}_t \mathrm{d}t -\varphi_t \bigg(c^2_t-\frac{\gamma}{\lambda}\sigma^2_t\bigg)\mathrm{d}t-\varphi_t Z^c_t \mathrm{d}W_t.
\end{align*}
Again using the ODE~\eqref{eq:ode} to replace $\bar{\xi}_t$ with $\dot{\varphi}_t+c_t \varphi_t$, it follows that the trading rate~\eqref{eq:ode} indeed has the required dynamics:
\[
\mathrm{d}\dot{\varphi}_t =\frac{\gamma\sigma^2_t}{\lambda}(\varphi_t-\xi_t)\mathrm{d}t+(Z^\xi_t-\varphi_t Z^c_t)\mathrm{d}W_t.
\]
Since $c$ is nonnegative and $\bar{\xi} \in \mathbb{H}^2$, we have $\varphi \in \mathcal{S}_2$. As $\sigma \in \mathbb{H}^2_{\mathrm{BMO}}$, Lemma~\ref{lem:increasing process} (with $A_t=\sup_{s \in [0,t]} \varphi_s^2$ and $\beta_t=(Z^c_t)^2$) in turn shows that the local martingale in this decomposition is in fact a square-integrable martingale. The same argument shows that $\varphi \sigma \in \mathbb{H}^2$, and $\dot{\varphi}$ also belongs to $\mathbb{H}^2$ by \eqref{eq:ode} because $(\bar{\xi},\varphi) \in \mathbb{H}^2\times \mathbb{H}^2$ and $c$ is bounded. As a consequence, the admissible trading rate $\dot\varphi$ and the corresponding position $\varphi$ are optimal for \eqref{eq:pref3}. In particular, the solution is unique. Finally, if $\sigma \xi^{\frac{1}{2}} \in \mathbb{H}^2_{\mathrm{BMO}}$, $\bar\xi$ is bounded since $c$ is nonnegative. In view of~\eqref{eq:varphi}, $\varphi$ therefore is uniformly bounded as well as $\bar\xi$ is bounded and $c$ is nonnegative. The boundedness of $\dot{\varphi}$ in turn follows from~\eqref{eq:ode} since $\bar{\xi}$, $c$, and $\varphi$ are bounded.
\end{proof}

\section{Stability Results}\label{sec:stab}

We now derive a number of stability results, some of which might be interesting on their own right. These are the key ingredients for the convergence of the Picard iteration that allows us to prove existence for the FBSDE $(\ref{eq:thm:ex:fwd}$--$\ref{eq:thm:ex:Y})$ in Theorem~\ref{thm:ex}.

\medskip
We first consider the process $c$ from Lemma~\ref{lem:BRSDE}. Since it is positive, it also solves the counterpart of the BSDE~\eqref{eq:BRSDE} where the quadratic generator $f_t(y)=\frac{\gamma}{\lambda}\sigma_t^2-y^2$ is replaced by the monotone generator $g_t(y)=\frac{\gamma}{\lambda}\sigma_t^2-(y^+)^2$. The same argument can be applied to the $y-$derivative of the generator. Stability of the solution in turn follows from results for monotone BSDEs. To apply these estimates in the body of the paper, we develop them under an equivalent probability measure $\mathbb{P}^\alpha \sim \mathbb{P}$ with density process
\begin{equation}
\label{def:Q}
Z^\alpha:=\mathcal{E}\bigg(\int_0^\cdot \alpha_t \mathrm{d}W_t\bigg), \quad \mbox{for $\alpha \in \mathbb{H}^2_{\mathrm{BMO}}$.}
\end{equation}
Under $\mathbb{P}^\alpha$, the BSDE for $c$ can be rewritten as
\[
c_t =\int_t^T \left(\frac{\gamma}{\lambda}\sigma^2_s-c_s^2-\alpha_s Z_s\right)\mathrm{d}s-\int_t^TZ_s\mathrm{d}W^\alpha_s,
\]
for a $\mathbb{P}^\alpha$-Brownian motion $W^\alpha$. Writing $\mathbb{E}^\alpha[\cdot]$ for the expectation under $\mathbb{P}^\alpha$ to ease notation, we in turn have the following stability estimate. Notice that the techniques we use go somewhat beyond the usual ones for monotone BSDEs, as for instance in \cite{pardoux1998backward}, and rely on a trick for quadratic, one-dimensional, uncoupled BSDEs, introduced for the first time in \cite{barrieu2008closedness}.

\begin{lemma}\label{lem:stabilityc}
	Fix $(\gamma, \lambda,p,\alpha)\in(0,\infty)^2\times(1,2)\times \mathbb{H}^2_{\mathrm{BMO}}(\mathbb{P})$, with corresponding measure $\mathbb{P}^\alpha$ given by \eqref{def:Q}, and suppose that $\mathbb{E}^\alpha [\mathrm{e}^{\frac{2p}{2-p}\int_0^T \alpha^2_u \mathrm{d} u}] < \infty$. For $(\sigma,\tilde{\sigma}) \in \mathbb{H}^2_{\mathrm{BMO}}(\mathbb{P})\times \mathbb{H}^2_{\mathrm{BMO}}(\mathbb{P})$, denote by $c^\sigma$, and $c^{\tilde{\sigma}}$ the solutions of the {\rm BRSDE}~\eqref{eq:BRSDE}. Then
	\begin{equation*}
	\mathbb E^{\alpha}\bigg[\underset{t \in [0, T]}{\sup}|c^\sigma_t-c^{\tilde{\sigma}}_t|^2\bigg] \leq g_{c}(\gamma/\lambda, \alpha, \sigma, \tilde\sigma) \|\sigma-\tilde\sigma\|_{\mathbb{H}^2_{\mathrm{BMO}}(\mathbb{P}^\alpha)}^2,
	\end{equation*}
	where
	\begin{align*}
	g_{c}(\gamma/\lambda, \alpha, \sigma, \tilde\sigma) &:=  \bigg(\frac{p}{p-1}\bigg)^2  g_1(\gamma/\lambda, \alpha,\sigma, \tilde\sigma)^2  \mathbb{E}^\mathbb{\alpha}\Big[\mathrm{e}^{\frac{2p}{2-p}\int_0^T \alpha^2_u \mathrm{d} u}\Big]^{\frac{2-p}{p}},
	\end{align*}
	with 
	\[g_1(\gamma/\lambda, \alpha, \sigma, \tilde\sigma) := \frac{\gamma}{\lambda} 2^{3/2} \big(\|\sigma\|_{\mathbb{H}^2_{\mathrm{BMO}}(\mathbb{P}^\alpha)}^2+\|\tilde\sigma\|_{\mathbb{H}^2_{\mathrm{BMO}}(\mathbb{P}^\alpha)}^2\big)^{\frac{1}{2}}.
	\]
\end{lemma}

\begin{proof}
	For  $t\in[0,T]$, apply It\^o's formula to $\mathrm{e}^{\int_0^t \alpha^2_s d s}(c^\sigma_t-c^{\tilde{\sigma}}_t)^2$ and use that $\mathrm{e}^{\int_0^T \alpha^2_s \mathrm{d} s}(c^\sigma_T-c^{\tilde{\sigma}}_T)^2 =0$. Together with the {\rm BRSDE} dynamics~\eqref{eq:BRSDE}, this gives
	\begin{align}
	\notag \mathrm{e}^{\int_0^t \alpha^2_u \mathrm{d} u}(c^\sigma_t-c^{\tilde{\sigma}}_t)^2=& \int_t^T \mathrm{e}^{\int_0^s \alpha^2_u \mathrm{d} u}\Big(2\big(c^\sigma_s-c^{\tilde{\sigma}}_s)\Big(\frac{\gamma}{\lambda}\big(\sigma_s^2-\tilde\sigma_s^2\big)-\big(c^\sigma_s\big)^2+\big(c^{\tilde\sigma}_s\big)^2-\alpha_s\big(Z_s^\sigma-Z_s^{\tilde\sigma}\big)\Big)-\alpha_s^2 (c^\sigma_s-c^{\tilde{\sigma}}_s)^2\Big)\mathrm{d}s \notag \\
	&-2\int_t^T\mathrm{e}^{\int_0^s \alpha^2_u \mathrm{d} u} \big(c^\sigma_s-c^{\tilde{\sigma}}_s\big)\big(Z_s^\sigma-Z_s^{\tilde\sigma}\big)\mathrm{d}W^\alpha_s -\int_t^T \mathrm{e}^{\int_0^s \alpha^2_u \mathrm{d} u}\big(Z_s^\sigma-Z_s^{\tilde\sigma}\big)^2\mathrm{d}s.
	\label{eq:bound1}
	\end{align}
	Note that $\int_0^\cdot \mathrm{e}^{\int_0^s \alpha^2_u \mathrm{d} u} (c^\sigma_s-c^{\tilde{\sigma}}_s)(Z_s^\sigma-Z_s^{\tilde\sigma})\mathrm{d}W^\alpha_s$ is an $\mathbb{H}^1(\mathbb{P}^\alpha)$-martingale. Indeed, using that $c^\sigma, c^{\tilde\sigma} \in \mathcal{S}^\infty$ and $Z^\sigma,Z^{\tilde\sigma} \in \mathbb H^2(\mathbb{P}^\alpha)$ together with the elementary inequality $a b \leq a^2/2+ b^2/2$ and the fact that $\mathbb{E}^\mathbb{Q}[\mathrm{e}^{\int_0^T 2 \alpha^2_u \mathrm{d} u}] < \infty$ (since $p>1$ implies that $2p/(2-p)>2$), we obtain
	\begin{align*}
	\mathbb{E}^\alpha\bigg[\bigg(\int_0^T \mathrm{e}^{\int_0^s 2\alpha^2_u d u} \big(c^\sigma_s-c^{\tilde{\sigma}}_s\big)^2\big(Z_s^\sigma-Z_s^{\tilde\sigma}\big)^2 \mathrm{d}s  \bigg)^{\frac{1}{2}}\bigg] &\leq \big\| c^\sigma-c^{\tilde{\sigma}}\big\|_{\mathcal{S}^\infty} \mathbb{E}^\alpha\bigg[\mathrm{e}^{\int_0^T \alpha^2_u \mathrm{d} u}\bigg(\int_0^T \big(Z_s^\sigma-Z_s^{\tilde\sigma}\big)^2 \mathrm{d}s  \bigg)^{\frac{1}{2}}\bigg] \\
	& \leq \frac{1}{2} \big\|c^\sigma-c^{\tilde{\sigma}}\big\|_{\mathcal{S}^\infty} \mathbb{E}^\alpha\bigg[\mathrm{e}^{\int_0^T 2 \alpha^2_u \mathrm{d} u}+\int_0^T \big(Z_s^\sigma-Z_s^{\tilde\sigma}\big)^2 \mathrm{d}s \bigg] < \infty.
	\end{align*}
	Now, take conditional $\mathbb{P}^\alpha$-expectations on both sides of \eqref{eq:bound1}, use that $c^\sigma,$ and $c^{\tilde\sigma}$ are nonnegative to apply the elementary inequality $(x-y)(-x^2+y^2) = -(x-y)^2 (x +y)\leq 0$ for $x, y \geq 0$, and take into account the elementary inequality $-2 a b \leq a^2 + b^2$. This yields the estimate
	\begin{align}
	\mathrm{e}^{\int_0^t \alpha^2_u d u}(c^\sigma_t-c^{\tilde{\sigma}}_t)^2 \leq \mathbb E^\alpha_t\bigg[\int_t^T \mathrm{e}^{\int_0^s \alpha^2_u  \mathrm{d} u}2\frac{\gamma}{\lambda}(c^\sigma_s-c^{\tilde{\sigma}}_s)(\sigma_s^2-\tilde\sigma_s^2) \mathrm{d}s\bigg]. \label{eq:bound2}
	\end{align}
	Define the nondecreasing process
	\begin{equation*}\label{eq:A}
	A_t =\mathrm{e}^{\int_0^t \alpha^2_u d u} \sup_{u \in [0, t]}\big|c^\sigma_u-c^{\tilde{\sigma}}_u\big|,\quad 0\leq t \leq T.
	\end{equation*}
	Then by Lemma \ref{lem:increasing process}, we obtain
	\begin{align}
	\mathbb E^\alpha_t\bigg[\int_t^T \mathrm{e}^{\int_0^s \alpha^2_u \mathrm{d} u}2\frac{\gamma}{\lambda}(c^\sigma_s-c^{\tilde{\sigma}}_s)(\sigma_s^2-\tilde\sigma_s^2)\mathrm{d}s\bigg]&\leq 2\frac{\gamma}{\lambda} \mathbb E^\alpha_t\bigg[\int_t^TA_s \big|\sigma_s^2-\tilde\sigma_s^2\big|\mathrm{d}s\bigg]  \label{eq:part4}\\
	&\leq 2 \frac{\gamma}{\lambda} \bigg( \mathbb E^\alpha_t\bigg[A_t\int_t^T\big|\sigma_s^2-\tilde\sigma_s^2\big|\mathrm{d}s+\int_t^T \mathbb E^\alpha_s\bigg[\int_s^T\big|\sigma_u^2-\tilde\sigma_u^2\big|\mathrm{d}u\bigg]\mathrm{d}A_s \bigg]. 
	\notag
	\end{align}
	Next, for any $\tau\in\mathcal T_{0,T}$, the conditional version of the Cauchy--Schwarz inequality and the elementary inequality $(a+b)^2\leq 2 (a^2 + b^2)$ show that
	\begin{align*}
	\mathbb E^\alpha_\tau\bigg[\int_\tau^T \big|\sigma_s^2-\tilde\sigma_s^2\big|\mathrm{d}s\bigg]&\leq  \bigg(\mathbb E^\alpha_\tau\bigg[\int_\tau^T2 \big( |\sigma_s|^2+|\tilde\sigma_s|^2\big)\mathrm{d}s\bigg]\bigg)^{\frac12}\bigg(\mathbb E^\alpha_\tau\bigg[\int_\tau^T |\sigma_s-\tilde\sigma_s|^2\mathrm{d}s\bigg]\bigg)^{\frac12}\\
	&\leq\sqrt{2} \big(\|\sigma\|_{\mathbb{H}^2_{\mathrm{BMO}}(\mathbb{P}^\alpha)}^2+\|\tilde\sigma\|_{\mathbb{H}^2_{\mathrm{BMO}}(\mathbb{P}^\alpha)}^2\big)^{\frac{1}{2}}\|\sigma-\tilde\sigma\|_{\mathbb{H}^2_{\mathrm{BMO}}(\mathbb{P}^\alpha)}.
	\end{align*}
	Plugging this into \eqref{eq:part4}, we obtain
	\begin{align}
	\label{eq:part5}
	\mathbb E^\alpha_t\bigg[\int_t^T \mathrm{e}^{\int_0^s \alpha^2_u \mathrm{d} u}2\frac{\gamma}{\lambda}(c^\sigma_s-c^{\tilde{\sigma}}_s)(\sigma_s^2-\tilde\sigma_s^2)\mathrm{d}s\bigg]
	\leq  g_1(\gamma/\lambda, \alpha, \sigma, \tilde\sigma) \|\sigma-\tilde\sigma\|_{\mathbb{H}^2_{\mathrm{BMO}}(\mathbb{P}^\alpha)} \mathbb E^\alpha_t\big[A_T \big].
	\end{align}
	Inserting~\eqref{eq:part5} back into~\eqref{eq:bound2} gives
	\begin{align*}
	|c^\sigma_t-c^{\tilde{\sigma}}_t|^2 &\leq \mathrm{e}^{-\int_0^t \alpha^2_u \mathrm{d} u} g_1(\gamma/\lambda, \alpha, \sigma, \tilde\sigma) \|\sigma-\tilde\sigma\|_{\mathbb{H}^2_{\mathrm{BMO}}(\mathbb{P}^\alpha)}  \mathbb E^\alpha_t\big[A_T \big]\leq  g_1(\gamma/\lambda, \alpha,\sigma, \tilde\sigma) \|\sigma-\tilde\sigma\|_{\mathbb{H}^2_{\mathrm{BMO}}(\mathbb{P}^\alpha)}  \mathbb E^\alpha_t\big[A_T \big].
	\end{align*}
	Now, take the supremum over $t \in [0,T]$ on both sides, then $\mathbb{P}^\alpha$-expectations, and finally use Lemma \ref{lem:Doob 2} for a fixed $p\in(1,2)$. It follows that, for any $\varepsilon > 0$,
	\begin{align*}
	\mathbb E^\alpha\bigg[\underset{t \in [0, T]}{\sup}|c^\sigma_t-c^{\tilde{\sigma}}_t|^2\bigg] \leq &\ g_1(\gamma/\lambda, \alpha,\sigma, \tilde\sigma) \|\sigma-\tilde\sigma\|_{\mathbb{H}^2_{\mathrm{BMO}}(\mathbb{P}^\alpha)}\\
	& \times  \left(\frac{1}{\varepsilon}\mathbb{E}^\alpha\bigg[\sup_{0\leq s\leq T}\big|c^\sigma_u-c^{\tilde{\sigma}}_u\big|^2\bigg] + \frac{\varepsilon}{4}\bigg(\frac p{p-1}\bigg)^2  \mathbb{E}^\alpha\Big[\mathrm{e}^{\frac{2p}{2-p}\int_0^T \alpha^2_u \mathrm{d} u}\Big]^{\frac{2-p}{p}} \right).
	\end{align*}
	This implies that, for any $\varepsilon>g_1(\gamma/\lambda, \alpha,\sigma, \tilde\sigma) \|\sigma-\tilde\sigma\|_{\mathbb{H}^2_{\mathrm{BMO}}(\mathbb{P}^\alpha)}$,
	\begin{align*}
	\mathbb E^\alpha\bigg[\underset{t \in [0, T]}{\sup}|c^\sigma_t-c^{\tilde{\sigma}}_t|^2\bigg]&\leq  \frac{\varepsilon^2 g_1(\gamma/\lambda, \alpha,\sigma, \tilde\sigma) \|\sigma-\tilde\sigma\|_{\mathbb{H}^2_{\mathrm{BMO}}(\mathbb{P}^\alpha)}}{4\big(\varepsilon -g_1(\gamma/\lambda, \alpha,\sigma, \tilde\sigma) \|\sigma-\tilde\sigma\|_{\mathbb{H}^2_{\mathrm{BMO}}(\mathbb{P}^\alpha)}\big)} \bigg(\frac p{p-1}\bigg)^2  \mathbb{E}^\alpha\Big[\mathrm{e}^{\frac{2p}{2-p}\int_0^T \alpha^2_u \mathrm{d} u}\Big]^{\frac{2-p}{p}}.
	\end{align*}
	The asserted estimate in turn corresponds to the optimal choice $\varepsilon=2 g_1(\gamma/\lambda, \alpha,\sigma, \tilde\sigma) \|\sigma-\tilde\sigma\|_{\mathbb{H}^2_{\mathrm{BMO}}(\mathbb{P}^\alpha)}$.
\end{proof}

We now move on to the process 
\begin{equation}\label{eq:hatxi}
\bar{\xi}_t^\sigma = \frac{\gamma}{\lambda} \mathbb E_t\bigg[\int_t^T \mathrm{e}^{-\int_t^s c_u du} \xi_s \sigma^2_s \mathrm{d}s\bigg],
\end{equation}
from Lemma~\ref{lem:kt02}. The linear (and, in particular, monotone since $c$ is nonnegative) BSDE~\eqref{eq:dynsig} for this process rewrites as follows under the measure $\mathbb{P}^\alpha$:
\[
\bar{\xi}^\sigma_t = \int_t^T\bigg(\frac{\gamma}{\lambda}\sigma^2_s\xi_s-c_s \bar{\xi}_s^\sigma -\alpha_s Z^{\xi}_s  \bigg)\mathrm{d}s-\int_t^TZ^{\xi}_s \mathrm{d}W^\alpha_s, \quad t \in [0,T].
\]
We first record some uniform estimates, which are a direct consequence of the nonnegativity of $c$ established in Lemma~\ref{lem:BRSDE}.

\begin{corollary}\label{cor:boundxi}
	Suppose that the process $\xi = (\xi_t)_{t \in [0, T]}$ satisfies $\sigma |\xi|^{\frac{1}{2}} \in \mathbb H^2_{\mathrm{BMO}}(\mathbb P)$. Then for $(\gamma,\lambda) \in(0,\infty)^2$, the process $\bar{\xi}$ from~\eqref{eq:hatxi} satisfies
	\[
	\big\|\bar{\xi}^\sigma\big\|_{\mathcal{S}^\infty}\leq \frac{\gamma}{\lambda}  \|\sigma \xi^{\frac{1}{2}} \|^2_{\mathbb H^2_{\mathrm{BMO}}(\mathbb P)}.
	\]
\end{corollary}

Next, we show that the stability result for $c$ established in Lemma~\ref{lem:stabilityc} and another application of the stability theorem for monotone BSDEs yield the following stability result for $\bar{\xi}$.

\begin{corollary}\label{cor:stabilityxi}
	Fix $(\gamma, \lambda,p,\alpha)\in(0,\infty)^2\times(1,2)\times \mathbb{H}^2_{\mathrm{BMO}}(\mathbb{P})$, with corresponding measure $\mathbb{P}^\alpha$ given by \eqref{def:Q}, and suppose that $\mathbb{E}^\alpha\big[\mathrm{e}^{\frac{2p}{2-p}\int_0^T  \alpha^2_u \mathrm{d} u}\big] < \infty$. For $(\nu,\nu^\prime,\sigma,\sigma^\prime) \in \mathbb{H}^2_{\mathrm{BMO}}\times \mathcal{S}^{\infty}\times \mathbb{H}^2_{\mathrm{BMO}}(\mathbb{P})\times\mathbb{H}^2_{\mathrm{BMO}}(\mathbb{P})$, set $\xi^\sigma:=\frac{\nu}{\sigma}+\nu^\prime$ and $\xi^{\tilde\sigma}:=\frac{\nu}{\tilde\sigma}+\nu^\prime$, and denote by $\bar{\xi}^\sigma$ and $\bar{\xi}^{\tilde{\sigma}}$ the corresponding processes from \eqref{eq:hatxi}. Then
	\begin{equation*}
	\mathbb{E}^\alpha\bigg[\sup_{0 \leq t \leq T}|\bar{\xi}^\sigma_t-\bar{\xi}^{\tilde{\sigma}}_t|^2\bigg] \leq g_{\bar\xi}(\gamma/\lambda, \alpha, \sigma, \tilde\sigma, \nu, 
	\nu^\prime )\|\sigma-\tilde\sigma\|^2_{\mathbb{H}^2_{\mathrm{BMO}}(\mathbb{P}^\alpha)},
	\end{equation*}
	where
	\begin{align*}
	g_{\bar\xi}(\gamma/\lambda, \alpha, \sigma, \tilde\sigma, \nu, 
	\nu^\prime ) &:=\bigg(\frac{p}{p-1}\bigg)^2  \left(g_2(\gamma/\lambda, \alpha,\sigma, \tilde\sigma)\right)^2  \mathbb{E}^\alpha\Big[\mathrm{e}^{\frac{2p}{2-p}\int_0^T \alpha^2_u \mathrm{d} u}\Big]^{\frac{2-p}{p}}
	\end{align*}
	with
	\begin{align*}
	g_2(\gamma/\lambda, \alpha,\sigma, \tilde\sigma,\nu, 
	\nu^\prime ) &:= 2\frac{\gamma}{\lambda} \|\nu\|_{\mathbb{H}^2_{\mathrm{BMO}}(\mathbb{P}^\alpha)} + \|\nu^\prime\|_{\mathcal{S}^{\infty}} g_1(\gamma/\lambda, \alpha, \sigma, \tilde{\sigma})  + 2\frac{\gamma}{\lambda}\|\xi\|_{\mathcal{S}^\infty} \|\tilde\sigma\|^2_{\mathbb H^2_{\mathrm{BMO}}(\mathbb P)} T g_{c}(\gamma/\lambda, \alpha, \sigma, \tilde \sigma).
	\end{align*}
\end{corollary}

\begin{proof}
	The argument is similar to the proof of Lemma~\ref{lem:stabilityc}. For $t \in [0, T]$, apply It\^o's formula to $\mathrm{e}^{\int_0^t \alpha^2_s d s}(\bar{\xi}^\sigma_t-\bar{\xi}^{\tilde{\sigma}}_t)^2$ and use that 
	$\mathrm{e}^{\int_0^T \alpha^2_s d s}(\bar{\xi}^\sigma_T-\bar{\xi}^{\tilde{\sigma}}_T)^2 =0$. This gives
	\begin{align*}
	\mathrm{e}^{\int_0^t \alpha^2_u \mathrm{d}u}(\bar{\xi}^\sigma_t-\bar{\xi}^{\tilde{\sigma}}_t)^2 =&\ \int_t^T \mathrm{e}^{\int_0^s \alpha^2_u \mathrm{d}u} \Big(- \alpha^2_s (\bar{\xi}^\sigma_s-\bar{\xi}^{\tilde{\sigma}}_s)^2 + 2(\bar{\xi}^\sigma_s-\bar{\xi}^{\tilde{\sigma}}_s)\Big( \frac{\gamma}{\lambda} \nu_s (\sigma_s - \tilde\sigma_s)+ \frac{\gamma}{\lambda} \nu^\prime_t (\sigma_s^2 - \tilde\sigma^2_s)\Big)\Big) \mathrm{d}s \\
	& +  \int_t^T \mathrm{e}^{\int_0^s \alpha^2_u \mathrm{d}u} 2(\bar{\xi}^\sigma_s-\bar{\xi}^{\tilde{\sigma}}_s) \Big( - \big(c^{\sigma}_s \bar \xi^\sigma_s - c^{\tilde \sigma}_s \bar \xi^{\tilde \sigma}_s\big) - \alpha_s \big(Z^{\xi^{\sigma}}_s -Z^{\xi^{\tilde\sigma}}_s\big) \Big) \mathrm{d}s \\
	& +\int_t^T \mathrm{e}^{\int_0^s \alpha^2_u \mathrm{d}u} \big(Z^{\xi^{\sigma}}_s -Z^{\xi^{\tilde\sigma}}_s\big) \mathrm{d} W^{\alpha}_s - \int_t^T \mathrm{e}^{\int_0^s \alpha^2_u \mathrm{d}u} \big(Z^{\xi^{\sigma}}_s -Z^{\xi^{\tilde\sigma}}_s\big)^2 \mathrm{d}s.
	\end{align*}
	It follows as in the proof of Lemma~\ref{lem:stabilityc} that $\int_0^{\cdot} \mathrm{e}^{\int_0^s \alpha^2_u du} (Z^{\xi^{\sigma}}_s -Z^{\xi^{\tilde\sigma}}_s) d W^{\alpha}_s$ is an $\mathbb{H}^1(\mathbb{P}^\alpha)$-martingale. Now also use the elementary inequality $-2 a b \leq a^2 + b^2$, the identity $ab - cd = a(b-d) + d(a -c)$, and that $c^{\tilde\sigma}$ is non--negative. As a consequence,
	\begin{align}
	\mathrm{e}^{\int_0^t \alpha^2_u \mathrm{d}u}(\bar{\xi}^\sigma_t-\bar{\xi}^{\tilde{\sigma}}_t)^2 \leq \mathbb E^\alpha_t\bigg[ \int_t^T \mathrm{e}^{\int_0^s \alpha^2_u \mathrm{d}u}  2\big(\bar{\xi}^\sigma_s-\bar{\xi}^{\tilde{\sigma}}_s\big)\Big( \frac{\gamma}{\lambda} \nu_s (\sigma_s - \tilde\sigma_s)+ \frac{\gamma}{\lambda} \nu^\prime_s \big(\sigma_s^2 - \tilde\sigma^2_s\big)   - \bar\xi^{\tilde\sigma}_s\big(c^{\sigma}_s -c^{\tilde\sigma}_s \big) \Big) \mathrm{d}s  \bigg].
	\label{eq:pf:cor:boundxi:Ito}
	\end{align}
	Next, the conditional versions of the inequalities of Cauchy--Schwarz and Jensen's, the elementary inequality $(a+b)^2\leq 2 (a^2 + b^2)$, Lemma \ref{lem:stabilityc}, and Corollary~\ref{cor:boundxi} yield that, for any stopping time $\tau$,
	\begin{align*}
	&2\mathbb E^\alpha_\tau\bigg[ \int_\tau^T \left| \frac{\gamma}{\lambda} \nu_s (\sigma_s - \tilde\sigma_s)+ \frac{\gamma}{\lambda} \nu^\prime_s \big(\sigma_s^2 - \tilde\sigma^2_s\big)   - \bar\xi^{\tilde\sigma}_s\big(c^{\sigma}_s - c^{\tilde\sigma}_s\big) \right| \mathrm{d}s \bigg] \\
	&\leq 2\frac{\gamma}{\lambda} \mathbb E^\alpha_\tau \bigg[ \int_\tau^T (\sigma_s - \tilde\sigma_s)^2 \mathrm{d}s \bigg]^{\frac{1}{2}}\bigg(\mathbb E^\alpha_\tau \bigg[ \int_\tau^T  \nu_s^2 \mathrm{d}s \bigg]^{\frac{1}{2}}+\|\nu^\prime\|_{\mathcal{S}^{\infty}} \mathbb E^\mathbb{Q}_\tau\bigg[\int_\tau^T2 \big( \sigma_s^2+\bar\sigma_s^2\big)\mathrm{d}s\bigg]^{\frac12}\bigg)\\
	&\quad + 2\| \bar\xi^{\tilde\sigma} \|_{\mathcal{S}^{\infty}} T \mathbb E^\alpha_\tau \bigg[\underset{t \in [0, T]}{\sup}(c^\sigma_t-c^{\tilde{\sigma}}_t)^2\bigg]^{\frac12} \\
	&\leq \|\sigma - \tilde\sigma \|_{\mathbb{H}^2_{\mathrm{BMO}}(\mathbb{P}^\alpha)} \Big( 2\frac{\gamma}{\lambda} \|\nu\|_{\mathbb{H}^2_{\mathrm{BMO}}(\mathbb{P}^\alpha)} + \|\nu^\prime\|_{\mathcal{S}^{\infty}} g_1(\gamma/\lambda, \alpha, \sigma, \tilde{\sigma})  + 2\frac{\gamma}{\lambda}\|\xi\|_{\mathcal{S}^\infty} \|\tilde\sigma\|^2_{\mathbb H^2_{\mathrm{BMO}}(\mathbb P)} T g_{c}(\gamma/\lambda, \alpha, \sigma, \tilde\sigma)\Big).
	\end{align*}
	As in the proof of Lemma~\ref{lem:stabilityc}, we deduce that
	\begin{align*}
	(\bar{\xi}^\sigma_t-\bar{\xi}^{\tilde{\sigma}}_t)^2 \leq\ &\frac{\gamma}{\lambda}\|\sigma - \tilde\sigma \|_{\mathbb{H}^2_{\mathrm{BMO}}(\mathbb{P}^\alpha)} \mathbb{E}^\alpha_\tau\big[C_T\big]\\
	& \times
	\Big( 2 \|\nu\|_{\mathbb{H}^2_{\mathrm{BMO}}(\mathbb{P}^\alpha)} + \|\nu^\prime\|_{\mathcal{S}^{\infty}} g_1(\gamma/\lambda, \alpha, \sigma, \tilde{\sigma})  + 2\|\xi\|_{\mathcal{S}^\infty} \|\tilde\sigma\|^2_{\mathbb H^2_{\mathrm{BMO}}(\mathbb P)} T g_{c}(\gamma/\lambda, \alpha, \sigma, \tilde\sigma)\Big),
	\end{align*}
	where 
	\[
	C_t=\mathrm{e}^{\int_0^t \alpha^2_u d u} \sup_{u \in [0, t]}\big|\bar{\xi}^\sigma_u-\bar{\xi}^{\tilde{\sigma}}_u\big|,\quad  t \in [0,T].
	\]
	Then we can argue exactly as in the proof of Lemma~\ref{lem:stabilityc} to conclude.
\end{proof}

We finally turn to the optimal tracking strategies $\varphi$ from Lemma~\ref{lem:kt02}. Recall that these solve the (random) linear ODE
\[
\dot{\varphi}_t= \bar{\xi}_t -c_t \varphi_t, \quad \varphi_0 = x,
\]
which has the explicit solution 
\begin{equation}\label{eq:phi}
\varphi_t=\mathrm{e}^{-\int_0^t c_u \mathrm{d}u}x+\int_0^t \mathrm{e}^{-\int_s^t c_u \mathrm{d}u}\bar{\xi}_s \mathrm{d}s.
\end{equation}
Together with Corollary~\ref{cor:boundxi}, we obtain the following estimate:

\begin{corollary}\label{cor:boundphi}
	Let $(\gamma, \lambda)\in(0,\infty)^2$ and define $\xi :=\frac{\nu}{\sigma}+\nu^\prime$ for processes $(\nu, \sigma,\nu^\prime) \in \mathbb{H}^2_{\mathrm{BMO}}\times\mathbb{H}^2_{\mathrm{BMO}}\times\mathcal{S}^{\infty}$. Then, the process $\varphi$ from~\eqref{eq:phi} satisfies
	\[
	\|\varphi\|_{\mathcal{S}^\infty} \leq |\varphi_0|+ T \|\bar\xi \|_{\mathcal{S}^\infty}  \leq  |x| +  \frac{\gamma}{\lambda} T\left( \|\nu\|_{\mathbb H^2_{\mathrm{BMO}}(\mathbb P)}\|\sigma\|_{\mathbb H^2_{\mathrm{BMO}}(\mathbb P)} + \|\nu'\|_{\mathcal{S}^{\infty}}\|\sigma\|^2_{\mathbb H^2_{\mathrm{BMO}}(\mathbb P)}  \right).
	\]
\end{corollary}


This uniform bound together with the stability results for $c$ and $\bar{\xi}$ now allow to establish a stability result for the optimal tracking strategies in terms of the BMO--norm of the underlying volatility processes.

\begin{theorem}\label{thm:stabilityphi}
	Fix $(\gamma, \lambda,p,\alpha)\in(0,\infty)^2\times(1,2)\times \mathbb{H}^2_{\mathrm{BMO}}(\mathbb{P})$, with corresponding measure $\mathbb{P}^\alpha$ given by \eqref{def:Q}, and suppose that $\mathbb{E}^\alpha\big[\mathrm{e}^{\frac{2p}{2-p}\int_0^T  \alpha^2_u \mathrm{d} u}\big] < \infty$. For $(\nu,\nu^\prime, \varphi_0,\sigma,\tilde{\sigma})\in \mathbb{H}^2_{\mathrm{BMO}}\times \mathcal{S}^{\infty}\times \mathbb{R}\times\mathbb{H}^2_{\mathrm{BMO}}(\mathbb{P})\times\mathbb{H}^2_{\mathrm{BMO}}(\mathbb{P})$, set $\xi^\sigma:=\frac{\nu}{\sigma}+\nu^\prime$ and $\xi^{\tilde\sigma}:=\frac{\nu}{\sigma}+\nu^\prime$, and denote by $\varphi^\sigma$, and $\varphi^{\tilde{\sigma}}$ the corresponding strategies from~\eqref{eq:hatxi}. Then
	\begin{align*}
	&\mathbb E^\alpha\bigg[\sup_{t \in [0,T]}|\varphi^\sigma_t-\varphi^{\tilde{\sigma}}_t|^2\bigg]\leq g_{\varphi}(x,\gamma/\lambda, \alpha, \sigma, \tilde\sigma, \nu, \nu^\prime)\|\sigma-\tilde\sigma\|^2_{\mathbb{H}^2_{\mathrm{BMO}}(\mathbb{P}^\alpha)},
	\end{align*}
	where
	\begin{align*}
	g_{\varphi}(x,\gamma/\lambda, \alpha, \sigma, \tilde\sigma, \nu, 
	\nu^\prime )  :=3T^2 \left( \Big(x^2 +T^2 \frac{\gamma^2}{\lambda^2} \|\xi\|^2_{\mathcal{S}^\infty} \|\sigma\|^4_{\mathbb H^2_{\mathrm{BMO}}(\mathbb P)}\Big) g_{c}(\gamma/\lambda, \alpha, \sigma, \tilde\sigma)^2 + g_{\bar\xi}(\gamma/\lambda, \alpha, \sigma, \tilde\sigma, \nu, 
	\nu^\prime )^2 \right).
	\end{align*}
\end{theorem}

\begin{proof}
	Observe that the map $x \mapsto\mathrm{e}^{-x}$ is Lipschitz continuous on $\mathbb R_+$ with Lipschitz constant $1$. Whence, for $t \in [0, T]$,
	\begin{align*}
	|\varphi^\sigma_t-\varphi^{\tilde{\sigma}}_t| &\leq |x| \int_0^t |c^\sigma_u-c^{\tilde\sigma}_u|\mathrm{d}u+\int_0^t \left(\int_s^t |c^\sigma_u -c^{\tilde\sigma}_u|du\right)\bar{\xi}^\sigma_s ds+\int_0^t \mathrm{e}^{-\int_s^t c^{\tilde\sigma}_u \mathrm{d}u} \big|\bar{\xi}^\sigma_s-\bar{\xi}^{\tilde{\sigma}}_s\big|\mathrm{d}s\\
	&\leq |x|T \sup_{u \in [0,T]}|c^\sigma_u-c^{\tilde\sigma}_u|+T^2\|\bar{\xi}^\sigma\|^{\infty} \sup_{u \in [0,T]}|c^\sigma_u-c^{\tilde\sigma}_u|+T \sup_{u \in [0,T]}|\bar{\xi}^\sigma_u-\bar{\xi}^{\tilde{\sigma}}_u|.
	\end{align*}
	Now, take the supremum over $t \in [0,T]$ and square the result. In view of the elementary inequality $(a + b+ c)^2 \leq 3 (a^2 + b^2 + c^2)$ for $a, b, c \in \mathbb{R}$, the assertion then follows by taking $\mathbb{P}^\alpha$-expectations.
\end{proof}

\section{Proofs for Section~\ref{s:equilibrium}}\label{sec:proofseq}

We first prove Proposition~\ref{prop:BSDE} on the existence and uniqueness of frictionless Radner equilibria under the following weaker (but more involved) version of Assumption~\ref{ass:int1}.

\begin{assumption}\label{ass:int1bis} $(i)$ $\beta^1+\beta^2 \in \mathbb H^2$ and the local martingale $Z^\beta$ from~\eqref{eq:Pbeta} is a martingale;

\medskip
$(ii)$ $\mathbb E^\beta[\mathrm{e}^{-2\bar\gamma s\mathfrak S}]<\infty$;

\medskip
$(iii)$ $\mathbb{E}^\beta\Big[(Z_T^\beta)^{-\frac{p\varepsilon}{1+\varepsilon}}\Big] + \mathbb E^\beta\Big[\mathrm{e}^{-4s(1+\varepsilon)\bar\gamma \mathfrak S}\Big]+\mathbb E^\beta\Big[\mathrm{e}^{\frac{4ps\bar\gamma(1+\varepsilon)}{\varepsilon(p-1)} \mathfrak S}\Big]< \infty$ for some $\varepsilon>0$ and $p>1$.
\end{assumption}

\begin{remark}
Notice that if Assumption \ref{ass:int1} holds, then it is immediate that Assumption~\ref{ass:int1bis}$(i)$ is satisfied, since $\mathbb H^2_{\mathrm{BMO}}\subset \mathbb H^2$ and stochastic exponentials of stochastic integrals $($with respect to a Brownian motion$)$ of processes in $\mathbb H^2_{\mathrm{BMO}}$ are uniformly integrable martingales. Moreover, \ref{ass:int1bis}$(ii)$ and $(iii)$ also hold as $\mathfrak S$ has exponential moments of any order, and since $Z^\beta$ satisfies the so-called ``Muckenhoupt condition'' by  {\rm \cite[Theorem 2.4]{kazamaki.94}} because $\beta^1,\beta^2 \in \mathbb{H}^2_{\mathrm{BMO}}$.
\end{remark}

\begin{proof}[Proof of Proposition~\ref{prop:BSDE}]
The existence of a solution to \eqref{eq:bsdebis} with the appropriate properties is immediate from direct calculations or \cite[Theorem 2.1]{delbaen.al.15}.\footnote{Note that the assumption $\mathfrak S^+\in\mathbb L^1$ is not needed here.} For uniqueness, notice that for any such solution, the martingale
\[
 M_t:=\mathbb{E}_t^{\beta}\big[\mathrm{e}^{-2\bar\gamma s\mathfrak{S}}\big],\quad t\in[0,T],
\]
is uniformly integrable. It\^o's formula gives
\[
\mathrm{e}^{-2\bar\gamma sS_t}=\mathrm{e}^{-2\bar\gamma s\mathfrak S}-\int_t^T\mathrm{e}^{-2\bar\gamma sS_u} \sigma_u\mathrm{d}W^\beta_u,\quad t\in[0,T].
\]
The stochastic integral on the right-hand side must be a martingale, since the left-hand side is. We can thus take conditional expectations to deduce that
\[
S_t=-\frac{1}{2\bar\gamma s}\log\Big(\mathbb E^\beta_t\big[\mathrm{e}^{-2\bar\gamma s\mathfrak S}\big]\Big),\quad t\in[0,T].
\]
Uniqueness of $\sigma$ in turn follows from the martingale representation theorem.

Let us now verify that this price process $S$ indeed defines a Radner equilibrium. Its drift under $\mathbb P$ is immediately given by Girsanov's theorem,
\[
\mu_t=\bar\gamma s\sigma_t^2+\bar\gamma\big(\beta^1_t+\beta^2_t)\sigma_t,\quad t\in[0,T].
\]
Since $\beta^1+\beta^2\in\mathbb H^2$, we just need to verify that $\sigma\in\mathbb H^2$. To this end, notice that since the martingale $M$ satisfies, by Doob's inequality
\[
\mathbb E^\beta\bigg[\sup_{0\leq t\leq T}M_t^{2(1+\varepsilon)}\bigg]\leq \bigg(\frac{2(1+\varepsilon)}{1+2\varepsilon}\bigg)^{2(1+\varepsilon)}\mathbb E^\beta\big[\mathrm{e}^{-4(1+\varepsilon)\bar\gamma s\mathfrak S}\big]<\infty,
\]
then the martingale representation property implies the existence of a process $Z\in\mathbb H^{2+\varepsilon}(\mathbb P^\beta)$, such that 
\[
\mathrm{d}M_t=Z_t\mathrm{d}W^\beta_t,
\]
from which we deduce that
\[
\sigma_t=-\frac{1}{2\bar\gamma sM_t}Z_t.
\]
We then estimate that
\begin{align*}
\mathbb E\bigg[\int_0^T\sigma_t^2\mathrm{d}t\bigg]&=\frac{1}{4\bar\gamma^2}\mathbb E^\beta\bigg[\big(Z_T^\beta\big)^{-1}\int_0^T\frac{Z_t^2}{M_t^2}\mathrm{d}t\bigg]\\
&\leq \frac{1}{4\bar\gamma^2}\mathbb E^\beta\bigg[\big(Z^\beta_T\big)^{-\frac{1+\varepsilon}\varepsilon}\sup_{t\in[0,T]}M_t^{-\frac{2(1+\varepsilon)}\varepsilon}\bigg]^{\frac{\varepsilon}{1+\varepsilon}}\mathbb E^\beta\bigg[\bigg(\int_0^TZ_t^2\mathrm{d}t\bigg)^{1+\varepsilon}\bigg]^{\frac1{1+\varepsilon}}\\
&\leq \frac{1}{4\bar\gamma^2}\bigg(\frac{2(1+\varepsilon)}{1+2\varepsilon}\bigg)^{\frac{p(1+\varepsilon)}{\varepsilon}}\mathbb E^\beta\Big[\big(Z^\beta_T\big)^{-\frac{p(1+\varepsilon)}\varepsilon}\Big]^{\frac1p}\mathbb E^\beta\Big[\mathrm{e}^{\frac{4ps\bar\gamma(1+\varepsilon)}{\varepsilon(p-1)}\mathfrak S}\Big]^{\frac{\varepsilon (p-1)}{p(1+\varepsilon)}}\|Z\|^{\frac{2+\varepsilon}{1+\varepsilon}}_{\mathbb H^{2+\varepsilon}(\mathbb P^\beta)}<\infty.
\end{align*}
Since the market also obviously clears, this completes the proof.
\end{proof}

\begin{proof}[Proof of Corollary \ref{cor:BSDE}]
The uniqueness is clear by Proposition \ref{prop:BSDE}, and the existence of a solution in $\mathcal S^\infty\times\mathbb H^2_{\mathrm{BMO}}$ is classical, see for instance \cite[Corollary 2.1]{briand2013simple}.
\end{proof}

Next, we show that sufficiently integrable solutions of the FBSDE system $(\ref{eq:eqfwd}$--$\ref{eq:BSDES})$ indeed identify equilibria with transaction costs:

\begin{proof}[Proof of Proposition~\ref{prop:eq}]
First, Property $(ii)$ and and market clearing in Property $(iii)$ from Definition~\ref{def:equi} hold by assumption. Next, $\dot{\varphi}^1 \in \mathbb{H}^2 $ gives $\varphi^1 \in \mathcal{S}^2$. Thus, using that $\sigma \in \mathbb{H}^2_{\mathrm{BMO}}$ it follows from Lemma~\ref{lem:increasing process} that $ \sigma \varphi^1 \in  \mathbb{H}^2$ and in turn also $ \sigma \varphi^2 \in  \mathbb{H}^2$. Now, using that $\beta^1,  \beta^2, \sigma \varphi^1 \in \mathbb{H}^2$ and $\sigma \in \mathbb{H}^2_{\mathrm{BMO}} \subset \mathbb{H}^2$ 
gives Property~$(i)$.

It remains to show that $\dot \varphi^1, \dot \varphi^2$  are indeed optimal for agents 1 and 2.  By Lemma \ref{lem:kt02}, we need to check that $(\varphi^n, \dot \varphi^n)$ solves the FBSDEs characterisation of agent $n$'s individually optimal trading in  $(\ref{eq:indopt21}$--$\ref{eq:indopt22})$. This follows immediately from the forward-backward dynamics $(\ref{eq:eqfwd}$--$\ref{eq:eqbwd})$ by inserting the definition~\eqref{eq:lipr2} of $\mu$.
\end{proof}
 
 Finally, we provide a well-posedness result for the FBSDE system characterising the frictional equilibrium price, positions, and trading rates. In order to work with small processes for $\gamma^1 \approx \gamma^2$, we pass from from the frictional equilibrium price $S$ to its deviation $Y=S-\bar{S}$ from its frictional counterpart $\bar{S}$. Subtracting \eqref{eq:bsdebis1}  from \eqref{eq:BSDES} and denoting the frictionless equilibrium volatility by $\bar{\sigma}$, we obtain the following BSDE for $Y$ which is coupled to (\ref{eq:eqfwd}--\ref{eq:eqbwd}):
\begin{align}
\mathrm{d}Y_t = &\bigg( \frac{\gamma^1 - \gamma^2}{2} (\bar \sigma_t + Z^Y_t)^2 \varphi^1_t  + \frac{\gamma^2s}{2} (Z^Y_t)^2  + Z^Y_t \bigg(\gamma^2s \bar \sigma_t +\frac{\gamma^1 \beta^1_t +\gamma^2 \beta^2_t}{2}\bigg)- \frac{\gamma^1 - \gamma^2}{2} \bar \sigma_t^2\bar \varphi^1_t \bigg)\mathrm{d}t\notag\\
  &\quad+ Z^Y_t \mathrm{d} W_t, \qquad Y_T= 0,
\label{eq:BSDEY}
\end{align}
where 
\[\bar{\varphi}^1:=\frac{\gamma^2s}{\gamma^1+\gamma^2}+\frac{\gamma^2\beta^2-\gamma^1\beta^1}{(\gamma^1+\gamma^2)\bar\sigma},\]
denotes the frictionless equilibrium position of agent $1$. Well-posedness of the system  (\ref{eq:eqfwd}--\ref{eq:eqbwd}, \ref{eq:BSDEY}) will be a special case of Theorem~\ref{thm:ex} below. 
\begin{theorem}\label{thm:ex}
	Let $(\gamma^1, \gamma^2, \tilde \gamma, \kappa,\bar \sigma, \nu, \alpha,\nu^\prime) \in(0,\infty)^4\times \big(\mathbb{H}^2_{\mathrm{BMO}}\big)^3\times \mathcal{S}^\infty$. Define the measure $\mathbb P^\alpha \sim \mathbb P$ by $\frac{\mathrm{d} \mathbb P^\alpha}{\mathrm{d} \mathbb P} := \mathcal{E}\big(\int_0^\cdot \alpha_s \mathrm{d}W_s\big)_T$, and assume that for some $p\in(1,2)$, $\mathbb{E}^{\mathbb P^\alpha}\big[\mathrm{e}^{\frac{2p}{2-p}\int_0^T \alpha^2_u \mathrm{d} u}\big] < \infty$. 
	Let 
	\[R <  \min\bigg\{\|\bar \sigma\|_{\mathbb{H}^2_{\mathrm{BMO}}(\mathbb P^\alpha)}, \frac{1}{8 \kappa}\bigg\}=:R_{\rm max},\] 
and assume that
	\begin{align*}
|\gamma^1-\gamma^2| < \|\bar\sigma\|_{\mathbb{H}^2_{\mathrm{BMO}}(\mathbb P^\alpha)}^{-1} \min \Bigg(&\frac{R(1/\sqrt{2}-2 \kappa R)}{4 \|\bar\sigma\|_{\mathbb{H}^2_{\mathrm{BMO}}(\mathbb P^\alpha)} h_{\varphi}(x,\tilde \gamma, \alpha, \bar \sigma, \nu, \nu') + \|\nu\|_{\mathbb H^2_{\mathrm{BMO}}(\mathbb P^\alpha)}+ \|\nu'\|_{\mathcal{S}^{\infty}}\|\bar \sigma\|_{\mathbb H^2_{\mathrm{BMO}}(\mathbb P^\alpha)}}, \\
& \frac{ 1- 8\kappa R}{8\|\bar \sigma \|_{\mathbb{H}^2_{\mathrm{BMO}}(\mathbb P^\alpha)} \big(g_\varphi(x, \tilde \gamma, \alpha, \sqrt{2} \bar \sigma, \sqrt{2} \bar \sigma, \nu, \nu^\prime)^{\frac{1}{2}} + h_\varphi(x, \tilde \gamma, \alpha, \bar \sigma, \nu, \nu^\prime)\big)}\Bigg)=:\varepsilon_{\max},
	\end{align*}	
	where
	\begin{align*}
		h_\varphi(x, \tilde \gamma, \alpha, \bar \sigma, \nu, \nu^\prime)&:=|x| + 32  \tilde \gamma T\left( \|\nu\|_{\mathbb H^2_{\mathrm{BMO}}(\mathbb P^\alpha)}\|\bar \sigma\|_{\mathbb H^2_{\mathrm{BMO}}(\mathbb P^\alpha)} + \|\nu'\|_{\mathcal{S}^{\infty}}\|\bar \sigma\|^2_{\mathbb H^2_{\mathrm{BMO}}(\mathbb P^\alpha)}  \right) \big(\|\alpha\|_{\mathbb{H}^2_{\mathrm{BMO}}(\mathbb{P})} + 1\big)^2 ,
		\end{align*}
		and $g_\varphi$ is defined in Theorem \ref{thm:stabilityphi}. Then, the system of coupled forward-backward {\rm SDEs} 
	\begin{alignat}{2}
		\mathrm{d}\varphi_t &= \dot{\varphi}_t\mathrm{d}t, \qquad && \varphi_0=x, \label{eq:thm:ex:fwd}\\
		\mathrm{d}\dot{\varphi}_t &=  \tilde \gamma (\bar \sigma_t + Z_t)^2\bigg(\varphi_t -  \frac{\nu_t}{\bar \sigma_t + Z_t}- \nu^\prime_t  \bigg) \mathrm{d}t +\dot{Z}_t\mathrm{d}W_t, \qquad && \dot{\varphi}^1_T=0, \label{eq:thm:ex:bwd}\\
		\mathrm{d}Y_t &= \bigg(\frac{\gamma^1-\gamma^2}{2} (\bar{\sigma}_t+Z_t)^2 \varphi_t + \kappa Z_t^2 -\alpha_t Z_t  - \frac{\gamma^1-\gamma^2}{2} \bar \sigma_t^2 \left(\frac{\nu_t}{\bar \sigma_t}+ \nu^\prime_t \right)\bigg) + Z_t \mathrm{d}W_t,\qquad && Y_T=0, \label{eq:thm:ex:Y}
	\end{alignat}
	has a unique solution for $(Y, Z)$ lying inside a ball of radius $R$ for the norm on $\mathcal{S}^\infty\times\mathbb H^2_{\rm BMO}(\mathbb P^\alpha)$. Moreover $\varphi$ and $\dot{\varphi}$ are both uniformly bounded. For
	\[
\tilde\gamma:=\frac{\gamma^1+\gamma^2}{2\lambda},\quad \nu:=\frac{\gamma^2\beta^2-\gamma^1\beta^1}{\gamma^1+\gamma^2},\quad \nu^\prime:=\frac{\gamma^2s}{\gamma^1+\gamma^2},\quad \kappa:=\frac{\gamma^2s}{2},\quad \alpha:=-\gamma^2s\bar\sigma-\frac{\gamma^1\beta^1+\gamma^2\beta^2}2,
\]
the solution to $(\ref{eq:thm:ex:fwd}$--$\ref{eq:thm:ex:Y})$ provides the unique solution of the {\rm FBSDEs} $(\ref{eq:eqfwd}$--$\ref{eq:BSDES})$ for which $(S-\bar{S},\sigma-\bar{\sigma})$ lies inside a ball of radius $R$ on $\mathcal{S}^\infty\times\mathbb H^2_{\rm BMO}(\mathbb Q^\beta)$ by defining $S:=\bar S_t+ Y_t$, $\sigma:=\bar\sigma+Z$.
\end{theorem}

Before proving the theorem, let us briefly relate our system \eqref{eq:thm:ex:fwd}--\eqref{eq:thm:ex:Y} to the existing literature. It belongs to two main strands:
\begin{itemize}
\item[$(i)$] degenerate fully coupled FBSDEs, since the forward process $\varphi$ has bounded variation and appears in the generators of the backward equations, and the component $\dot{\varphi}$ of the backward part appears in the drift of the forward equation;

\item[$(ii)$] multidimensional BSDEs with quadratic growth, since both generators of the backward equations have quadratic growth in the $Z$ component.
\end{itemize}
Both types of equations are already very challenging by themselves. Despite having been studied for almost 30 years, there still does not exist a general theory even for simpler fully--coupled FBSDEs with Lipschitz generators and with a one-dimensional backward component, the closest being \cite{ma2015well} which unifies several existing approaches (also compare \cite{ankirchner2019transformation} for some very recent progress in this direction). This approach is however limited to the one-dimensional setting, which automatically excludes our system. Other results applicable to Lipschitz or locally Lipschitz multi--dimensional FBSDEs have been proposed, notably in \cite{antonelli2006existence,fromm2013existence}, but under monotonicity conditions which do not hold in our context, or for proving existence of a solution over a maximal interval which in general will be strictly smaller than $[0,T]$. 

Similarly, the analysis of multi-dimensional (uncoupled) quadratic BSDEs is involved in its own right and also relies on assumptions about the structure of the problem at hand; we refer to the most general results to date \cite{xing2016class, harter2019stability} for more details. 

Evidently, settings that combine aspects (i) and (ii) above are even more challenging to deal with. As far as we know, the only works addressing multi-dimensional fully coupled quadratic FBSDEs are the references \cite{antonelli2006existence,fromm2013existence} already mentioned above, \cite{luo2017solvability} which considers diagonally quadratic generators (an assumption not satisfied in our context), as well as \cite{luo2019multidimensional}, which considers the Markovian case and obtains global existence under a uniform non-degeneracy assumption for the volatility of the forward process (that does not hold for our system).

Our approach borrows ideas from the existing literature, notably the fixed-point argument of Tevzadze \cite{tevzadze.08}. But more importantly, our approach exploits the specific structure of our problem to obtain an existence result that is global in time. The main difficulty lies in the fact that a naive Picard iteration for all three components of the FBSDE (\ref{eq:thm:ex:fwd}--\ref{eq:thm:ex:Y}) does not work. Indeed, because of the quadratic nature of the problem, we want to use BMO-type arguments.  To this end, we have to ensure that each step of the iteration remains in a sufficiently small ball (for the appropriate norms). This is feasible for \eqref{eq:thm:ex:Y}, since we assume that $\gamma^1-\gamma^2$ is small. However, there is no reason to expect that successive Picard iterations of \eqref{eq:thm:ex:fwd} and \eqref{eq:thm:ex:bwd} remain small -- unless the time horizon is also sufficiently small, which we do not want to assume because costs on the turnover rate than essentially lead to a no-trade equilibrium. The key idea to overcome this issue is to use the specific structure of our problem and to realise that one should only perform the iteration on \eqref{eq:thm:ex:Y}, and use our well-posedness result for (\ref{eq:eqfwd}--\ref{eq:eqbwd}), using the $Z$ given in each step of the iteration. Finally, the very precise estimates and stability results developed in Section~\ref{sec:stab} then allow us to obtain a desired contraction property.

\begin{proof}
We first establish two a priori estimates that will be used throughout the proof. Let $Z \in \mathbb{H}^2_{\mathrm{BMO}}(\mathbb P^\alpha)$ with  $\|Z\|_{\mathbb{H}^2_{\mathrm{BMO}}(\mathbb P^\alpha)} \leq \|\bar \sigma\|_{\mathbb{H}^2_{\mathrm{BMO}}(\mathbb P^\alpha)}$. Then by the elementary inequality $(a + b)^2 \leq 2(a^2 + b^2)$ for $(a, b) \in \mathbb R^2$, we have
\begin{equation}
\label{eq:pf:thm:ex:sigma + Z}
\|\bar \sigma + Z \|^2_{\mathbb{H}^2_{\mathrm{BMO}}(\mathbb P^\alpha)} \leq 2 \left(\|\bar \sigma\|^2_{\mathbb{H}^2_{\mathrm{BMO}}(\mathbb Q)} + \|Z\|^2_{\mathbb{H}^2_{\mathrm{BMO}}(\mathbb P^\alpha)}\right) \leq 4\|\bar \sigma\|^2_{\mathbb{H}^2_{\mathrm{BMO}}(\mathbb P^\alpha)}.
\end{equation}
Moreover, Corollary \ref{cor:boundphi}, Lemma \ref{lem:BMO P Q}, and \eqref{eq:pf:thm:ex:sigma + Z} show that the FBSDE \eqref{eq:thm:ex:fwd}--\eqref{eq:thm:ex:bwd} (with this fixed $Z$) has a bounded solution such that $\varphi$ satisfies the estimate
\begin{align}
\left\|\varphi\right\|_{\mathcal{S}^\infty} &\leq  |x| + \tilde \gamma T\left( \|\nu\|_{\mathbb H^2_{\mathrm{BMO}}(\mathbb P)}\|\bar \sigma + Z\|_{\mathbb H^2_{\mathrm{BMO}}(\mathbb P)} + \|\nu'\|_{\mathcal{S}^{\infty}}\|\bar \sigma + Z\|^2_{\mathbb H^2_{\mathrm{BMO}}(\mathbb P)}  \right) \notag \\
&\leq |x| +  8 \tilde \gamma T\left( \|\nu\|_{\mathbb H^2_{\mathrm{BMO}}(\mathbb P^\alpha)}\|\bar \sigma + Z\|_{\mathbb H^2_{\mathrm{BMO}}(\mathbb P^\alpha)} + \|\nu'\|_{\mathcal{S}^{\infty}}\|\bar \sigma + Z\|^2_{\mathbb H^2_{\mathrm{BMO}}(\mathbb P^\alpha)}  \right) \big(\|\alpha\|_{\mathbb{H}^2_{\mathrm{BMO}}(\mathbb{P})} + 1\big)^2 \notag \\
&\leq |x| + 32  \tilde \gamma T\left( \|\nu\|_{\mathbb H^2_{\mathrm{BMO}}(\mathbb P^\alpha)}\|\bar \sigma\|_{\mathbb H^2_{\mathrm{BMO}}(\mathbb P^\alpha)} + \|\nu'\|_{\mathcal{S}^{\infty}}\|\bar \sigma\|^2_{\mathbb H^2_{\mathrm{BMO}}(\mathbb P^\alpha)}  \right) \big(\|\alpha\|_{\mathbb{H}^2_{\mathrm{BMO}}(\mathbb{P})} + 1\big)^2 \\
&=: h_{\varphi}(x, \tilde \gamma, \alpha, \bar \sigma, \nu, \nu').
\label{eq:pf:thm:ex:a priori}
\end{align}
Next, let $Z^{0}:=0$, and define $(\varphi^{1},\dot\varphi^{1})$ as the solution of the FBSDEs~\eqref{eq:thm:ex:fwd}--\eqref{eq:thm:ex:bwd}, corresponding to the volatility $\bar{\sigma}+Z^0 \in \mathbb{H}^2_{\mathrm{BMO}}(\mathbb P^\alpha)$, and $(Y^{1},Z^1)$ as the solution of 
\[
\mathrm{d}Y^1_t= \bigg((\bar{\sigma}_t+Z^0_t)^2 \frac{\gamma^1-\gamma^2}{2}\varphi^1_t+\kappa(Z^0_t)^2 - \frac{\gamma^1-\gamma^2}{2} \bar \sigma_t^2 \left(\frac{\nu_t}{\bar \sigma_t}+ \nu^\prime_t\right) \bigg)\mathrm{d}t+Z^1_t \mathrm{d}W^{\mathbb P^\alpha}_t, \quad Y^1_T=0.
\]
By the a priori estimate \eqref{eq:pf:thm:ex:a priori}, we know that $\varphi^1$ is bounded. This implies that $(Y^1,Z^1)$ is well defined and belongs to $\mathcal{S}^{\infty} \times \mathbb{H}^2_{\mathrm{BMO}}(\mathbb P^\alpha)$.

\vspace{0.5em}
For $n \geq 2$, we continue by induction. Given $(Y^{n-1}, Z^{n-1}) \in \mathcal{S}^{\infty} \times \mathbb{H}^2_{\mathrm{BMO}}(\mathbb P^\alpha)$, let $\varphi^{n},\dot\varphi^{n}$ be defined as the solution of the FBSDEs \eqref{eq:thm:ex:fwd}--\eqref{eq:thm:ex:bwd} corresponding to the volatility $\bar{\sigma}+Z^{n-1}\in \mathbb{H}^2_{\mathrm{BMO}}(\mathbb P^\alpha)$, and $(Y^{n},Z^n) \in \mathcal{S}^{\infty} \times \mathbb{H}^2_{\mathrm{BMO}}(\mathbb P^\alpha)$ as the solution of 
\[
\mathrm{d}Y^n_t =  \bigg((\bar{\sigma}_t+Z^{n-1}_t)^2 \frac{\gamma^1-\gamma^2}{2}\varphi^1_t+\kappa(Z^{n-1}_t)^2 - \frac{\gamma^1-\gamma^2}{2} \bar \sigma_t^2 \left(\frac{\nu_t}{\bar \sigma_t}+ \nu^\prime_t\right) \bigg)\mathrm{d}t+Z^n_t \mathrm{d}W^{\mathbb P^\alpha}_t, \quad Y^n_T=0.
\]
We proceed to show that for sufficiently small $|\gamma^1-\gamma^2|$, this iteration is a contraction on $ \mathcal{S}^{\infty} \times \mathbb{H}^2_{\mathrm{BMO}}(\mathbb P^\alpha)$. By the Banach fixed point theorem, it therefore has a unique fixed point $(Y,Z)$. Together with the pair $(\varphi,\dot\varphi)$ that solves the tracking problem corresponding to the volatility $\bar\sigma+Z$, we have in turn constructed the desired solution of (\ref{eq:thm:ex:fwd}--\ref{eq:thm:ex:Y}).

\medskip
To establish that our mapping is indeed a contraction, we first show as in \cite{tevzadze.08} that it maps sufficiently small balls in $\mathcal{S}^\infty \times \mathbb{H}^2_{\mathrm{BMO}}(\mathbb P^\alpha)$ into themselves. To this end, suppose that 
\[\|Y^{n-1}\|_{\mathcal{S}^\infty}^2 + \|Z^{n-1}\|^2_{\mathbb{H}^2_{\mathrm{BMO}}(\mathbb P^\alpha)} \leq R^2,\]
where we recall that  $R < \min(\|\bar \sigma\|_{\mathbb{H}_{\mathrm{BMO}}(\mathbb P^\alpha)},\frac{1}{4\sqrt{2}x \kappa}) $. Apply It\^o's formula to $(Y^n)^2$ and use that $Y^n_T = 0$. Then take conditional $\mathbb{Q}$-expectation and use that $Y^n$ is bounded and $Z^n \in \mathbb{H}^2_{\mathrm{BMO}}(\mathbb P^\alpha)$. For any stopping time $\tau$ with values in $[0,T]$, this gives
\begin{align}
0=(Y^n_\tau)^2 + \mathbb E^\mathbb Q_\tau\bigg[\int_\tau^T (Z^n_s)^2 \mathrm{d}s\bigg] &+\mathbb E^{\mathbb P^\alpha}_\tau\bigg[\int_\tau^T 2Y_s^n(\bar{\sigma}_s+Z^{n-1}_s)^2\frac{\gamma^1-\gamma^2}{2}\varphi^{n}_s+2Y^n_s \kappa(Z^{n-1}_s)^2\mathrm{d}s\bigg] \notag \\
&- \mathbb E^{\mathbb P^\alpha}_\tau\bigg[\int_\tau^T 2Y_s^n \frac{\gamma^1-\gamma^2}{2} (\bar \sigma_s)^2\left(\frac{\nu_s}{\bar \sigma_s}+ \nu^\prime_s\right)  \mathrm{d}s\bigg].
\end{align}
Now use that $Y^n \in \mathcal{S}^\infty$ and $\|Z^{n-1}\|_{\mathbb{H}^2_{\mathrm{BMO}}(\mathbb P^\alpha)} \leq R$. Together with the a priori estimates \eqref{eq:pf:thm:ex:sigma + Z} and \eqref{eq:pf:thm:ex:a priori}, this yields
\begin{align}
&\nonumber (Y^n_\tau)^2+ \mathbb E^{\mathbb P^\alpha}_\tau\bigg[\int_\tau^T (Z^n_s)^2 \mathrm{d}s\bigg] \leq  |\gamma^1-\gamma^2| \|Y^n\|_{\mathcal{S}^{\infty}} \|\bar\sigma + Z^{n-1}\|^2_{\mathbb{H}^2_{\mathrm{BMO}}(\mathbb P^\alpha)}\left\|\varphi^{n}\right\|_{\mathcal{S}^\infty} +2\kappa \|Y^n\|_{\mathcal{S}^{\infty}}\|Z^{n-1}\|_{\mathbb{H}^2_{\mathrm{BMO}}(\mathbb P^\alpha)}^2 \notag  \\
&\quad\quad\quad\quad\quad\quad\quad\quad\quad\quad\quad\quad\quad\quad +  |\gamma^1-\gamma^2| \|Y^n\|_{\mathcal{S}^{\infty}} \left( \|\nu\|_{\mathbb H^2_{\mathrm{BMO}}(\mathbb Q)}\|\bar \sigma\|_{\mathbb H^2_{\mathrm{BMO}}(\mathbb P^\alpha)} + \|\nu'\|_{\mathcal{S}^{\infty}}\|\bar \sigma\|^2_{\mathbb H^2_{\mathrm{BMO}}(\mathbb P^\alpha)}  \right) \notag \\
&\leq\ \|Y^n\|_{\mathcal{S}^{\infty}}   \bigg( |\gamma^1-\gamma^2| \left(4 \|\bar\sigma\|^2_{\mathbb{H}^2_{\mathrm{BMO}}(\mathbb P^\alpha)} \left\|\varphi^{n}\right\|_{\mathcal{S}^\infty} + \|\nu\|_{\mathbb H^2_{\mathrm{BMO}}(\mathbb P^\alpha)}\|\bar \sigma\|_{\mathbb H^2_{\mathrm{BMO}}(\mathbb P^\alpha)} + \|\nu'\|_{\mathcal{S}^{\infty}}\|\bar \sigma\|^2_{\mathbb H^2_{\mathrm{BMO}}(\mathbb P^\alpha)} \right)+2 \kappa R^2\bigg) \notag \\
&=: \|Y^n\|_{\mathcal{S}^{\infty}} \left( |\gamma^1-\gamma^2| h_R(x,\tilde \gamma, \alpha, \bar \sigma, \nu, \nu')+2 \kappa R^2\right),
\label{eq:pf:thm:ex:01}
\end{align}
where
\begin{align*}
h_R(x,\tilde \gamma, \alpha, \bar \sigma, \nu, \nu') &:= 4 \|\bar\sigma\|^2_{\mathbb{H}^2_{\mathrm{BMO}}(\mathbb P^\alpha)} h_{\varphi}(x,\tilde \gamma, \alpha, \bar \sigma, \nu, \nu') + \|\nu\|_{\mathbb H^2_{\mathrm{BMO}}(\mathbb P^\alpha)}\|\bar \sigma\|_{\mathbb H^2_{\mathrm{BMO}}(\mathbb P^\alpha)} + \|\nu'\|_{\mathcal{S}^{\infty}}\|\bar \sigma\|^2_{\mathbb H^2_{\mathrm{BMO}}(\mathbb P^\alpha)}.
\end{align*}
Taking the supremum over all $\tau$ (for $Y^n$) and rearranging yields
\begin{align}
\label{eq:pf:thm:ex:cond Y}
\|Y^n\|_{\mathcal{S}^{\infty}} &\leq |\gamma^1-\gamma^2|h_R(x, \tilde \gamma, \alpha, \bar \sigma, \nu, \nu^\prime)+ 2 \kappa R^2.
\end{align}
Now taking the supremum over all $\tau$ in \eqref{eq:pf:thm:ex:01} (for $Z^n$) and using \eqref{eq:pf:thm:ex:cond Y}, we obtain
\begin{align}
\label{eq:pf:thm:ex:cond Z}
\|Z^n\|^2_{\mathbb{H}^2_{\mathrm{BMO}}(\mathbb P^\alpha)} &\leq \big( |\gamma^1-\gamma^2|h_R(x, \tilde \gamma, \alpha, \bar \sigma, \nu, \nu^\prime)+ 2 \kappa R^2\big)^2 .
\end{align}
Using our bounds on $|\gamma^1-\gamma^2|$, and the fact that $R \leq \frac{1}{4 \sqrt{2} \kappa}$, we deduce that
\[\|Y^{n}\|_{\mathcal{S}^\infty}^2 + \|Z^{n}\|^2_{\mathbb{H}^2_{\mathrm{BMO}}(\mathbb P^\alpha)} \leq 2 \bigg( |\gamma^1-\gamma^2|h_R(x, \tilde \gamma, \alpha, \bar \sigma, \nu, \nu^\prime)+ 2 \kappa R^2\bigg)^2 \leq R^2.\]
We now show that our iteration is a contraction on the ball $B_R$ in $\mathcal{S}^{\infty} \times \mathbb{H}^2_{\mathrm{BMO}}(\mathbb P^\alpha)$. To this end, consider $(y,z), ((y^\prime,z^\prime)) \in B_R^2,$ and write $(Y,Z)$, $(Y^\prime,Z^\prime)$ for their images produced by our iteration. Also denote by $(\varphi,\dot{\varphi})$, $(\varphi^\prime,\dot{\varphi}^\prime)$ the corresponding optimal tracking strategies (corresponding to volatilities $\bar\sigma+z$ and $\bar\sigma+z^\prime$, respectively). To verify that our iteration is indeed a contraction, we have to show that for some $\eta \in(0,1)$,
\[
\|Y-Y^\prime\|_{\mathcal{S}^\infty}^2+\|Z-Z^\prime\|_{\mathbb{H}^2_{\mathrm{BMO}}(\mathbb P^\alpha)}^2 \leq \eta \Big(\|y-y^\prime\|_{\mathcal{S}^\infty}^2+\|z-z^\prime\|^2_{\mathbb{H}^2_{\mathrm{BMO}}(\mathbb P^\alpha)}\Big).
\]
To ease notation, set
\[ \delta y:=y-y^\prime,\quad \delta z:=z-z^\prime,\quad \delta Y:=Y-Y^\prime,\quad \delta Z:=Z-Z^\prime.\] 
Applying It\^o's formula on $[\tau,T]$ for any $[0,T]$-valued stopping time $\tau$, inserting the dynamics of $Y$ and $Y^\prime$, taking $\mathbb P^\alpha$-conditional expectations, and using the identity $a b - c d = a(b-d) + (a -c) d$ for $(a, b, c, d) \in \mathbb{R}^4$, we obtain
\begin{align}
\delta Y_\tau^2 +\mathbb E^{\mathbb P^\alpha}_\tau\bigg[\int_\tau^T \delta Z_t^2 \mathrm{d}t\bigg]  &= \mathbb E^{\mathbb P^\alpha}_\tau\bigg[(\gamma^2-\gamma^1)\int_\tau^T \delta Y_t\Big((\bar{\sigma}_t+z_t)^2 \varphi_t-(\bar{\sigma}_t+z^\prime_t)^2 \varphi^\prime_t\Big) \mathrm{d}t - 2\kappa \int_\tau^T \delta Y_t\big((z_t)^2-(z'_t)^2\big)\mathrm{d}t\bigg]\notag\\
& \leq \|\delta Y\|_{\mathcal{S}^{\infty}} |\gamma^1-\gamma^2| \bigg(\mathbb E^{\mathbb P^\alpha}_\tau\bigg[\int_\tau^T (\bar{\sigma}_t+z_t)^2 |\varphi_t-\varphi^\prime_t| \mathrm{d}t+\int_\tau^T \big|2\bar{\sigma}_t+z_t+z^\prime_t\big||\delta z_t| |\varphi^\prime_t|\mathrm{d}t\bigg]\bigg)\notag \\
&\quad + 2 \kappa \|\delta Y\|_{\mathcal{S}^{\infty}} \mathbb E^{\mathbb P^\alpha}_\tau\bigg[\int_\tau^T |z_t+z^\prime_t||\delta z_t| \mathrm{d}t\bigg]. \label{eq:pf:thm:ex:delta estimate}
\end{align}
To estimate the conditional expectation in the first term on the right-hand side of \eqref{eq:pf:thm:ex:delta estimate}, define the process 
\[A_{t}:=\underset{u \in [0, t]}{\sup} |\varphi_u-\varphi^\prime_u|,\quad t\in[0,T].\]
Lemma \ref{lem:increasing process}, \eqref{eq:pf:thm:ex:sigma + Z}, Jensen's inequality, and Theorem \ref{thm:stabilityphi} in turn yield
\begin{align}
\mathbb E^{\mathbb P^\alpha}_\tau\bigg[\int_\tau^T (\bar{\sigma}_t+z_t)^2 |\varphi_t-\varphi^\prime_t| \mathrm{d}t\bigg] &\leq \mathbb E^{\mathbb P^\alpha}_\tau\bigg[\int_\tau^T (\bar{\sigma}_t+z_t)^2 A_t \mathrm{d}t\bigg]  \leq 4\|\bar \sigma \|^2_{\mathbb{H}^2_{\mathrm{BMO}}(\mathbb P^\alpha)}  \mathbb E^{\mathbb P^\alpha}_\tau \bigg[\sup_{u \in [0, T]}|\varphi_u-\varphi^\prime_u| \bigg] \notag \\
&\leq 4\|\bar \sigma \|^2_{\mathbb{H}^2_{\mathrm{BMO}}(\mathbb P^\alpha)}  g_\varphi(x, \tilde \gamma, \alpha, \sqrt{2} \bar \sigma, \sqrt{2} \bar \sigma, \nu, \nu')^{\frac{1}{2}} \|\delta z \|_{\mathbb{H}^2_{\mathrm{BMO}}(\mathbb P^\alpha)}.
 \label{eq:pf:thm:ex:delta estimate:one}
\end{align}
To estimate the conditional expectation in the second term on the right-hand side of \eqref{eq:pf:thm:ex:delta estimate}, we use that $\varphi^\prime \in \mathcal{S}^\infty$,
the conditional version of the Cauchy--Schwarz inequality and the elementary inequality $(a + b + c)^2 \leq 2 a^2 + 4 b^2 + 4 c^2$. Together with the fact that both $\|z\|^2_{\mathbb{H}^2_{\mathrm{BMO}}(\mathbb P^\alpha)}$ and  $\|z^\prime\|^2_{\mathbb{H}^2_{\mathrm{BMO}}(\mathbb P^\alpha)} $ are smaller than  $\|\bar \sigma \|^2_{\mathbb{H}^2_{\mathrm{BMO}}(\mathbb P^\alpha)} $ and the a priori estimate \eqref{eq:pf:thm:ex:a priori}, this yields
\begin{align}
\mathbb E^{\mathbb P^\alpha}_\tau\bigg[\int_\tau^T |2\bar{\sigma}_t+z_t+z^\prime_t||\delta z_t| |\varphi^\prime_t|\mathrm{d}t\bigg] &\leq \|\varphi^\prime\|_{\mathcal{S}^\infty} \mathbb E^{\mathbb P^\alpha}_\tau\bigg[\int_\tau^T (2\bar{\sigma}_t+z_t+z^\prime_t)^2 \mathrm{d}t\bigg]^{\frac{1}{2}}  \mathbb E^{\mathbb P^\alpha}_\tau\bigg[\int_\tau^T \delta z_t^2 \mathrm{d}t\bigg]^{\frac{1}{2}} \notag \\
&\leq 4 h_{\varphi}(x, \tilde \gamma, \alpha, \bar \sigma, \nu, \nu^\prime)  \|\bar \sigma \|_{\mathbb{H}^2_{\mathrm{BMO}}(\mathbb P^\alpha)} \|\delta z \|_{\mathbb{H}^2_{\mathrm{BMO}}(\mathbb P^\alpha)} .
 \label{eq:pf:thm:ex:delta estimate:two}
\end{align}
To estimate the the conditional expectation in the third term on the right-hand side of \eqref{eq:pf:thm:ex:delta estimate}, we argue in a similar fashion and obtain
\begin{align}
\mathbb E^{\mathbb P^\alpha}_\tau\bigg[\int_\tau^T |z_t+z^\prime_t||\delta z_t| \mathrm{d}t\bigg] &\leq 2  R \|\delta z \|_{\mathbb{H}^2_{\mathrm{BMO}}(\mathbb P^\alpha)}.
\label{eq:pf:thm:ex:delta estimate:three}
\end{align}
Now, plugging $\eqref{eq:pf:thm:ex:delta estimate:one}$--$\eqref{eq:pf:thm:ex:delta estimate:three}$ into \eqref{eq:pf:thm:ex:delta estimate}, taking the supremum over all $\tau$ (both for $Y$ and $Z$), then taking conditional $\mathbb{P}^\alpha$-expectations, applying Lemma \ref{lem:Doob 1} and Theorem \ref{thm:stabilityphi} (together with \eqref{eq:pf:thm:ex:sigma + Z}), and using the elementary inequality $2a b \leq \frac{1}{\varepsilon} a^2+ \varepsilon b^2$ for $\varepsilon > 0$ yields
\begin{align}
 \|\delta Y\|^2_{\mathcal{S}^\infty} +\| \delta Z\|^2_{\mathbb H^2_{\rm BMO}(\mathbb P^\alpha)} \leq&\  8 |\gamma^1-\gamma^2|\|\delta Y\|_{\mathcal{S}^{\infty}} \|\bar \sigma \|^2_{\mathbb{H}^2_{\mathrm{BMO}}(\mathbb P^\alpha)} g_\varphi(x, \tilde \gamma, \alpha, \sqrt{2} \bar \sigma, \sqrt{2} \bar \sigma, \nu, \nu^\prime)^{\frac{1}{2}} \| \delta z\|_{\mathbb H^2_{\rm BMO}(\mathbb P^\alpha)} \notag \\
 &+ 8 |\gamma^1-\gamma^2|\|\delta Y\|_{\mathcal{S}^{\infty}} \|\bar \sigma \|^2_{\mathbb{H}^2_{\mathrm{BMO}}(\mathbb P^\alpha)} k_\varphi(x, \tilde \gamma, \alpha, \bar \sigma, \nu, \nu^\prime) \| \delta z\|_{\mathbb H^2_{\rm BMO}(\mathbb P^\alpha)} \notag \\
 & +8 \kappa \|\delta Y\|_{\mathcal{S}^{\infty}} R \|\delta z \|_{\mathbb{H}^2_{\mathrm{BMO}}(\mathbb P^\alpha)} \notag \\
\leq&\ \frac1\varepsilon \|\delta Y\|^2_{\mathcal{S}^\infty} + \varepsilon \eta^2 \| \delta z\|^2_{\mathbb H^2_{\rm BMO}(\mathbb P^\alpha)},
\end{align}
where
\begin{equation*}
\eta: = 4 |\gamma^1-\gamma^2|\|\bar \sigma \|^2_{\mathbb{H}^2_{\mathrm{BMO}}(\mathbb P^\alpha)} \left(g_\varphi(x, \tilde \gamma, \alpha, \sqrt{2} \bar \sigma, \sqrt{2} \bar \sigma, \nu, \nu^\prime)^{\frac{1}{2}} + h_\varphi(x, \tilde \gamma, \alpha, \bar \sigma, \nu, \nu^\prime)\right)  + 4 \kappa R.
\end{equation*}
We deduce that for any $\varepsilon>1$,
\[ \|\delta Y\|^2_{\mathcal{S}^\infty} +\| \delta Z\|^2_{\mathbb H^2_{\rm BMO}(\mathbb P^\alpha)}\leq \|\delta Y\|^2_{\mathcal{S}^\infty} + \frac{\varepsilon}{\varepsilon -1}\| \delta Z\|^2_{\mathbb H^2_{\rm BMO}(\mathbb P^\alpha)} \leq \frac{\varepsilon^2}{\varepsilon-1}\eta^2 \| \delta z\|^2_{\mathbb H^2_{\rm BMO}(\mathbb P^\alpha)}.\]
We choose $\varepsilon=2$ and deduce the desired result, since by our assumptions
\begin{equation*}
\frac{\varepsilon^2}{\varepsilon-1}\eta^2=4 \eta^2<1.
\end{equation*}

For the last part of the result, observe that these specific parameter choices satisfy all the requirements in Theorem \ref{thm:ex} in view of Assumptions \ref{ass:int1} and \ref{assump:friction}. This gives us a unique solution to the associated FBSDE system  (\ref{eq:eqfwd}--\ref{eq:eqbwd}, \ref{eq:BSDEY}). Any solution of (\ref{eq:eqfwd}--\ref{eq:eqbwd}, \ref{eq:BSDEY}) in turn provides a solution to $(\ref{eq:eqfwd}$--$\ref{eq:BSDES})$ by defining $S:=\bar S_t+ Y_t$ and $\sigma:=\bar\sigma+Z$. The converse is obviously true for solutions as in Theorem~\ref{thm:ex3}. 
 \qedhere
\end{proof}

We now prove Proposition~\ref{prop:odes}, which characterises equilibria with transaction costs via coupled systems of Riccati ODEs in a particular model with linear state dynamics and terminal condition.

\begin{proof}[Proof of Proposition~\ref{prop:odes}]
First notice that the functions $A(t)$, $D(t)$ satisfy the following Riccati equations:
\begin{alignat*}{2}
A'(t) &=-\bar{\gamma}sa^2 +\frac{\gamma^2 s}{2}(a+B(t))^2 - C(t)D(t), \qquad && A(T)=0,\\
D'(t) &= -\frac{\gamma^2s}{2\lambda}(a+B(t))^2-F(t)D(t), \qquad && D(T)=0. 
\end{alignat*}
Together with the Riccati ODEs for the functions $B(t)$, $C(t)$, $E(t)$, $F(t)$, it follows that the functions
\begin{align*}
f(t,x,y) =A(t)+B(t)x+C(t)y, \qquad g(t,x,y) =D(t)+E(t)x+F(t)y,
\end{align*}
solve the following semilinear PDEs (here, the arguments $(t,x,y)$ are omitted to ease notation):
\begin{align*}
f_t +\frac{1}{2} f_{xx} +f_y g &=-\bar \gamma sa^2 + \frac{\gamma^2}{2}(a+ B)^2 + \frac{\gamma_1 - \gamma_2}{2} \beta (a + B) x +\frac{\gamma_1 - \gamma_2}{2}  (a + B)^2 y \\
 &= \frac{\gamma^1-\gamma^2}{2}(a+f_x)^2 y +\frac{\gamma^2}{2}f_x^2 +f_x\left(\gamma^2 a+\frac{\gamma^1-\gamma^2}{2}\beta x\right)-\frac{\gamma^1-\gamma^2}{2}a^2\left(\frac{\gamma^2s}{\gamma^1+\gamma^2}-\frac{\beta}{a}x\right),\\
g_t +\frac{1}{2} g_{xx}+ g_y g &=\frac{\gamma^1+\gamma^2}{2\lambda} (a+f_x)\beta x -\frac{\gamma^2s}{2\lambda}(a+f_x)^2+\frac{\gamma^1+\gamma^2}{2\lambda}(a+f_x)^2y, 
\end{align*}
on $[0,T) \times \mathbb{R}^2$, with terminal conditions $f(T,x,y)= g(T,x,y)=0$. By definition of $\varphi^1_t$,
\[\dot{\varphi}^1_t = g(t,W_t,\varphi^1_t).\]
Now set
\[
Y_t=f(t,W_t,\varphi^1_t), \qquad Z_t = f_x(t,W_t,\varphi^1_t)=B(t).
\]
Then, It\^o's formula, the PDEs for $f(t,x,y)$, $g(t,x,y)$, and the definition of $Z$ show that $\dot{\varphi}^1$, $Y$, $Z$, satisfy the BSDEs
\begin{align*}
\mathrm{d}\dot{\varphi}_t &= \frac{\gamma^1+\gamma^2}{2\lambda}\bigg(\beta W_t(a+Z_t) -\frac{\gamma^2s}{\gamma^1+\gamma^2}(a+Z_t)^2+(a+Z_t)^2\varphi^1_t\bigg)\mathrm{d}t+E(t)\mathrm{d}W_t,\\
\mathrm{d}Y_t &= \bigg(\frac{\gamma^1-\gamma^2}{2}(a+Z_t)^2\varphi^1_t +\frac{\gamma^2s}{2}Z_t^2 +Z_t\bigg(\gamma^2 sa+\frac{\gamma^1-\gamma^2}{2}\beta W_t\bigg)-\frac{\gamma^1-\gamma^2}{2}a^2\bigg(\frac{\gamma^2s}{\gamma^1+\gamma^2}-\frac{\beta}{a}W_t\bigg)\bigg)\mathrm{d}t+Z_t\mathrm{d}W_t,
\end{align*}
with terminal conditions $\dot{\varphi}^1_T=Y_T=0$. Together with the forward equation $\mathrm{d}\varphi^1_t=\dot{\varphi}^1_t \mathrm{d}t$, as well as the BSDE for the frictionless equilibrium price $\bar{S}$ from Proposition~\ref{prop:BSDE}, it follows that $S=\bar{S}+Y$, $\sigma=a + Z = \bar{\sigma}+Z$, $\dot{\varphi}^1$, $E$, and $\varphi^1$ indeed solve the forward--backward equations~$(\ref{eq:eqfwd}$--$\ref{eq:BSDES})$. Since the frictionless equilibrium volatility is constant here, $\bar{\sigma}=a$ and $Z_t =B(t)$ is deterministic, we evidently have $\sigma \in \mathbb{H}^2_{\mathrm{BMO}}$.  Since the Brownian motion $W$ has finite moments and zero autocorrelation function, one also readily verifies that $\dot{\varphi}^1 \in \mathbb{H}^2$. The assertion in turn follows from Proposition~\ref{prop:eq}.
\end{proof}

We now turn to the proof of Theorem \ref{thm:odeex}, which guarantees existence of the Riccati system from Proposition~\ref{thm:ex3} for sufficiently similar risk aversion parameters. The argument is very close in spirit to that of Theorem~\ref{thm:ex}. Indeed, we also obtain well-posedness of the system by a Picard iteration scheme which is devised so that the successive iterations remain in a sufficiently small ball. And in order to achieve this, a naive direct iteration of the four equations does not work unless the time horizon is sufficiently short. Instead, we have to start by studying separately the system satisfied by $C$, $E$, $F$ for fixed $B$, exactly as we did for (\ref{eq:thm:ex:fwd}--\ref{eq:thm:ex:bwd}), when $Z$ is fixed, in the proof of Theorem \ref{thm:ex}. After developing the necessary stability estimates, we can then proceed to the iteration for $B$ and obtain the desired result. This shows that the approach underlying Theorem~\ref{thm:ex} is not crucially tied to the stringent integrability assumptions imposed there to deal with a general setting, but can also be adapted to other specific settings on a case by case basis.

\begin{proof}[Proof of Theorem~\ref{thm:odeex}]
To ease notation, set 
\[
\hat\gamma:=\frac{\gamma^1+\gamma^2}{2},\qquad \varepsilon:=\gamma^1-\gamma^2,
\]
as well as
\begin{equation*}
\tilde B(t) := B(t) + a, \quad t \in [0, T].
\end{equation*}
\emph{Step 1: Dealing with $(C,E,F)$.} We start by giving ourselves some bounded map $\tilde B:[0,T]\to \mathbb R$ and analyse the following coupled system of ODEs on $[0,T]$:
\begin{equation}\label{eq:EF}
\begin{cases}
\displaystyle C^{\tilde B}(t)=-\int_t^T\bigg(\frac\varepsilon2\tilde B(s)^2-F^{\tilde B}(s)C^{\tilde B}(s)\bigg)\mathrm{d}s,\\
\displaystyle E^{\tilde B}(t)=-\int_t^T\bigg(\frac{\beta\hat\gamma}{\lambda}\tilde B(s)-F^{\tilde B}(s)E^{\tilde B}(s)\bigg)\mathrm{d}s,\\
\displaystyle F^{\tilde B}(t)=-\int_t^T\bigg(\frac{\hat\gamma}{\lambda}\tilde B(s)^2-(F^{\tilde B})^2(s)\bigg)\mathrm{d}s.
\end{cases}
\end{equation}
As $\tilde B$ is bounded, the equation for $F^{\tilde B}$ has a unique solution. Using that $0 \leq \frac{\hat\gamma}{\lambda}\tilde B(s)^2 \leq \frac{\hat\gamma}{\lambda}(\|\tilde B\|_\infty )^2$, the comparison theorem for ODEs gives the estimate
\begin{equation}\label{eq:boundF}
-\sqrt{\frac{\hat\gamma}{\lambda}}(\|\tilde B\|_\infty )\leq-\sqrt{\frac{\hat\gamma}{\lambda}}(\|\tilde B\|_\infty )\mathrm{tanh}\bigg(\sqrt{\frac{\hat\gamma}{\lambda}}(\|\tilde B\|_\infty\big)(T - t)\bigg)\leq F^{\tilde B}(t)\leq 0,\quad t\in[0,T].
\end{equation}
The ODEs for $E^{\tilde B}$ and $C^{\tilde B}$ are linear and have the unique solutions 
\[
E^{\tilde B}(t)=-\frac{\hat\gamma\beta}{\lambda}\int_t^T\tilde B(s)\mathrm{e}^{\int_t^sF^{\tilde B}(r)\mathrm{d}r}
\mathrm{d}s,\quad
 C^{\tilde B}(t)=-\frac\varepsilon2\int_t^T\tilde B(s)^2\mathrm{e}^{\int_t^sF^{\tilde B}(r)\mathrm{d}r}
\mathrm{d}s,\quad  t\in[0,T].
\]
In particular,  non-positivity of $F$ implies
\begin{align}\label{eq:boundE}
\big|E^{\tilde B}(t)\big| &\leq \frac{\hat\gamma \beta}{\lambda} \|\tilde B\|_\infty (T- t), \quad \big|C^{\tilde B}(t)\big| \leq \frac{|\varepsilon|}{2} \|\tilde B\|_\infty^2(T-t), \quad \mbox{for all $t \in [0, T]$.}
\end{align}
We will also need some stability results for these solutions with respect to variations of $\tilde B$. Fix thus two bounded functions $\tilde B$ and $\tilde B^\prime$. Using that $F^{\tilde B}-F^{\tilde B^\prime}$ satisfies the ODE
\begin{equation*}
\left(F^{\tilde B} - F^{\tilde{B}^\prime}\right)(t) = -\int_t^T\bigg(\frac{\hat\gamma}{\lambda}\big(\tilde B(s) + \tilde B^\prime(s)\big)(\tilde B(s) - \tilde B^\prime(s)) -\big(F^{\tilde B}(s) + F^{\tilde{B}^\prime}(s)\big)(F^{\tilde B} - F^{\tilde {B}^\prime})(s)\bigg)\mathrm{d}s,
\end{equation*}
we obtain
\begin{align*}
F^{\tilde B}(t)-F^{\tilde{B}^\prime}(t) &=-\frac{\hat\gamma}{\lambda}\int_t^T\mathrm{e}^{\int_t^s(F^{\tilde B}(r)+F^{B^\prime}(r))\mathrm{d}r}\big(\tilde B(s)+\tilde B^\prime(s)\big)\big(\tilde B(s)-\tilde B^\prime(s)\big)\mathrm{d}s.
\end{align*}
Non-positivity of $F^{\tilde B}$ and $F^{B^\prime}$ gives
\begin{equation*}
\big|F^{\tilde B}(t)-F^{\tilde B^\prime}(t)\big| \leq \frac{\hat\gamma}{\lambda} \big( \|\tilde B\|_\infty+\|\tilde B^\prime\|_\infty \big)\big\|\tilde B-\tilde B^\prime\big\|_\infty (T -t), \quad \mbox{for all $t \in [0, T]$.}
\end{equation*}
Using that the $x\mapsto\mathrm{e}^x$ is $1-$Lipschitz continuous on $(-\infty,0]$, this implies 
\begin{align}
\big|\mathrm{e}^{\int_t^sF^{\tilde B}(r)\mathrm{d}r}-\mathrm{e}^{\int_t^sF^{\tilde B^\prime}(r)\mathrm{d}r}\big| &\leq \frac{\hat\gamma}{\lambda} \big( \|\tilde B\|_\infty+\|\tilde B^\prime\|_\infty \big)\big\|\tilde B-\tilde B^\prime\big\|_\infty (T -t)^2, \quad \mbox{for $0 \leq t \leq s \leq T$.}
\label{eq:stabF:int}
\end{align}
Taking into account the explicit expressions for $E^{\tilde B}$ and $C^{\tilde B}$, we also deduce that
\[
E^{\tilde B}(t)-E^{\tilde B^\prime}(t)=-\frac{\hat\gamma\beta}{\lambda}\int_t^T\Big(\mathrm{e}^{\int_t^sF^{\tilde B}(r)\mathrm{d}r}\big(\tilde B(s)-\tilde B^\prime(s)\big)+\tilde B^\prime(s)\big(\mathrm{e}^{\int_t^sF^{\tilde B}(r)\mathrm{d}r}-\mathrm{e}^{\int_t^sF^{B^\prime}(r)\mathrm{d}r}\big)\Big)\mathrm{d}s,
\]
\[
C^{\tilde B}(t)-C^{\tilde B^\prime}(t)=-\frac\varepsilon2\int_t^T\Big(\mathrm{e}^{\int_t^sF^{\tilde B}(r)\mathrm{d}r}\big(\tilde B(s)+\tilde B^\prime(s)\big)\big(\tilde B(s)-\tilde B^\prime(s)\big)+\tilde B^\prime(s)^2\big(\mathrm{e}^{\int_t^sF^{\tilde B}(r)\mathrm{d}r}-\mathrm{e}^{\int_t^sF^{B^\prime}(r)\mathrm{d}r}\big)\Big)\mathrm{d}s.
\]
Together with \eqref{eq:stabF:int}, this yields
\begin{equation}\label{eq:stabE}
\big\|E^{\tilde B}-E^{B^\prime}\big\|_\infty\leq \frac{\hat\gamma\beta T}{\lambda}\bigg(1+ \frac{\hat \gamma T^2}{\lambda}\|\tilde B^\prime\|_\infty \big( \|\tilde B\|_\infty+\|\tilde B^\prime\|_\infty\big) \bigg)\big\|\tilde B-\tilde B^\prime\big\|_\infty,
\end{equation}
as well as
\begin{equation}\label{eq:stabC}
\big\|C^{\tilde B}-C^{B^\prime}\big\|_\infty \leq \frac{|\varepsilon| T}{2} \bigg(1+  \frac{\hat \gamma T^2}{\lambda}\|\tilde B^\prime\|_\infty^2 \bigg)\big( \|\tilde B\|_\infty+\|\tilde B^\prime\|_\infty\big)\big\|\tilde B-\tilde B^\prime\big\|_\infty.
\end{equation}

\emph{Step 2: A priori estimate for $\|\tilde B\|_\infty$}. Now, fix some $R> a$, define $\tilde B^0=a$ and, for a fixed integer $n\geq 1$, consider a continuous function $\tilde B^{n-1}$ with $||\tilde{B}^{n-1}||_\infty \leq R$. Let $(C^n,E^n,F^n)$ be the unique solution of the system \eqref{eq:EF} with $\tilde B:=\tilde B^{n-1}$. We then define $\tilde B^n$ as the unique solution of the following (linear) ODE (well-posedness is clear since $\tilde B^{n-1}$, $C^{n}$, $E^{n}$ and $F^n$ are all uniformly bounded):
\begin{equation}\label{eq:Bn}
\tilde B^n(t)=a -\int_t^T\big(\varepsilon\beta \tilde B^{n-1}(s)-E^n(s)C^{n}(s)\big)\mathrm{d}s,\quad t\in[0,T].
\end{equation}
Using the estimates on $E^n$ and $C^n$ from \eqref{eq:boundE}, we obtain
\begin{align*}
\big\|\tilde B^n\big\|_\infty \leq  a  + |\varepsilon| \beta T \big\|\tilde B^{n-1}\big\|_\infty + \frac{|\varepsilon| \hat \gamma \beta }{2 \lambda} T^3 \big\|\tilde B^{n-1}\big\|_\infty^3.
\end{align*}
Now, for $R = \frac{3}{2}a$ and $|\varepsilon|$ satisfying  \eqref{eq:thm:odeex} it follows that 
\[
a  + |\varepsilon| \beta R T + |\varepsilon| \beta \frac{\hat \gamma R^3 T^3}{2 \lambda}  \leq R.
\]
Then we have
\[
\big\|\tilde{B}^n\big\|_\infty \leq R.
\]

\emph{Step 3: Picard iteration for $\tilde B$}.
Finally, using the fact that
\begin{align*}
\tilde B^n(t)-\tilde B^{n\prime}(t)=&-\int_t^T \varepsilon \beta \big(\tilde B^{n-1}(s)-\tilde B^{n-1\prime}(s)\big)\mathrm{d}s\\
&-\int_t^TC^{n}(s)\big(E^n(s)-E^{n\prime}(s)\big)+E^{n\prime}(s)\big(C^n(s)-C^{n\prime}(s)\big)\mathrm{d}s,
\end{align*}
it follows from \eqref{eq:boundE}, \eqref{eq:stabE}, and \eqref{eq:stabC} that
\begin{align*}
\Vert \tilde B^n -\tilde B^{n\prime} \Vert_\infty 
&\leq T \left(|\varepsilon| \beta + 2 |\varepsilon| \beta \frac{\hat\gamma R^2 T^2 }{\lambda}\bigg(1+ \frac{\hat \gamma R^2 T^2}{\lambda} \bigg) \right)\Vert \tilde B^{n-1}-\tilde B^{n-1\prime} \Vert_\infty' 
\end{align*}
Now, for $R = \frac{3}{2}a$ and $|\varepsilon|$ satisfying  \eqref{eq:thm:odeex}, the constant is less than $1$ and  we have a contraction.
\end{proof}

\appendix
\section{BMO Results}\label{app:A}

This appendix collects some auxiliary results on BMO martingales that are used in the proofs of Theorem~\ref{thm:ex}, Lemma~\ref{lem:kt02}, and Proposition~\ref{prop:eq}.

\begin{lemma}
\label{lem:BMO P Q}
Let $(\Omega, \mathcal{F}, \mathbb{F} = (\mathcal{F}_t)_{t \in [0, T]}, \mathbb P)$ be a filtered probability space supporting a Brownian motion $(W_t)_{t \in [0, T]}$ and such that all $\mathbb{F}-$martingales are continuous. Let $(\alpha_t)_{t \in [0, T]} \in \mathbb{H}^2_{\mathrm{BMO}}(\mathbb{P})$ and define $\mathbb P^\alpha \sim \mathbb{P}$ on $\mathcal{F}_T$ by
\begin{equation*}
\frac{\mathrm{d}  \mathbb P^\alpha}{\mathrm{d}  \mathbb{P}} = \mathcal{E}\left(\int_0^\cdot \alpha_t \mathrm{d} W_t\right)_T.
\end{equation*}
Then $\alpha \in \mathbb{H}^2_{\mathrm{BMO}}(\mathbb P^\alpha)$. Moreover, if $(\sigma_t)_{t \in [0, T]} \in \mathbb{H}^2_{\mathrm{BMO}}(\mathbb{P})$, then $(\sigma_t)_{t \in [0, T]} \in \mathbb{H}^2_{\mathrm{BMO}}(\mathbb P^\alpha)$ and 
\begin{align*}
\|\sigma\|^2_{\mathbb{H}_{\mathrm{BMO}}(\mathbb P^\alpha)} &\leq 8 (\|\alpha\|_{\mathbb{H}^2_{\mathrm{BMO}}(\mathbb P^\alpha)} + 1)^2  \|\sigma\|^2_{\mathbb{H}^2_{\mathrm{BMO}}(\mathbb{P})}, \quad \|\sigma\|^2_{\mathbb{H}_{\mathrm{BMO}}(\mathbb{P})} \leq 8 (\|\alpha\|_{\mathbb{H}^2_{\mathrm{BMO}}(\mathbb{P})} + 1)^2  \|\sigma\|^2_{\mathbb{H}^2_{\mathrm{BMO}}(\mathbb P^\alpha)}.
\end{align*}
\end{lemma}

\begin{proof}
This follows immediately from the proof of \cite[Theorem 3.6]{kazamaki.94} and Lemma~\ref{lem:Ap estimate} applied under $\mathbb P^\alpha$ and $\mathbb{P}$.
\end{proof}

\begin{lemma}
\label{lem:Ap estimate}
Let $(\Omega, \mathcal{F}, \mathbb{F} = (\mathcal{F}_t)_{t \in [0, T]}, \mathbb P)$ be a probability space supporting a Brownian motion $(W_t)_{t \in [0, T]}$. Let $(\alpha_t)_{t \in [0, T]} \in \mathbb{H}^2_{\mathrm{BMO}}(\mathbb{P})$ and define the $\mathbb{P}$-martingale $(M_t)_{t \in [0, T]}$ by $M_t := \int_0^t \alpha_s \mathrm{d} W_s$. Then for any $p > 1$ with $p \geq (\|\alpha\|_{\mathbb{H}^2_{\mathrm{BMO}}(\mathbb{P})} + 1)^2$ and any stopping time $\tau$, we have
\begin{equation*}
\bigg\| \mathbb E_{\tau}\bigg[ \bigg(\frac{\mathcal{E}(M)_\tau}{\mathcal{E}(M)_T}\bigg)^{\frac{1}{p-1}}\bigg]\bigg\|^{\infty} \leq 2.
\end{equation*}

\end{lemma}

\begin{proof}
The condition on $p$ implies that $\|\alpha/{(\sqrt{2}(\sqrt{p} -1))}\|^2_{\mathbb{H}^2_{\mathrm{BMO}}(\mathbb{P})} \leq \frac{1}{2}$. Thus it follows from the John--Nirenberg inequality~\cite[Theorem 2.2]{kazamaki.94} that, for any stopping $\tau$,
\begin{equation*}
\mathbb E_{\tau}\bigg[\frac{1}{2 (\sqrt{p}-1)^2} \Big(\langle M \rangle_T -  \langle M \rangle_\tau \Big)\bigg] \leq 2.
\end{equation*}
The claim now follows from the proof of (a) $\Rightarrow$ (b) in~\cite[Theorem 2.4]{kazamaki.94} with $C_p = 2$.
\end{proof}

\begin{lemma}
\label{lem:increasing process}
Let $(\Omega, \mathcal{F}, \mathbb{F} = (\mathcal{F}_t)_{t \in [0, T]}, \mathbb P)$ be a probability space, $(\beta_t)_{t \in [0, T]}$ a nonnegative process and $(A_t)_{t \in [0, T]}$ a nondecreasing process. Then, for any $[0, T]$-valued stopping time $\tau$,
\begin{equation}
\label{eq:lem:increasing process:1}
\mathbb E_\tau\bigg[\int_\tau^T A_s \beta_s \mathrm{d} s \bigg] \leq A_\tau  \mathbb E_\tau \bigg[\int_\tau^T \beta_s \mathrm{d} s \bigg] + \mathbb E_\tau\bigg[\int_\tau^T  \mathbb E_u\bigg[\int_u^T \beta_s \mathrm{d} s \bigg] \mathrm{d} A_u \bigg]. 
\end{equation}
Moreover, if $\sqrt{\beta} \in \mathbb{H}^2_{\mathrm{BMO}}$, then
\begin{equation}
\label{eq:lem:increasing process:2}
\mathbb E_\tau\bigg[\int_\tau^T A_s \beta_s \mathrm{d}  s \bigg] \leq \big\| \sqrt{\beta} \big\|^2_{\mathbb{H}^2_{\mathrm{BMO}}} \mathbb E_\tau\big[A_T \big].
\end{equation}
\end{lemma}

\begin{proof}
Write $A_s = A_\tau + \int_t^s \mathrm{d} A_u$ for $s \geq \tau$. Fubini's theorem in turn gives
\begin{equation*}
 \int_\tau^T A_s \beta_s \mathrm{d}  s = A_\tau  \int_\tau^T \beta_s \mathrm{d}  s + \int_\tau^T  \int_\tau^s \beta_s  \mathrm{d}  A_u \mathrm{d} s = A_\tau  \int_\tau^T \beta_s \mathrm{d} s + \int_\tau^T \int_u^T \beta_s \mathrm{d} s \mathrm{d} A_u.
\end{equation*}
Now, \eqref{eq:lem:increasing process:1} follows from taking conditional expectations, using the conditional result corresponding to (the optional version of)~\cite[Theorem VI.57]{DelMeyB}, and that the optional projection of the process $(\int_u^T \beta_s ds)_{u \in [0, T]}$ is $(\mathbb E_u[\int_u^T \beta_s^2 ds])_{u \in [0, T]}$. Moreover, \eqref{eq:lem:increasing process:2} follows from \eqref{eq:lem:increasing process:1} by the definition of the BMO norm.
\end{proof}

\section{Variations on Doob's Inequality}\label{app:B}

The following versions of Doob's inequality are used in the proofs of Theorem~\ref{thm:ex} and Lemma~\ref{lem:stabilityc}, respectively. They easily follow using the inequalities of  H\"older and Doob.

\begin{lemma}
	\label{lem:Doob 1}
	Let $(\Omega, \mathcal{F}, \mathbb{F} = (\mathcal{F}_t)_{t \in [0, T]}, \mathbb P)$ be a filtered probability space, and $X$ an $\mathcal{F}_T$-measurable nonnegative random variable with $\mathbb E[X^2] < \infty$. Then
	\begin{equation*}
	\mathbb{E}\bigg[ \sup_{t \in [0, T]} \mathbb{E}_t\big[X\big]\bigg] \leq 2\mathbb{E}\big[X^2 \big]^{\frac{1}{2}} .
	\end{equation*}
\end{lemma}

\begin{lemma}
\label{lem:Doob 2}
Let $(\Omega, \mathcal{F}, \mathbb{F} = (\mathcal{F}_t)_{t \in [0, T]}, \mathbb P)$ be a filtered probability space, $p \in (1, 2)$, and let $X$ and $Y$ be $\mathcal{F}_T$--measurable non--negative random variables, with $\mathbb E[X^2] < \infty$ and $\mathbb E[Y^{2p/(2-p)}] < \infty$. Then, for $\varepsilon > 0$, 
\begin{equation*}
\mathbb{E}\bigg[ \sup_{t \in [0, T]} \mathbb{E}_t\big[X Y\big]\bigg] \leq \frac{1}{\varepsilon}  \mathbb{E}\big[X^2 \big] + \frac\varepsilon4 \bigg(\frac{p}{p-1}\bigg)^{2}  \mathbb{E}\Big[Y^{\frac{2p}{2-p}}\Big]^{\frac{2-p}{p}}.
\end{equation*}
\end{lemma}

{\small
\bibliographystyle{abbrv}
\bibliography{eqfric}}

\begin{thebibliography}{10}

\bibitem{adam.al.15}
K.~Adam, J.~Beutel, A.~Marcet, and S.~Merkel.
\newblock Can a financial transaction tax prevent stock price booms?
\newblock {\em J. Mon. Econ.}, 76:90--109, 2015.

\bibitem{almgren.chriss.01}
R.~Almgren and N.~Chriss.
\newblock Optimal execution of portfolio transactions.
\newblock {\em J. Risk}, 3(2):5--39, 2001.

\bibitem{amihud.mendelson.86a}
Y.~Amihud and H.~Mendelson.
\newblock Asset pricing and the bid-ask spread.
\newblock {\em J. Fin. Econ.}, 17(2):223--249, 1986.

\bibitem{ankirchner2019transformation}
S.~Ankirchner, A.~Fromm, and J.~Wendt.
\newblock A transformation method to study the solvability of fully coupled
  {FBSDE}s.
\newblock Preprint, available online at
  \url{https://hal.archives-ouvertes.fr/hal-02351469/}, 2019.

\bibitem{annkirchner.kruse.15}
S.~Ankirchner and T.~Kruse.
\newblock Optimal position targeting with stochastic linear--quadratic costs.
\newblock {\em Banach Center Publ.}, 104(1):9--24, 2015.

\bibitem{antonelli2006existence}
F.~Antonelli and S.~Hamad\`ene.
\newblock Existence of the solutions of backward--forward {SDE}'s with
  continuous monotone coefficients.
\newblock {\em Stat. \& Probab. Letters}, 76(14):1559--1569, 2006.

\bibitem{bank.al.17}
P.~Bank, H.~Soner, and M.~Vo{\ss}.
\newblock Hedging with temporary price impact.
\newblock {\em Math. Fin. Econ.}, 11(2):215--239, 2017.

\bibitem{bank.voss.18}
P.~Bank and M.~Vo{\ss}.
\newblock Linear quadratic stochastic control problems with stochastic terminal
  constraint.
\newblock {\em SIAM J. Control Optim.}, 56(2):672--699, 2018.

\bibitem{barrieu2008closedness}
P.~Barrieu, N.~Cazanave, and N.~El~Karoui.
\newblock Closedness results for {BMO} semi--martingales and application to
  quadratic {BSDE}s.
\newblock {\em C. R. Acad. Sci.}, 346(15--16):881--886, 2008.

\bibitem{bouchard.al.17}
B.~Bouchard, M.~Fukasawa, M.~Herdegen, and J.~Muhle-Karbe.
\newblock Equilibrium returns with transaction costs.
\newblock {\em Finance Stoch.}, 22(3):569--601, 2018.

\bibitem{brennan.subrahmanyam.96}
M.~J. Brennan and A.~Subrahmanyam.
\newblock Market microstructure and asset pricing: on the compensation for
  illiquidity in stock returns.
\newblock {\em J. Fin. Econ.}, 41(3):441--464, 1996.

\bibitem{briand2013simple}
P.~Briand and R.~\'Elie.
\newblock A simple constructive approach to quadratic {BSDE}s with or without
  delay.
\newblock {\em Stoch. Process. Appl.}, 123(8):2921--2939, 2013.

\bibitem{buss.dumas.17}
A.~Buss and B.~Dumas.
\newblock The dynamic properties of financial-market equilibrium with trading
  fees.
\newblock {\em J. Finance}, 74(2):795--844, 2019.

\bibitem{buss.al.16}
A.~Buss, B.~Dumas, R.~Uppal, and G.~Vilkov.
\newblock The intended and unintended consequences of financial-market
  regulations: a general-equilibrium analysis.
\newblock {\em J. Mon. Econ.}, 81:25--43, 2016.

\bibitem{cheridito.al.15}
P.~Cheridito, U.~Horst, M.~Kupper, and T.~A. Pirvu.
\newblock Equilibrium pricing in incomplete markets under translation invariant
  preferences.
\newblock {\em Math. Oper. Res.}, 41(1):174--195, 2015.

\bibitem{choi.al.18}
J.-H. Choi, K.~Larsen, and D.~Seppi.
\newblock Equilibrium effects of intraday order-splitting benchmarks.
\newblock To appear in \emph{Math. Fin. Econ.}, available online at
  \url{https://link.springer.com/article/10.1007/s11579-020-00278-7}.

\bibitem{constantinides.86}
G.~M. Constantinides.
\newblock Capital market equilibrium with transaction costs.
\newblock {\em J. Pol. Econ.}, 94(4):842, 1986.

\bibitem{danilova.julliard.19}
A.~Danilova and C.~Julliard.
\newblock Understanding volatility, liquidity, and the {T}obin tax.
\newblock Preprint, available online at
  \url{http://personal.lse.ac.uk/julliard/papers/Understanding_Vol.pdf}, 2019.

\bibitem{bouchaud.al.12}
J.~De~Lataillade, C.~Deremble, M.~Potters, and J.-P. Bouchaud.
\newblock Optimal trading with linear costs.
\newblock {\em RISK}, 1(3):2047--1246, 2012.

\bibitem{delbaen.al.15}
F.~Delbaen, Y.~Hu, and A.~Richou.
\newblock On the uniqueness of solutions to quadratic {BSDEs} with convex
  generators and unbounded terminal conditions: the critical case.
\newblock {\em Discrete Contin. Dynam. Systems A}, 35(11):5273--5283, 2015.

\bibitem{DelMeyB}
C.~Dellacherie and P.-A. Meyer.
\newblock {\em {Probabilities and Potential {B}}}.
\newblock North-Holland Publishing Co., Amsterdam, 1982.

\bibitem{ekeland.temam.99}
I.~Ekeland and R.~Temam.
\newblock {\em Convex analysis and variational problems}.
\newblock SIAM, Philadelphia, PA, 1999.

\bibitem{elie.al.18}
R.~\'Elie, L.~Moreau, and D.~Possama{\"i}.
\newblock On a class of path-dependent singular stochastic control problems.
\newblock {\em SIAM J. Control Optim.}, 56(5):3260--3295, 2018.

\bibitem{fromm2013existence}
A.~Fromm and P.~Imkeller.
\newblock Existence, uniqueness and regularity of decoupling fields to
  multidimensional fully coupled {FBSDE}s.
\newblock Preprint, alailable online at \url{https://arxiv.org/abs/1310.0499},
  2013.

\bibitem{garleanu.pedersen.13}
N.~Garleanu and L.~H. Pedersen.
\newblock Dynamic trading with predictable returns and transaction costs.
\newblock {\em J. Finance}, 68(6):2309--2340, 2013.

\bibitem{garleanu.pedersen.16}
N.~Garleanu and L.~H. Pedersen.
\newblock Dynamic portfolio choice with frictions.
\newblock {\em J. Econ. Theory}, 165:487--516, 2016.

\bibitem{gonon.al.19}
L.~Gonon, J.~Muhle-Karbe, and X.~Shi.
\newblock Asset pricing with general transaction costs: Theory and numerics.
\newblock Preprint, available online at \url{https://arxiv.org/abs/1905.05027},
  2019.

\bibitem{grossman.stiglitz.80}
S.~J. Grossman and J.~E. Stiglitz.
\newblock On the impossibility of informationally efficient markets.
\newblock {\em Am. Econ. Rev.}, 70(3):393--408, 1980.

\bibitem{guasoni.weber.15}
P.~Guasoni and M.~Weber.
\newblock Dynamic trading volume.
\newblock {\em Math. Finance}, 27(2):313--349, 2017.

\bibitem{harter2019stability}
J.~Harter and A.~Richou.
\newblock A stability approach for solving multidimensional quadratic {BSDE}s.
\newblock {\em Electron. J. Probab.}, 24(4):1--51, 2019.

\bibitem{hau.06}
H.~Hau.
\newblock The role of transaction costs for financial volatility: Evidence from
  the {Paris} bourse.
\newblock {\em J. Europ. Econ. Ass.}, 4(4):862--890, 2006.

\bibitem{heaton.lucas.96}
J.~Heaton and D.~J. Lucas.
\newblock Evaluating the effects of incomplete markets on risk sharing and
  asset pricing.
\newblock {\em J.Pol. Econ.}, 3(104):443--487, 1996.

\bibitem{isaenko.20}
S.~Isaenko.
\newblock Slow-moving capital and stock returns.
\newblock {\em Quant. Finance}, 20(6):969--984, 2015.

\bibitem{jones.seguin.97}
C.~M. Jones and P.~J. Seguin.
\newblock Transaction costs and price volatility: Evidence from commission
  deregulation.
\newblock {\em Am. Econ. Rev.}, 4(87):728--737, 1997.

\bibitem{kallsen.02}
J.~Kallsen.
\newblock Derivative pricing based on local utility maximization.
\newblock {\em Finance Stoch.}, 6(1):115--140, 2002.

\bibitem{kazamaki.94}
N.~Kazamaki.
\newblock {\em Continuous exponential martingales and {BMO}}.
\newblock Springer, Berlin, 1994.

\bibitem{kim.omberg.96}
T.~S. Kim and E.~Omberg.
\newblock Dynamic nonmyopic portfolio behavior.
\newblock {\em Rev. Fin. Stud.}, 9(1):141--161, 1996.

\bibitem{kohlmann.tang.02}
M.~Kohlmann and S.~Tang.
\newblock Global adapted solution of one-dimensional backward stochastic
  {Riccati} equations, with application to the mean--variance hedging.
\newblock {\em Stoch. Process. Appl.}, 97(2):255--288, 2002.

\bibitem{kramkov.15}
D.~Kramkov.
\newblock Existence of an endogenously complete equilibrium driven by a
  diffusion.
\newblock {\em Finance Stoch.}, 19(1):1--22, 2015.

\bibitem{kramkov.pulido.16}
D.~Kramkov and S.~Pulido.
\newblock A system of quadratic {BSDEs} arising in a price impact model.
\newblock {\em Ann. Appl. Probab.}, 26(2):794--817, 2016.

\bibitem{luo2019multidimensional}
M.~Kupper, P.~Luo, and L.~Tangpi.
\newblock Multidimensional {Markovian FBSDEs} with superquadratic growth.
\newblock {\em Stochastic Process. Appl.}, 129(3):902--923, 2019.

\bibitem{kyle.85}
A.~S. Kyle.
\newblock Continuous auctions and insider trading.
\newblock {\em Econometrica}, 53(6):1315--1335, 1985.

\bibitem{lo.al.04}
A.~W. Lo, H.~Mamaysky, and J.~Wang.
\newblock Asset prices and trading volume under fixed transaction costs.
\newblock {\em J. Pol. Econ.}, 112(5):1054--1090, 2004.

\bibitem{lo.wang.00}
A.~W. Lo and J.~Wang.
\newblock Traing volume: definitions, data analysis, and implications of
  portfolio theory.
\newblock {\em Rev. Fin. Studies}, 13(2):257--300, 2000.

\bibitem{luo2017solvability}
P.~Luo and L.~Tangpi.
\newblock Solvability of coupled {FBSDE}s with diagonally quadratic generators.
\newblock {\em Stochastics Dyn.}, 17(06):1750043, 2017.

\bibitem{lynch.tan.11}
A.~W. Lynch and S.~Tan.
\newblock Explaining the magnitude of liquidity premia: The roles of return
  predictability, wealth shocks, and state-dependent transaction costs.
\newblock {\em J. Finance}, 66(4):1329--1368, 2011.

\bibitem{ma2015well}
J.~Ma, Z.~Wu, D.~Zhang, and J.~Zhang.
\newblock On well-posedness of forward--backward {SDE}s -- a unified approach.
\newblock {\em Ann. Appl. Probab.}, 25(4):2168--2214, 2015.

\bibitem{martin.12}
R.~Martin.
\newblock Optimal trading under proportional transaction costs.
\newblock {\em RISK}, August:54--59, 2014.

\bibitem{martin.schoeneborn.11}
R.~Martin and T.~Sch{\"o}neborn.
\newblock Mean reversion pays, but costs.
\newblock {\em RISK}, February:96--101, 2011.

\bibitem{moreau.al.15}
L.~Moreau, J.~Muhle-Karbe, and H.~M. Soner.
\newblock Trading with small price impact.
\newblock {\em Math. Finance}, 27(2):350--400, 2017.

\bibitem{muhlekarbe.al.20}
J.~Muhle-Karbe, M.~Nutz, and X.~Tan.
\newblock Asset pricing with heterogenous beliefs and illiquidity.
\newblock {\em Math. Finance}, 30(4):1392--1421, 2020.

\bibitem{pardoux1998backward}
E.~Pardoux.
\newblock Backward stochastic differential equations and viscosity solutions of
  systems of semilinear parabolic and elliptic {PDE}s of second order.
\newblock In L.~Decreusefond, B.~{\O}ksendal, J.~Gjerde, and {\"U}.~A.S.,
  editors, {\em Stochastic analysis and related topics VI. Proceedings of the
  sixth Oslo--Silivri workshop, Geilo, 1996}, volume~42 of {\em Progress in
  probability}, pages 79--127, 1998.

\bibitem{pastor.stambaugh.03}
L.~P{\'a}stor and R.~F. Stambaugh.
\newblock Liquidity risk and expected stock returns.
\newblock {\em J. Pol. Econ.}, 111(3):642--685, 2003.

\bibitem{sannikov.skrzypacz.16}
Y.~Sannikov and A.~Skrzypacz.
\newblock Dynamic trading: price inertia and front-running.
\newblock Preprint, available online at
  \url{https://web.stanford.edu/~skrz/Dynamic_Trading.pdf}, 2016.

\bibitem{tevzadze.08}
R.~Tevzadze.
\newblock Solvability of backward stochastic differential equations with
  quadratic growth.
\newblock {\em Stoch. Process. Appl.}, 118(3):503--515, 2008.

\bibitem{touzi.13}
N.~Touzi.
\newblock {\em Optimal Stochastic Control, Stochastic Target Problems, and
  Backward SDE}.
\newblock Springer, New York, 2013.

\bibitem{umlauf.93}
S.~R. Umlauf.
\newblock Transaction taxes and the behavior of the {S}wedish stock market.
\newblock {\em J. Fin. Econ.}, 2(33):227--240, 1993.

\bibitem{vayanos.98}
D.~Vayanos.
\newblock Transaction costs and asset prices: A dynamic equilibrium model.
\newblock {\em Rev. Fin. Stud.}, 11(1):1--58, 1998.

\bibitem{vayanos.vila.99}
D.~Vayanos and J.-L. Vila.
\newblock Equilibrium interest rate and liquidity premium with transaction
  costs.
\newblock {\em Econ. Theory}, 13(3):509--539, 1999.

\bibitem{weston.17}
K.~Weston.
\newblock Existence of a {R}adner equilibrium in a model with transaction
  costs.
\newblock {\em Math. Fin. Econ.}, 12(4):517--539, 2018.

\bibitem{xing2016class}
H.~Xing and G.~{\v{Z}}itkovi{\'c}.
\newblock A class of globally solvable {M}arkovian quadratic {BSDE} systems and
  applications.
\newblock {\em Ann. Probab.}, 46(1):491--550, 2018.

\end{thebibliography}

\end{document}